%% file: paper.tex
\documentclass[acmsmall]{acmart}
\input{macros}

\makeatletter\if@ACM@journal\makeatother
\acmJournal{PACMPL}
\acmVolume{1}
\acmNumber{1}
\acmArticle{1}
\acmYear{2017}
\acmMonth{1}
\acmDOI{10.1145/nnnnnnn.nnnnnnn}
\startPage{1}
\else\makeatother
\acmConference[PL'17]{ACM SIGPLAN Conference on Programming Languages}{January 01--03, 2017}{New York, NY, USA}
\acmYear{2017}
\acmISBN{978-x-xxxx-xxxx-x/YY/MM}
\acmDOI{10.1145/nnnnnnn.nnnnnnn}
\startPage{1}
\fi

\setcopyright{none}             %% For review submission
\begin{document}

%% Title information
\title[]{Trace Abstraction Modulo Probability}         %% [Short Title] is optional;
                                        %% when present, will be used in
                                        %% header instead of Full Title.
%\titlenote{with title note}             %% \titlenote is optional;
                                        %% can be repeated if necessary;
                                        %% contents suppressed with 'anonymous'
%\subtitle{A  Type-Directed Approach}                     %% \subtitle is optional
%\subtitlenote{with subtitle note}       %% \subtitlenote is optional;
                                        %% can be repeated if necessary;
                                        %% contents suppressed with 'anonymous'

%% Author information
%% Contents and number of authors suppressed with 'anonymous'.
%% Each author should be introduced by \author, followed by
%% \authornote (optional), \orcid (optional), \affiliation, and
%% \email.
%% An author may have multiple affiliations and/or emails; repeat the
%% appropriate command.
%% Many elements are not rendered, but should be provided for metadata
%% extraction tools.

%% Author with single affiliation.
\author{Calvin Smith}
\affiliation{%
  \institution{University of Wisconsin--Madison}            %% \institution is required
  \department{Computer Sciences Department}
  \streetaddress{1210 West Dayton St.}
  \city{Madison}
  \state{WI}
  \postcode{53706}
  \country{USA}}
\email{cjsmith@cs.wisc.edu}
\author{Justin Hsu}
\affiliation{%
  \institution{University of Wisconsin--Madison}            %% \institution is required
  \department{Computer Sciences Department}
  \streetaddress{1210 West Dayton St.}
  \city{Madison}
  \state{WI}
  \postcode{53706}
  \country{USA}}
\email{email@justinh.su}
\author{Aws Albarghouthi}
\affiliation{%
  \institution{University of Wisconsin--Madison}            %% \institution is required
  \department{Computer Sciences Department}
  \streetaddress{1210 West Dayton St.}
  \city{Madison}
  \state{WI}
  \postcode{53706}
  \country{USA}}
\email{aws@cs.wisc.edu}

%\author{Aws Albarghouthi}

%\authornote{with author1 note}          %% \authornote is optional;
                                        %% can be repeated if necessary

%\email{first1.last1@inst1.edu}          %% \email is recommended

%% Author with two affiliations and emails.
%\authornote{with author2 note}          %% \authornote is optional;
                                        %% can be repeated if necessary
%\orcid{nnnn-nnnn-nnnn-nnnn}             %% \orcid is optional

%% Paper note
%% The \thanks command may be used to create a "paper note" ---
%% similar to a title note or an author note, but not explicitly
%% associated with a particular element.  It will appear immediately
%% above the permission/copyright statement.
%\thanks{with paper note}                %% \thanks is optional
                                        %% can be repeated if necesary
                                        %% contents suppressed with 'anonymous'

%% Abstract
%% Note: \begin{abstract}...\end{abstract} environment must come
%% before \maketitle command
\begin{abstract}
  We propose \emph{trace abstraction modulo probability}, a proof technique for
  verifying high-probability accuracy guarantees of probabilistic programs.
  Our proofs overapproximate the set of program traces using \emph{failure automata}, finite-state automata that upper bound the probability of failing to satisfy a target specification.
  % Automata maintain two kinds of information: regular assertions over program states and a number bounding the probability that the assertion does not hold
  % Automata represent probabilistic invariants as two separate pieces: a
  % standard, non-probabilistic assertion about program states, and a numerical expression bounding the probability that the assertion does not hold.
  %
  % The non-probabilistic assertions can be manipulated by typical program
  % verification methods, while the failure probability is tracked off to the side
  % by leveraging the \emph{union bound} technique from basic probability theory.
%
  We automate proof construction by reducing probabilistic reasoning to logical
  reasoning: we use
  program synthesis methods to
  select axioms for sampling instructions, and then apply Craig
  interpolation to prove that traces fail the target
  specification with only a small probability.
  %
  % Unlike typical techniques in probabilistic
  % verification, our approach relies purely on logical methods instead of
  % numerical methods.
  %
  Our method handles programs with
  unknown inputs, parameterized distributions,  infinite state spaces,
  and parameterized  specifications.
  We evaluate our technique
  on a range of randomized algorithms drawn from the differential privacy
  literature and beyond.
  To our knowledge, our approach is the first to automatically establish
  accuracy properties of these algorithms.

\end{abstract}

\maketitle

\input{introduction}
%
\input{example}
\input{problem}
%
\input{annotations}
%
\input{algorithm}
%
\input{interpolants}
%
\input{evaluation}
%
\input{relatedwork}
\input{conclusion}

\begin{acks}
We thank Thomas Reps, Zachary Kincaid, and the anonymous referees for their comments on earlier drafts of this work.
This work is supported by the National Science Foundation CCF under Grant Nos. 1566015, 1704117, 1652140, and 1637532.
\end{acks}

%% Bibliography
\bibliography{header,references}

%% Appendix
\pagebreak
\appendix

\input{appendix}
\input{app_implementation}

\end{document}

%% file: macros.tex
\usepackage{rotating}
\usepackage{caption}
\usepackage{subcaption}
\usepackage{prooftree}
\usepackage{soul}
\usepackage{tikz}
\usepackage{siunitx}
\usepackage{listings}
\usepackage{pifont}
\usepackage{multirow}
\usepackage{amsmath}
\usepackage{amsfonts}
\usepackage{amssymb}
\usepackage{graphicx}
\usepackage{paralist}
\usepackage{wrapfig}
\usepackage{balance}
\usepackage{dsfont}
\usepackage{enumitem}
\usepackage{mathtools}
\usepackage{listings}
\usepackage[plain]{algorithm}
\usepackage[noend]{algpseudocode}
\usepackage{blindtext}
\usepackage{enumitem}
\usepackage{comment}
\usepackage{url}
\usepackage{etoolbox}
\usepackage{pbox}
\usepackage{csquotes}
\usepackage{ctable} % for \specialrule command
\usepackage{textcase}
\usepackage{wasysym}
\usepackage{mdframed}
\usepackage{amsthm}
\usepackage{stmaryrd}

\newcommand{\ifrac}[2]{#1/#2}

\newcommand\myhrulefill[1]{\leavevmode\leaders\hrule height#1\hfill\kern0pt}
\DeclareCaptionFormat{myformat}{{\myhrulefill{0.06em}}\\#1#2#3}
\renewcommand{\paragraph}[1]{\vspace{.07in}\noindent\textbf{{#1}~~~}}
\renewcommand{\qedsymbol}{{\footnotesize{$\blacksquare$}}}
\AtEndEnvironment{example}{\hfill\qedsymbol}%
\AtEndEnvironment{definition}{\hfill\qedsymbol}%

\usepackage{enumitem}% http://ctan.org/pkg/enumitem
\makeatletter
\lst@InstallKeywords k{attributes}{attributestyle}\slshape{attributestyle}{}ld
\makeatother

\usepackage[scaled=0.80]{DejaVuSansMono}
\lstdefinestyle{customc}{
  belowcaptionskip=1\baselineskip,
  breaklines=true,
  xleftmargin=\parindent,
  language=C,
  showstringspaces=false,
  basicstyle=\ttfamily,
  keywordstyle=\bfseries\ttfamily\color{keywordcolor},
  commentstyle=\itshape\color{black},
  identifierstyle=\ttfamily\color{black},
  stringstyle=\itshape\color{NavyBlue},
  keywords={ map, flatMap, reduce, then, in, if, else, reduceByKey},
moreattributes={let, where}, % etc...
attributestyle = \bfseries\ttfamily\color{attributecolor}
}

\newcommand{\aws}[1]{{\color{brown}{\noindent\textsf{\textbf{Aws}---#1}}}}
 \newcommand{\jh}[1]{{\color{blue}{\noindent\textsf{\textbf{Justin}---#1}}}}
 \newcommand{\cs}[1]{{\color{magenta}{\noindent\textsf{\textbf{Calvin}---#1}}}}
 \definecolor{darkgray}{gray}{0.4}
 \newcommand{\note}[1]{{\color{darkgray}{\noindent\textsf{\textbf{Note}---#1}}}}
 \definecolor{mypink3}{cmyk}{0, 0.7808, 0.4429, 0.1412}
\newcommand{\maybe}[1]{{\color{mypink3} \textsf{\textbf{Remove?}---#1}}}
\newcommand{\todo}[1]{{\color{red} \textsf{\textbf{To do}---#1}}}

 \renewcommand{\aws}[1]{}
 \renewcommand{\jh}[1]{}
 \renewcommand{\cs}[1]{}
 \renewcommand{\note}[1]{}
 \renewcommand{\maybe}[1]{}
 \renewcommand{\todo}[1]{}

\everypar{\looseness=-1}

\definecolor{lightgray}{gray}{0.9}
\definecolor{midgray}{gray}{0.65}
\definecolor{darkgray}{gray}{0.4}

\newcommand{\abr}[1]{\textsc{\MakeLowercase{#1}}}

\newcommand{\abrs}[1]{\abr{#1}{\footnotesize{s}}\xspace}

\renewcommand{\leq}{\leqslant}
\renewcommand{\geq}{\geqslant}

\newcommand{\rone}{(\emph{i})~}
\newcommand{\rtwo}{(\emph{ii})~}
\newcommand{\rthree}{(\emph{iii})~}
\newcommand{\rfour}{(\emph{iv})~}

\definecolor{keywordcolor}{gray}{0.0}
\definecolor{attributecolor}{gray}{0.0}
\definecolor{wildcolor}{gray}{0.8}

\usepackage{relsize}

\newcommand\earr\hookrightarrow
\newcommand{\wild}{{\scriptsize \color{black} \CIRCLE}}

\let\oldwild\wild
\renewcommand{\wild}{%
  \mathchoice{\raisebox{.4pt}{$\displaystyle\oldwild$}}
             {\raisebox{.5pt}{$\oldwild$}}
             {\raisebox{0.5pt}{$\scriptstyle\oldwild$}}
             {\raisebox{0.2pt}{$\scriptscriptstyle\oldwild$}}}

\newcommand{\aset}{C}
\newcommand{\aelem}{c}

\newcommand{\dist}{\mathit{dist}}
\newcommand{\sdist}{\dist}
\newcommand{\support}{\mathit{supp}}

\newcommand{\pr}{\mathrm{Pr}}

\newcommand{\model}{M}

\newcommand{\prog}{P}
\newcommand{\var}{v}
\newcommand{\vars}{V}
\newcommand{\varsd}{V^\mathsf{det}}

\newcommand{\varsi}{\vars^\mathsf{in}}

\newcommand{\wpprob}{\mathit{wp}^\mathsf{f}}

\newcommand{\locs}{L}

\newcommand{\stmt}{\mathit{s\!t}}
\newcommand{\stmts}{\Sigma}

\newcommand{\expr}{\mathit{e}}
\newcommand{\bexpr}{\mathit{b}}
\newcommand{\dexpr}{\mathit{d}}

\newcommand{\bern}{\mathsf{bern}}
\newcommand{\lap}{\mathsf{lap}}
\newcommand{\olap}{\mathsf{exp}}
\newcommand{\unif}{\mathsf{unif}}

\newcommand{\assume}{\mathsf{assume}}

\newcommand{\sem}[1]{\llbracket #1\rrbracket}

\newcommand{\dom}{D}

\newcommand{\trace}{\tau}

\newcommand{\cost}{\omega}
\newcommand{\halt}{h}

\newcommand{\fsm}{A}
\newcommand{\fsms}{\mathcal{A}}
\newcommand{\lang}{\mathcal{L}}
\newcommand{\tuple}[1]{\langle #1 \rangle}
\newcommand{\entry}{\mathsf{in}}
\newcommand{\exit}{\mathsf{ac}}
\newcommand{\loc}{\ell}
\newcommand{\lentry}{\ell^\entry}
\newcommand{\lexit}{\ell^\exit}
\newcommand{\lto}[1]{\xrightarrow{#1}}

\newcommand{\stts}{S}
\newcommand{\stt}{s}
\newcommand{\store}{\stt}

\newcommand{\labelp}{\lambda}
\newcommand{\labele}{\kappa}
\newcommand{\labeltrace}{\mathsf{label}}

\newcommand{\qentry}{q^\entry}
\newcommand{\qexit}{q^\exit}

\newcommand{\eps}\epsilon

\newcommand{\cinit}{\textsc{init}}
\newcommand{\ctrace}{\textsc{trace}}
\newcommand{\cgen}{\textsc{generalize}}
\newcommand{\cmerge}{\textsc{merge}}
\newcommand{\ccorr}{\textsc{correct}}

\newcommand{\mergeop}{\merge}

\algrenewcommand\algorithmicfunction{\textsf{\textbf{fun}}}

\newcommand{\dunit}{\mathit{unit}}
\newcommand{\dbind}{\mathit{bind}}

\newcommand{\pre}{\varphi_\mathsf{pre}}
\newcommand{\post}{\varphi_\mathsf{post}}
\newcommand{\errp}{\beta}
\newcommand{\hoare}[4]{\vdash_{#1} \{#2\}~ #3 ~\{#4\}}

\newcommand{\true}{\mathit{true}}
\newcommand{\false}{\mathit{false}}

\renewcommand{\epsilon}{\varepsilon}

\newcommand{\encode}{\mathit{enc}}
\newcommand{\encodeD}{\mathit{enc}}

\newcommand{\fv}{\mathit{fv}}

\newcommand{\axasm}{\varphi^\mathsf{ax}}
\newcommand{\axub}{\expr^\mathsf{ub}}
\usepackage[capitalize]{cleveref}

\crefname{section}{\S}{\S}
\Crefname{section}{\S}{\S}
\newcommand{\Skip}{\textbf{skip}}
\newcommand{\Assm}[1]{\assume(#1)}
\newcommand{\Assg}[2]{#1 \gets #2}
\newcommand{\Rand}[2]{#1 \sim #2}
\newcommand{\Seqn}[2]{#1 \mathbin{;} #2}
\newcommand{\Cond}[3]{\textbf{if } #1 \textbf{ then } #2 \textbf{ else } #3}
\newcommand{\Whil}[2]{\textbf{while } #1 \textbf{ do } #2}

\usepackage{mathpartir}
\ifdefined\isconf
  \newcommand{\smref}[2]{the extended version~\citep{#1}}
\else
  \newcommand{\smref}[2]{\cref{#2} in the supplementary materials}
\fi

%% file: introduction.tex
%!TEX root=paper.tex

\section{Introduction}
\label{sec:introduction}

% \aws{We might be asked: "why is this an important target for formal
% verification". We might want to say something about that. E.g., automation
% allows algorithm designers to automatically prove establish accuracy guarantees
% without manual reasoning or experimentation...}

With the recent explosion of interest in data
analysis, randomized algorithms are increasingly seeing applications across all
of computer science. These algorithms satisfy a wide variety of subtle
probabilistic properties: modeling statistical privacy of database queries~\citep{dwork2014algorithmic},
stability and generalization of machine learning procedures~\citep{bousquet2002stability},
fairness of decision-making algorithms~\citep{dwork2012fairness},
and more.

Many probabilistic properties are tailored to specific applications, but perhaps
the most fundamental properties are \emph{accuracy} gurantees, the probabilistic
analogue of functional correctness.
While such specifications would ideally hold all of the time---with probability
$1$---such stringent guarantees rule out many useful applications of of
randomization.
Accuracy properties are often phrased
as \emph{high-probability guarantees}: for all program inputs, an output
sampled from the final distribution satisfies $\varphi$ except with
some probability $\beta$. For instance, a noisy
numeric output $r_{\emph{noisy}}$ might satisfy a precision bound $\varphi \triangleq
|r_{\emph{noisy}} - r_\emph{exact}| < 5.2$ except with probability at most
$\beta \triangleq 0.01$, where $r_\emph{exact}$ is the answer without
noise. While the \emph{failure probability} $\beta$ can be a
concrete value in $[0, 1]$, it is often treated symbolically
so that the guarantee $\varphi$ may depend on $\beta$---e.g.,
$\varphi(\beta) \triangleq |r_\emph{noisy}-r_\emph{exact}| < 1/\beta$
gives higher confidence guarantees by widening the error range. This kind of
property describes how well an algorithm will perform at varying levels of
confidence, crucial information for the algorithm designer.
Our goal is to enable algorithm designers to prove accuracy
specifications fully automatically.

While simple to state, accuracy guarantees---like other probabilistic
properties---pose interesting challenges for automated verification.
Current techniques have focused on more tractable models of randomized
computation, especially probabilistic automata and Markov Decision
Processes. (Recent surveys by \citet{DBLP:reference/mc/BaierAFK18} and
\citet{Katoen:2016:PMC:2933575.2934574} provide a good overview.)
By treating parameters and inputs as known constants, tools
can apply numerical methods to compute event
probabilities in the output distribution. There are now several mature
verification tools (e.g.,
\citep{kwiatkowska2011prism,DBLP:journals/corr/DehnertJK017}), which have found
notable success in helping designers automatically analyze complex probabilistic
systems rigorously. However, their common foundation leads to common weaknesses:
they are mostly restricted to closed programs with fixed inputs and
finite state spaces, and support for properties with symbolic parameters remains
limited.

In this paper, we start from established automated verification techniques for
non-probabilistic programs and extend them to the probabilistic setting.
Our logic-based approach yields several benefits.
By reasoning symbolically instead of numerically, we can
\rone directly establish properties for all inputs rather than requiring fixed inputs,
\rtwo handle programs that sample from distributions with unknown parameters,
possibly over infinite ranges, and
\rthree prove parametric accuracy properties, making it possible to
automatically establish tradeoffs between accuracy and failure probabilities,
and capture the dependence on other input parameters.

\subsection{An Overview of Our Approach}

\paragraph{Trace Abstraction Modulo Probability.}
Our approach is based on \emph{trace
abstraction}~\citep{heizmann2009refinement,Heizmann10,heizmann2013software,Farzan:2013},
a proof technique for non-probabilistic verification. A program $\prog$ is represented by a language
$\lang(\prog)$ of syntactic execution \emph{traces} through the control-flow
graph. To prove that $\prog$ satisfies $\varphi$, we
first overapproximate $\lang(P)$ with finite-state
automata: % $\fsm_1,\ldots,\fsm_n$:
\[
\lang(P) \subseteq \lang(\fsm_1) \cup \cdots \cup \lang(\fsm_n)
\]
To show that all traces in $\cup_i \lang(\fsm_i)$ satisfy $\varphi$, the automata
are annotated with Hoare-style assertions overapproximating reachable states along
execution traces---annotations are usually computed via \emph{predicate
abstraction}~\citep{graf1997construction} or \emph{Craig interpolation}~\citep{mcmillan2006lazy}.
Intuitively, trace abstraction constructs proofs for
subsets of traces; the combined proofs verify the
whole program.

To extend this idea to probabilistic programs, suppose we want to prove that
$\varphi$ holds except with probabiltity at most $\beta$. We construct a set of
automata that \rone overapproximates all traces in $\lang(\prog)$ and \rtwo
satisfies the following probabilistic bound:
\[
  \sum_{\trace \in \cup_i\lang(\fsm_i)} \Pr[ \varphi \text{ not true after running } \trace ] \leq \beta
\]
In words, the total failure probability---across all traces represented by
the automata---is at most $\beta$.

While trace inclusion is relatively simple to check, verifying the probabilistic
bound over infinitely many traces is more challenging. To ease the task, our
automata are annotated with Hoare-style assertions describing reachable program
states, and also \emph{failure probabilities}, upper bounds on the probability
of \emph{not} reaching those program states.

\paragraph{Automating Proofs.}
To construct trace abstractions automatically,
our approach proves that individual program traces $\trace$ satisfy $\varphi$ with a failure probability $\beta$, and then generalizes the proof into a  (potentially infinite) set of traces represented by an automaton $\fsm_\trace$ such that the sum of failure probabilities of traces in $\fsm_\trace$ is at most $\beta$.
We repeatedly pick, prove, and generalize program traces until we have overapproximated the set of all possible traces $\lang(\prog)$.

The most technically intricate piece of our algorithm is the proving step. Given
a trace $\trace$, we want to
automatically prove that it satisfies $\varphi$ with failure probability $\beta$.
In the non-probabilistic setting, trace semantics can be encoded as a
logical formula and verification conditions can be discharged with an \abr{SMT}
solver. In our setting, however, traces have probabilistic semantics and a
na\"ive encoding would be prohibitively complex.
Instead, we reduce probabilistic reasoning to logical reasoning. Specifically,
we encode the verification condition as a
\emph{constraint-based synthesis problem} of the form $\exists f \ldotp \forall
X \ldotp \varphi$, where the function $f$ chooses between different
\emph{axiomatizations} of probability distributions that are sampled along the
given trace.
This choice affects the failure
probability of the entire computation, and also determines what we can assume
about the result of the random sampling statements later on in the trace. The
axioms can be
seen as approximating the semantics of probabilistic samplings using
a first-order theory amenable for checking by \abr{SMT}.

After proving
correctness of a trace $\trace$ by solving the formula
$\exists f \ldotp \forall X \ldotp \varphi$ for $f$, we demonstrate how to
use \emph{Craig interpolation}---a well-studied proof technique in
traditional, non-probabilistic verification---to construct a Hoare-style proof
of the trace along with failure probabilities. This proof can then be
generalized to cover a potentially infinite set of program traces.

\paragraph{Implementation \& Case Studies.}
We have implemented our algorithm and applied it to a range of sophisticated
randomized algorithms, mainly from the \emph{differential privacy}
literature~\citep{DMNS06}. In differential privacy, algorithm designers must add
noise to protect personal information, but try to guarantee good accuracy for
their data analyses.  Our technique automatically proves intricate accuracy specifications
capturing this tradeoff.
To demonstrate our approach's generality, we apply our technique
to reason about reliability of programs running on approximate hardware with probabilistic failures.

\paragraph{Design Principle: Reduce Probabilistic Reasoning.}
Conceptually, our approach divides assertions about state distributions into two
pieces: a standard, non-probabilistic predicate $\varphi$ on states, and a
single number $\beta$ bounding the probability that $\varphi$ fails to hold.
As a result, much of the reasoning deals with logical state predicates instead
of probabilistic assertions; the failure probability can be cleanly
tracked off to the side. This idea is inspired by the recent probabilistic Hoare logic
\abr{aHL}~\citep{DBLP:conf/icalp/BartheGGHS16}.

Reducing probabilistic reasoning to non-probabilistic reasoning is a useful
methodology to tame the complexity of deductive
verification~\citep{DBLP:conf/icalp/BartheGGHS16,barthe12,barthe14,JHThesis}. Avoiding probabilistic assertions means that our technique cannot take advantage of more precise analyses, such as concentration bounds based on independence of random variables, but it enables the application of classical verification techniques---trace abstraction, synthesis, and interpolation---yielding a high degree of automation. We view this as well worth the cost. Furthermore, our proof technique compares favorably to standard trace abstraction \citep{heizmann2009refinement, heizmann2013software} and \abr{aHL} \citep{DBLP:conf/icalp/BartheGGHS16} as discuss in \cref{sec:theory}.

\subsection{Outline and Contributions}

After demonstrating our verification strategy on two worked examples
(\cref{sec:overview}) and introducing the program model (\cref{sec:problem}), we
offer the following technical contributions.
\begin{itemize}
  \item \textbf{Trace abstraction modulo probability (\cref{sec:annotations}):}
  We present a proof rule for probabilistic accuracy properties, extending
  trace abstraction to the probabilistic setting. Our proof technique is
  based on the new notion of \emph{failure automata}, labeled automata
  overapproximating program traces and their probability of failing to satisfy a
  given postcondition.

  \item \textbf{Automating trace-based proofs (\cref{sec:algorithm}):}
  We present an algorithm for constructing trace-based proofs by iteratively
  building failure automata, proving correctness of finite program traces,
  and then generalizing the automata to cover potentially infinite sets of
  traces.

  \item \textbf{Proofs \& interpolation for probabilistic traces (\cref{sec:interpolation}):}
  We show that we can prove correctness of individual probabilistic program
  traces via a reduction to a \emph{constraint-based synthesis} problem, making
  probabilistic reasoning unnecessary. Then, we demonstrate how to apply
  \emph{Craig interpolation} to construct failure automata.

  \item \textbf{Implementation \& case studies (\cref{sec:evaluation}):}
   We implement our approach and use it to automatically prove accuracy
   guarantees of a range of randomized algorithms from the theory of differential privacy.
   We also study an example from the approximate computing literature. Our
   implementation establishes accuracy properties with symbolic parameters for
   programs with parametric inputs and infinite states, a first for automated
   verification.
\end{itemize}
Finally, we survey related work (\cref{sec:relatedwork}) and conclude
(\cref{sec:conclusions}).

%% file: example.tex
%!TEX root=paper.tex

\section{Overview and Illustration}\label{sec:overview}

In this section, we provide an overview of our proof technique on two
simple examples.

\begin{figure}[t]
\includegraphics[scale=1]{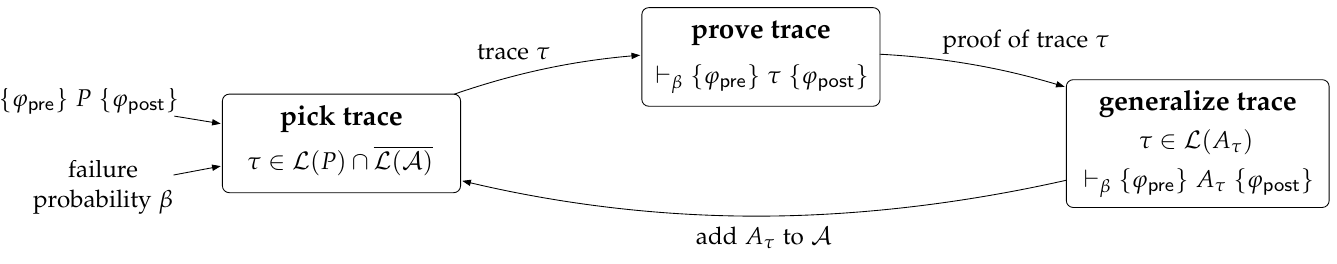}
% The loop terminates with a proof when \rone there are no more traces to pick and \rtwo the combined failure probabilities of $\fsms$ is at most $\beta$.
\caption{Main loop of verification algorithm}
\label{fig:overview}
\end{figure}

\subsection{High-Level Overview}
Suppose we are given a probabilistic program $\prog$, pre- and
post-conditions $\pre$ and $\post$, and a numeric expression $\errp$
representing the maximum allowed failure probability.
Our goal is to prove that if we start executing $\prog$ from any state satisfying
$\pre$, the probability that the output state does not satisfy $\post$ upon termination is at most $\errp$.
This property is denoted by the following formula, reminiscent of a Hoare triple:
\[
  \hoare{\errp}{\pre}{\prog}{\post}
\]

\paragraph{Proof Rule.}
We view $\prog$ as a \emph{control-flow automaton} whose language $\lang(\prog)$
is the set of all \emph{traces} from the program's entry location to its exit
location. Our proof rule \emph{overapproximates} $\lang(\prog)$ by a larger set
of traces, represented by a set of finite automata $\fsms$, while ensuring that
the total failure probability across all traces in $\lang(\fsms)$ is at most
$\errp$.

\paragraph{Automation.}
To apply our proof rule automatically, we apply an algorithmic technique
summarized in \cref{fig:overview}. The technique repeatedly tries to \rone pick
a program trace $\trace \in \lang(\prog)$ outside the approximation
$\lang(\fsms)$, \rtwo prove that $\hoare{\errp}{\pre}{\trace}{\post}$, i.e., the
probability that the trace falsifies the Hoare triple is at most $\errp$, and
then \rthree generalize the trace $\trace$ into an automaton $\fsm_\trace$
encoding a set of traces with total failure probability at most $\errp$. Our
approach succeeds if it constructs a set of automata $\fsms$ modeling all
program traces $\lang(\prog)$, with total failure probability at most $\errp$.

\subsection{Illustrative Example: Loop-free Program}\label{sec:ex1}
To warm up, we consider
the loop-free program in \cref{fig:ex1_prog}.
The function $\mathsf{ex1}$ takes a single [0,1]-valued input $p$
and returns a Boolean value $y$.
Our goal is to prove the following accuracy property:
\[
\hoare{p}{\true}{\mathsf{ex1}(p)}{\neg y}
\]
In words, the program fails to return $y = \false$ with probability at most $p$.
This property can be established informally:
\rone the probability that the program takes the \emph{then} branch and returns $y = \true$ is $0.5p$;
\rtwo the probability that it takes the \emph{else} branch and returns $y = \true$ is $0.25p$.
Therefore, the failure probability is $0.5p + 0.25p \leq p$.

\paragraph{Illustrating Proof Artifacts.}
We begin by describing the proof artifacts constructed by our approach.  The
program $\mathsf{ex1}$ is presented as a \emph{control-flow automaton} over the
alphabet of program statements, as shown in the left side of
\cref{fig:ex1}.  Edge labels of the form $[c]$ are guards (also known as
\emph{assume} statements) encoding possible branches of the conditional statement.
Accepted traces start from the initial node $\texttt{in}$ and end in the final, accepting node $\texttt{ac}$.

\begin{wrapfigure}{r}{0.3\textwidth}
  \centering
  %\vspace{-.2in}
  \begin{algorithmic}
  \Function{\sf ex1}{$p$}
  \State $x \sim \texttt{bern}(0.5)$
  \If {$x$}
    \State $y \sim \texttt{bern}(p)$
  \Else
   \State $y \sim \texttt{bern}(0.5p)$
  \EndIf
  \State \Return $y$
  \EndFunction
  \end{algorithmic}
  \caption{Loop-free example}
  \label{fig:ex1_prog}
\end{wrapfigure}

\begin{figure}[t!]
\includegraphics[scale=1]{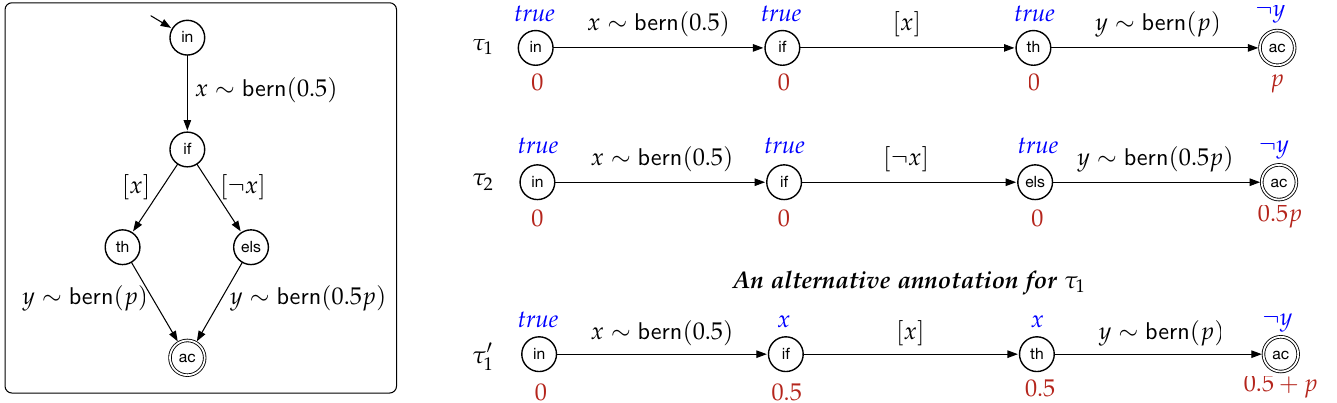}
\caption{A simple probabilistic program and possible trace annotations}
\label{fig:ex1}
\end{figure}

Our verification approach focuses on one trace at a time.  There are two
possible traces in our example program: one through the then branch and one
through the else branch of the conditional. We refer to these traces as
$\trace_1$ and $\trace_2$, respectively.

To prove accuracy properties about each trace, our technique annotates traces
with auxiliary information. Let us consider the annotated trace $\trace_1$
in \cref{fig:ex1_prog}. Each node along the trace is annotated with two
labels:
\rone the top/blue label is a logical formula representing a set of reachable program states at that point (these can be viewed as Hoare-style annotations);
\rtwo the bottom/red label is an expression representing the probability that the program does not end up in the blue states.
Consider node $\texttt{in}$ from $\trace_1$: it is labeled by
$\true$ and $0$, indicating that the probability of failing to arrive in a
program state satisfying $\true$ is 0 (as expected).
However, consider node $\texttt{ac}$: it is labeled with $\neg y$ and $p$, indicating that the probability of failing to arrive in a state where $y = \false$ is at most $p$.
The other program trace $\trace_2$, which traverses the other branch, is
similar; the annotation of $\trace_2$ demonstrates that its failure probability is at most $0.5p$.

At this point we have considered all of the $\mathsf{ex1}$'s traces.
If we na\"ively sum up their probabilities of failure to bound the total
failure probability, we get a failure probability of at most $p + 0.5p = 1.5p$,
which is too weak---we wanted to prove an upper bound of $p$.
However, we can give a more precise analysis since
the two traces consider two mutually disjoint events: one path assumes $x$
is true while the other assumes $x$ is false.
In this case, we can soundly take the \emph{maximum} of the two failure
probabilities, $0.5p$ and $p$, arriving at a total failure probability of $p$
and concluding the proof.\footnote{%
  This argument is an instance of a more general proof technique over sets of
  traces represented as automata; we will later formalize this idea as
  \emph{merging} two automata.}

Given a labeled trace, it is relatively straightforward to check if the
annotations are valid. However, constructing the annotations may not be so easy.
The main challenge is selecting labels for the results of sampling
instructions---the invariants are not fully determined by the program, and in
general the proper choice depends on the target property we are trying to
establish. For instance, it is also possible to give an alternative annotation of
$\trace_1$, denoted $\trace_1'$ in \cref{fig:ex1_prog}. Node $\texttt{if}$ is
labeled with $x$ and $0.5$, indicating that the probability of \emph{not}
arriving in a state where $x = \true$ is at most $0.5$.

This annotated trace illustrates another general feature of our analysis:
failure probabilities sum up along traces. Intuitively, this principle
corresponds to a basic property of probabilities called the \emph{union bound}:
$\pr(A \cup B) \leq \pr(A) + \pr(B)$ for any two events $A$ and $B$. In
particular, if $A$ and $B$ are interpreted as \emph{bad} events---events
violating labels at different nodes---the probability of any failure
occurring along a trace is at most the sum of the failure probabilities of
individual steps.
In $\trace_1'$, the probability of $y = \true$ at node $\texttt{ac}$ of
$\trace_1'$ is $p$, so the final
failure probability computed for this trace is $0.5 + p$.  While this annotation
in $\trace_1'$ is sound, it is too weak to prove our desired property.

\paragraph{Encoding Trace Semantics.}
Our technique cleanly separates probabilistic assertions into two pieces: a
non-probabilistic component describing the state of program variables (the blue
annotations in \cref{fig:ex1}), and a single number summarizing the
probabilistic part of the assertion (the red annotations in \cref{fig:ex1}). As a
result, we can reduce probabilistic reasoning to logical reasoning, allowing us
to harness the power of \abr{SMT} solvers and \emph{synthesis} techniques.

To illustrate, we show how to construct trace labels for $\trace_1$.
Our method proceeds in two steps.
First, like in traditional \emph{verification-condition generation}, we encode the semantics of trace $\trace_1$ and the specification as a logical formula, which, if valid, implies that $\hoare{p}{\true}{\trace_1}{\neg y}$.
% As is standard in non-probabilistic verification conditions, we encode each statement along the trace individually and conjoin the results.
Specifically, we construct the following verification condition:\footnote{We have simplified some aspects of
the encoding here; \cref{sec:interpolation} provides a formal treatment.}
\begin{equation}\label{eq:ex1}
\exists f_x, f_y \ldotp \forall x,y,\cost_i \ldotp (\cost_0 = 0 \land  \varphi ) \Rightarrow (\neg y \land \cost_3 \leq p)
\end{equation}
Above, $\varphi$ is a set of conjuncts, each encoding the semantics of one statement in $\trace_1$:
\[
\varphi \triangleq
\underbrace{
\left(\begin{array}{ll}
f_x = 1 \Rightarrow& x \land \cost_1 = \cost_0 + 0.5\\
f_x = 2 \Rightarrow& \neg x \land \cost_1 = \cost_0 + 0.5\\
f_x = 3 \Rightarrow& \cost_1 = \cost_0\\
\end{array}\right)}_{x \sim \bern(0.5)}
\land
\underbrace{
 (x \land \cost_2 = \cost_1)
}_{[x]}
\land
\underbrace{
\left(\begin{array}{ll}
f_y = 1 \Rightarrow& y \land \cost_3 = \cost_2 + 1-p\\
f_y = 2 \Rightarrow& \neg y \land \cost_3 = \cost_2 + p\\
f_y = 3 \Rightarrow& \cost_3 = \cost_2
\end{array}\right)
}_{y \sim \bern(p)}
\]

Let us explain how the encoding models the program $\prog$.
The variables $\cost_i$ are fresh real-valued variables that represent the probability of failure along the path---$\cost_0$, the initial probability at node $\texttt{in}$, is constrained to 0.
The right-hand side of the implication in \cref{eq:ex1} encodes the postcondition $\neg y$ and the upper bound on the failure probability $\cost_3 \leq p$.

The more interesting parts of the encoding are the existentially quantified variables $f_x,f_y$, which appear in $\varphi$; we assume that $f_x,f_y \in \{1,2,3\}$.
These are used to \emph{select} an \emph{axiomatization} for each sampling statement.
\emph{Synthesizing} the right values for $f_x$ and $f_y$ allows us to show that \cref{eq:ex1} is valid, and therefore prove correctness of $\trace_1$.
For instance, if $f_x$ is set to $1$, then $x\sim \bern(0.5)$ is encoded as an assignment statement $x \gets \true$ with an accumulated failure probability of $0.5$, since $x$ is not $\true$ with a probability of 0.5;
if $f_x$ is set to $3$, then $x$ is treated as a non-deterministic Boolean, incurring no probability of failure.

It is not hard to check that any proof of validity of \cref{eq:ex1} must set $f_x = 3$ and $f_y = 2$, as otherwise we cannot establish the postcondition, $\neg y$, or the upper bound on failure, $p$.
In general, we treat $f_x$ and $f_y$ as uninterpreted functions whose arguments
are program inputs, so that the choice of axiomatization may depend on the
program state (\cref{sec:interpolation} presents the general form).

\paragraph{Labels via Craig Interpolation.}
Suppose that we have proved validity of \cref{eq:ex1} and discovered that
setting $f_x = 3$ and $f_y = 2$ yields a satisfiable formula.
Plugging these values into \cref{eq:ex1} and negating the postcondition, we arrive at the following unsatisfiable formula:
\[
\cost_0 = 0
\land
\underbrace{
\cost_1 = \cost_0
}_{x \sim \bern(0.5)}
\land
\underbrace{
 x \land \cost_2 = \cost_1
}_{[x]}
\land
\underbrace{
 \neg y \land \cost_3 = \cost_2 + p
}_{y \sim \bern(p)}  \land (y \lor \cost_3 > p)
\]
In first-order logic, it is known that if $A \land B$ is unsatisfiable,
then there is a formula $I$ over the shared vocabulary of $A$ and $B$
such that $A \Rightarrow I$ and $I \Rightarrow \neg B$ are valid.
$I$ is called a \emph{Craig interpolant}.
Intuitively, an interpolant overapproximates $A$ while maintaining
unsatisfiability with $B$; this overapproximation can be seen as trying to
generalize the assertions as much as possible.
In our unsatisfiable formula above, we can compute a \emph{sequence} of
interpolants by splitting the formula into $A$ and $B$ segments after every
statement's encoding.\footnote{Equivalently, we can encode the problem as
solving \emph{recursion-free Horn clauses}~\citep{rummer2013classifying}.}
The resulting interpolants compactly encode the two labels on traces, the sets of states and probabilities of failure.
E.g., consider the  split:
\[
\begin{array}{rrr}
A \triangleq \cost_0 = 0
\land
\underbrace{
\cost_1 = \cost_0
}_{x \sim \bern(0.5)}
\land
\underbrace{
 x \land \cost_2 = \cost_1
}_{[x]}
&&
B \triangleq
\underbrace{
 \neg y \land \cost_3 = \cost_2 + p
}_{y \sim \bern(p)}  \land (y \lor \cost_3 > p)
\end{array}
\]
A possible interpolant for $A \land B$ is $I \triangleq \cost_2 = 0$.  This
indicates that any program state is reachable at node $\texttt{th}$ (since
program variables are unconstrained in $I$) with a probability of failure $0$.  The
interpolant condition ensures that $I$ can only mention $\cost_2$,
the only variable shared by $A$ and $B$.

\subsection{Illustrative Example: Handling Loops}\label{sec:ex2}
\begin{wrapfigure}{r}{0.3\textwidth}
  \vspace{-.2in}
\begin{algorithmic}
\Function{\sf ex2}{$q, n, \eps, p$}
\State $i \gets  0$
\While {$i < n$}
  \State $a[i] \gets \lap(q[i],\frac{1}{\eps})$
  \State $i \gets i + 1$
\EndWhile
\State \Return $a$
\EndFunction
\end{algorithmic}
\caption{Example with a loop}
\label{fig:ex2_prog}
\end{wrapfigure}
We now consider a more complex example with loops, $\mathsf{ex2}$ in
\cref{fig:ex2_prog}. $\mathsf{ex2}$ is a simplified sketch of mechanisms from
\emph{differential privacy}~\citep{DMNS06}, which carefully add random noise to
query results before disclosing them.  The program $\mathsf{ex2}$ takes an array
of integers $q$ of length $n$, and constructs an array $a$ whose values are
noisy versions of those in $q$.  Specifically, for each element $q[i]$, $a[i]$
is noise drawn from the Laplace distribution with \emph{mean} $q[i]$ and
\emph{scale} $1/\eps$, where $\eps > 0$ is a real-valued input to the program.
(All primitive distributions are defined in \cref{sec:problem}.)

Our goal is to prove the accuracy property
$
\hoare{p \cdot n}{\true}{\mathsf{ex2}(q,n,\eps)}{\post},
$
where the post-condition is defined to be
\[
\post \triangleq \forall j \in [0,n) \ldotp |a[j] - q[j]| \leq \frac{1}{\eps}\log\left(\frac{1}{p}\right) .
\]
In other words, for any $p \in (0,1]$, we want to verify that the difference
between $a[j]$ and $q[j]$ is bounded by a function of $\eps$ and $p$. Observe
that $\post$
involves input parameters $q, n, \epsilon$, and $p$, but $p$ does not
appear in the program---the accuracy property is a parameterized family of
properties. From our postcondition, we see that we can guarantee tighter bounds
on the error---the difference between the exact answer $q[j]$ and the noisy
answer $a[j]$---if we are willing to allow this property to be violated with
larger probability $p \cdot n$. This style of postcondition is common for
many randomized algorithms, capturing the
relationship between \emph{accuracy}---how far the results are from the exact
values---and probability of failure, or how often the target property will not hold.

\begin{figure}[t!]
  \includegraphics[scale=1]{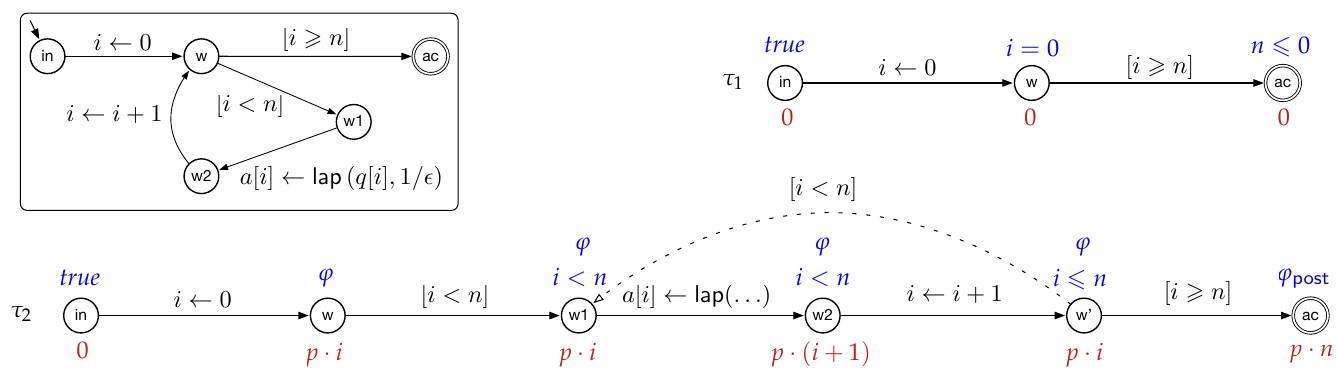}
  \caption{A looping illustrative example}
\label{fig:ex2}
\end{figure}

\paragraph{Trace Generalization.}
The control-flow automaton representation of $\mathsf{ex2}$ is shown in the box in
\cref{fig:ex2}.
While the total number of loop iterations is at most the input parameter $n$ in
the original program, the automaton abstraction overapproximates these program
behaviors with an infinite number of traces due to the loop.
Therefore, unlike our first example, we cannot construct a proof for every trace individually.
Our technique proceeds by \emph{picking} traces, \emph{proving} them correct, and \emph{generalizing} them into automata representing infinite sets.

Let us first consider trace $\trace_1$ in \cref{fig:ex2}; this trace does
not enter the loop. The trace is easily shown to be correct since not entering
the loop implies that $n\leq 0$, vacuously implying $\post$ with failure
probability $0$. More interesting is trace $\trace_2$ in
\cref{fig:ex2}, which executes the loop body once and exits.  The formula
$\varphi$ in the annotation is defined as follows:
\[
\varphi \triangleq \forall j \in [0,i) \ldotp |a[j] - q[j]| \leq \frac{1}{\eps}\log\left(\frac{1}{p}\right)
\]
Notice the probability of failure is $p \cdot i$ on nodes $\texttt{w},\texttt{w1},\texttt{w2},$ and $\texttt{w'}$.
After loop exit, using the exit condition, we conclude that the probability of
failure is $p \cdot n$.  Informally, these labels capture the fact that the
failure probability depends on how many times we have executed the loop, which
is tracked in the counter $i$.

Our algorithm discovers that the labels are \emph{inductive}: no matter how many times we execute the loop, the probability of failing to satisfy $\varphi \land i < n$ at loop entry is $p \cdot i$.
Therefore, the algorithm generalizes this trace
into an infinite set of traces by adding an edge from node $\texttt{w'}$ to
$\texttt{w1}$ with the statement $[i < n]$.
% Using this annotation, our algorithm discovers that we can generalize this trace
% into an infinite set of traces by adding an edge from node $\texttt{w'}$ to
% $\texttt{w1}$ with the statement $[i < n]$. This maintains soundness of the
% annotation since starting from a state satisfying the annotation $\varphi \land
% i \leq n$ at $\texttt{w'}$ and assuming $i < n$ establishes the annotation
% $\varphi \land i < n$ at $\texttt{w1}$ without increasing the failure
% probability. In a formula:
% \[
% \hoare{0}{\varphi \land i \leq n}{[i < n]}{\varphi \land i < n}
% \]
% \cs{This last sentence might be problematic - not much to suggest why there's no additional probability of failure at this point. }
% \jh{Hope this is better...}
%
With this additional edge in place, we now have an automaton representing  all
traces that go through the loop at least once. The total failure probability of
those traces is the label under node $\texttt{ac}$: $p \cdot n$.  Combined with
trace $\trace_1$, we have covered all the traces of $\mathsf{ex2}$, proving that
the total probability of failure is $p \cdot n + 0 = p \cdot n$ as desired.

% \jh{It doesn't matter here, but is this conceptually a $+$ or a max?}
% \aws{it's a + since we cannot merge}
% \jh{Returning to this: if $\trace_2$ had node $w$ labeled by $(i = 0, 0)$ then
% we could merge, yes? I think this would also be a well-labeled trace.}
% \aws{yes, that's right}

\paragraph{Selecting Axioms for the Laplace Distribution.}
The sampling statement $a[i] \sim \lap(\ldots)$ in $\trace_2$ is encoded by the following logical formula:\footnote{In practice, we treat non-linear arithmetic operations and transcendentals (e.g., log) as uninterpreted functions and use the theorem-enumeration technique recently proposed by~\citet{Srikanth2017} to axiomatize them.}
\[
\begin{array}{lcr}
|a[i] - q[i]| \leq \frac{1}{\eps}\log\left(\frac{1}{f_a(i,p,n)}\right) & \land &
\cost_3 = \cost_2 + f_a(i,p,n)
\end{array}
\]
The left conjunct specifies that we can assume that the difference between $a[i]$ and
$q[i]$ is at most $\frac{1}{\eps}\log\left(\frac{1}{f_a(i,p,n)}\right)$; the
right conjunct specifies that this assumption fails with a probability of
$f_a(i,p,n)$.  We treat $f_a$ as an uninterpreted function with range $(0,1]$,
so that there are infinitely many possible interpretations of $f_a$
corresponding to different failure probability/accuracy tradeoffs for the
Laplace distribution. To get the annotation proving correctness of $\trace_2$ in
\cref{fig:ex2}, our technique synthesizes the interpretation $f_a(i,p,n) = p$. With
this choice, our analysis accumulates a probability of failure of $p$ for every
loop iteration, ending up with a total probability of failure of $p \cdot n$.

%% file: problem.tex
\section{Preliminaries} \label{sec:problem}

In this section, we formalize our program model and accuracy specifications.

\subsection{Program Model and Semantics}\label{sec:model}

\paragraph{Probability Distributions.}
To model probabilistic computation mathematically, we use probability
sub-distributions. A function $\mu : \aset \to [0,1]$ defines a \emph{discrete
sub-distribution} over a set $\aset$ if it is non-zero for at most countably
many elements in $\aset$, and $\sum_{\aelem \in \aset} \mu(\aelem) \leq 1$; we
will abbreviate discrete sub-distribution as distribution throughout this paper.
We will often write $\mu(\aset')$ for a subset $\aset' \subseteq \aset$ to mean
$\sum_{\aelem \in \aset'} \mu(\aelem)$. We write $\sdist(\aset)$ for the set of
all distributions over $\aset$. The \emph{support} of a distribution $\mu$ is
defined as $\support(\mu) \triangleq \{\aelem \in \aset \mid \mu(\aelem) > 0\}$.

% \aws{i modified sentence below to say we sample from countable sets,
% but other program variables may be continuous}
We focus on discrete sub-distributions to keep measure-theory overhead to a
minimum. As a consequence, we only allow programs to sample from primitive
discrete distributions. Supporting continuous primitive distributions (e.g., the
Gaussian distribution) would not introduce any difficulties beyond requiring a
more technically involved definition of the program semantics.

\begin{table}[t]
  \caption{Distribution expressions and their semantics}
  \label{tbl:dexpr}
  \footnotesize
  \begin{tabular}{llll}
    \toprule
    Name & Dist. expr. $\dexpr$ & Parameters & Semantics $\stt(\dexpr)$ \\
    \midrule
    Bernoulli
    &
    $\bern(\expr)$
    &
    $\expr \in [0,1]$
    &
    $\mu(\true) = \stt(\expr)$ and $\mu(\false) = 1-\stt(\expr)$
    \\
    Uniform
    &
    $\unif(\expr)$
    &
    $\expr$ is a finite set
    &
    $\mu(c) = 1/|\stt(\expr)|$, for $c \in \stt(\expr)$
    \\
    Laplace
    &
    $\lap(\expr_1,\expr_2)$
    &
    mean $\expr_1 \in \mathds{Z}$; scale $\expr_2 \in \mathds{R}^{>0}$
    &
    $\mu(c) \propto \exp\left( - \frac{ |c - s(\expr_1)| }{ s(\expr_2) } \right)
    $, for $c \in \mathds{Z}$
    \\
    Exponential
    &
    $\olap(\expr_1,\expr_2)$
    &
    shift $\expr_1 \in \mathds{Z}$;  scale $\expr_2 \in \mathds{R}^{>0}$
    &
    $\mu(c) \propto \exp\left( - \frac{ c - s(\expr_1) }{ s(\expr_2) } \right)$,
    for $c \in \mathds{Z}$ and $c \geq s(\expr_1)$
    \\
    \bottomrule
  \end{tabular}
\end{table}

\paragraph{Program Expressions.}
We fix a set of variables $\vars$ that appear in the program.
A program \emph{state} $\stt$ is a map assigning every variable $\var \in
\vars$ to a value. We will use $\stts$ to denote the set of all possible
states.
Given variable $v$, we use $\stt(v)$ to denote the value of $v$ in state
$\stt$.
Given constant $c$, we use $\stt[v\mapsto c]$ to denote the state $\stt$
with variable $v$ mapped to $c$.
The semantics of an expression $\expr$ is a function $\sem{\expr} : \stts \to
\dom$ from a state to an element of some type $\dom$.
For instance, the expression $x + y$ in state $s$ is interpreted as
$\sem{x+y}(\stt) = \stt(x) + \stt(y)$.
We will often abbreviate $\sem{\expr}(\stt)$ by $\stt(\expr)$.

% \aws{It may be unclear to the reader that $B,U,$ etc are distributions like $\mu$ above---should we use $\mu$-like notation?}
% \jh{I tried giving the types below... is it better, or still confusing?}
% \aws{In an attempt to distill the definitions below, I made a table.
% If you like it and I'll go ahead and squeeze the stuff below.}

% \aws{Please check below}
% \jh{LGTM}

\paragraph{Distribution Expressions.}
A \emph{distribution expression} $\dexpr$ is interpreted as a distribution
family $\sem{\dexpr} : \stts \to \sdist(\dom)$, mapping a state in $\stts$
to a distribution over $\dom$ with countable support.
Our framework can naturally handle any distribution expression that can be
interpreted as a discrete distribution. For concreteness, we will consider the
four primitive distributions in \cref{tbl:dexpr}.

Consider the Bernoulli distribution expression, $\bern(\expr)$.  Given a state
$\stt$, semantically $\bern(\expr)$ is the distribution $\mu \in
\dist(\mathds{B})$ where $\mu(\true) = \stt(\expr)$ and $\mu(\false) =
1-\stt(\expr)$.  Similarly, the uniform distribution expression $\unif(\expr)$,
where $\expr$ encodes to a finite set, is interpreted as the distribution
assigning equal probability to every element in $\stt(\expr)$.

We also use the (discrete) \emph{Laplace distribution}, a common primitive
distribution in the theory of differential privacy.  For a state $\stt$, the
distribution expression $\lap(\expr_1,\expr_2)$ is semantically the discrete
Laplace distribution with \emph{mean} $\stt(\expr_1)$ and \emph{scale}
$\stt(\expr_2)$: for every integer $c \in \mathds{Z}$, it assigns a probability
proportional to $\exp\left( - \frac{ |c - s(\expr_1)| }{ s(\expr_2) } \right)$.
The (discrete) \emph{exponential distribution} expression
$\olap(\expr_1,\expr_2)$ is similar, but only assigning positive probability to
integers above  the \emph{shift} $\stt(\expr_1)$.

We implicitly assume that arguments of distribution expressions are well-typed and valid.

\paragraph{Programs, Statements, and Traces.}
Our verification technique will target programs written in a probabilistic,
imperative language. The \emph{basic statements} are drawn from
a set $\stmts$:
\begin{itemize}
  \item \emph{Assignment} statements $\Assg{\var}{\expr}$, where $\expr$ is an
    expression over $\vars$, e.g., $\var_1 + \var_2$.
  \item \emph{Sampling} statements $\Rand{\var}{\dexpr}$, where $\dexpr$ is a
    distribution expression.
  \item \emph{Assume} statements $\Assm{\bexpr}$, where $\bexpr$ is a Boolean
    expression over $\vars$.
\end{itemize}
%
% \aws{i removed the program grammar -- i think it confuses to have different notations, but i kept the reference to conditional statements}
% Programs can combine basic statements using the usual imperative constructs:
% %
% \[
%   \prog \coloneqq \stmts
%   \mid \Seqn{\prog}{\prog}
%   \mid \Cond{\bexpr}{\prog}{\prog}
%   \mid \Whil{\bexpr}{\prog}
% \]
%
A \emph{trace} $\trace$ is a finite sequence of statements $\stmt_1 ; \cdots ;
\stmt_n$, and a program $\prog$ is interpreted as a (possibly infinite) set of
traces $\lang(\prog)$.
We include full details of the programming language in
\smref{smith2019:arxiv}{app:ahl};
the interpretation is standard, using assume
statements to model typical control-flow constructs.
For instance, a conditional statement $\Cond{\bexpr}{\trace_1}{\trace_2}$ can be
modeled as the pair of traces $\Assm{\bexpr} ; \trace_1$ and $\Assm{\neg \bexpr}
; \trace_2$.
By construction, traces in $\lang(\prog)$ are semantically disjoint---no trace
in $\lang(\prog)$ is a prefix of (or equal to) any other trace in
$\lang(\prog)$, and the first differing statements between any two traces are of
the form $\assume(\bexpr)$ and $\assume(\neg \bexpr)$.
% \aws{this still needs finessing -- traces have to be semantically+syntactically disjoint}
% \jh{Better, hopefully.}

\paragraph{Trace Semantics.}
We interpret a trace $\trace$ as a function $\sem{\trace}: \stts
\to \sdist(\stts)$ from input states to distributions over output states. To define this
semantics formally, we need two standard constructions on distributions. The map
$\dunit : \dom \to \sdist(\dom)$ maps $a \in \dom$ to the Dirac distribution
$\delta_a$ at $a$, i.e., the distribution that returns $1$ at $a$ and $0$ otherwise.
The map $\dbind : \sdist(\dom_1) \to (\dom_1 \to \sdist(\dom_2)) \to
\sdist(\dom_2)$ combines probabilistic computations in sequence:
$
  \dbind(\mu, f)(a_2) = \sum_{a_1 \in \dom_1} \mu(a_1) \cdot f(a_1)(a_2) .
$
These maps are the usual unit and bind for the (sub-)distribution monad. Then,
we can give semantics to basic statements and traces as shown in~\cref{fig:sem}.
\begin{figure}[t!]
\begin{align*}
  \sem{\Assg{\var}{\expr}}(\stt) &\triangleq \dunit(\stt[\var \mapsto \sem{\expr}(\stt)])
  &&&
  \sem{\Rand{\var}{\dexpr}}(\stt) &\triangleq \dbind(\sem{\dexpr}(\stt), \lambda x.\, \dunit(\stt[\var \mapsto x]))
  \\
  \sem{\Assm{\bexpr}}(\stt) &\triangleq \text{if } \sem{\bexpr}(\stt) \text{ then } \dunit(\stt) \text{ else } 0
  &&&
  \sem{\Seqn{\stmt}{\trace}}(\stt) &\triangleq \dbind(\sem{\stmt}(\stt), \sem{\trace})
\end{align*}
\caption{Statement and trace semantics}
\label{fig:sem}
\end{figure}
%
% This semantics extends to traces, by induction on the length of the trace:
% \[
%   \sem{\stmt ; \trace}(\stt) \triangleq \dbind(\sem{\stmt}(\stt), \sem{\trace})
% \]
% %
Finally, the semantics of a program $\prog$ is defined as the aggregate of its
traces. Formally, $\sem{\prog}: \stts \to \sdist(\stts)$ is defined as
\[
  \sem{\prog}(s) \triangleq \sum_{\trace \in \lang(\prog)} \sem{\trace}(\stt)
\]
where each term $\sem{\trace}(\stt)$ is the output distribution from running
$\trace$ starting from input $\stt$, and the sum of distributions is defined
pointwise. For any disjoint set of traces corresponding to a program $\prog$,
the sum on the right-hand side is indeed a distribution.

% \aws{the above definition could technically result in > 1 probability, since paths are not mutually disjoint}
% \jh{Fixed it up a bit. It's not completely formal, since we don't introduce the imperative language.}

\subsection{Programs as Automata}
We can encode
the set of possible traces of a program $\prog$ as a regular language $\lang(\prog)$ represented by all paths through its control-flow graph.
We begin with a general definition of automata over program statements, and then show how we represent programs as automata.

\paragraph{Automata over Statements.}
A \emph{finite-state automaton over statements} $\fsm$ is a graph
$\tuple{Q,\delta}$, where
\begin{itemize}
  \item $Q$ is a finite set of \emph{nodes}.
  \item $\delta \subseteq Q \times \stmts \times Q$ is the \emph{transition
    relation}, where $\stmts$ are basic statements.
  \item $\qentry, \qexit \in Q$ are special nodes called the \emph{initial} and \emph{accepting} nodes, respectively.
\end{itemize}
We will use $q_i \lto{\stmt} q_j$ to denote that $\tuple{q_1,\stmt,q_j} \in
\delta$. We write $\lang(\fsm)$ for the language of traces \emph{accepted} by
 $\fsm$, where a trace $\stmt_1,\ldots,\stmt_n$ is accepted iff
$
  \{\qentry \lto{\stmt_1} q_1, q_1 \lto{\stmt_2} q_2, \ldots, q_{n-1} \lto{\stmt_n} \qexit\} \subseteq \delta.
$
It will sometimes be useful to use multiple automata to model the traces in a
single program. We will use $\lang(\fsms)$ to denote the union of all languages
accepted by a set of automata $\fsms$, i.e., $\bigcup_{\fsm \in \fsms}
\lang(\fsm)$.

We assume that all nodes $q \in Q$ can reach the accepting node $\qexit$ via the
transition relation $\delta$, and that there are no transitions starting from
$\qexit$.  We also assume that automata model well-formed control flow, i.e.,
\rone all nodes $q_i \in Q$ have at most two outgoing transitions and \rtwo if
$q_i \lto{\stmt_1} q_j$ and $q_i \lto{\stmt_2} q_k$ for $j \neq k$, then
$\stmt_1, \stmt_2$ are of the form $\assume(\bexpr_1)$ and $\assume(\bexpr_2)$,
such that $\bexpr_1 \equiv \neg \bexpr_2$.

\paragraph{From Program Traces to Automata}
We will identify a \emph{program} with an automaton representing its
its \emph{control-flow graph} (\abr{CFG}). A program $\prog$ is of
the form $\tuple{\locs,\delta}$, where the nodes $\locs$ of the automaton
denote the set of \emph{program locations} (e.g., line numbers).  The special nodes
$\lentry,\lexit \in \locs$ model the first and last lines of the program.
To ensure there is no control-flow non-determinism, we assume that for any
$\loc_i \lto{\assume(\bexpr)} \loc_j$, there is a transition $\loc_i
\lto{\assume(\neg \bexpr)} \loc_k$.

We use $\varsi \subseteq \vars$ to denote the set of input variables, which are
not modified by the program.  We will also use $\varsd \subseteq \vars$ to
denote the set of program variables whose values are assigned
deterministically, i.e., not affected by probabilistic choice---by definition,
$\varsi \subseteq \varsd$. (We may not be able to determine $\varsd$ exactly in
practice, but we can under-approximate it via a simple static analysis.)

\subsection{Probabilistic Accuracy Properties}
% Now that we have seen the program model, we are ready to introduce our target
% properties. Specifications consist of a pre- and post-postcondition on states,
% along with an upper bound on the probability that the post-condition does not
% hold at the end of program execution; we call this upper bound the \emph{failure
% probability}.

We will define specifications using the Hoare-style statement
\[
\hoare{\errp}{\pre}{\trace}{\post}
\]
where
the \emph{precondition} $\pre \subseteq \stts$ and \emph{postcondition} $\post \subseteq \stts$ are  sets of program states,
and the \emph{failure probability} $\errp$ is a $[0,1]$-valued function
   over input variables $\varsi$.
For simplicity, we will treat $\errp$ as an expression over $\varsi$---%
e.g., $0$ or $p \cdot n$ in \cref{sec:ex2}---and use $\stt(\errp)$ to denote the value of $\errp$ in state $\stt$.

We say that $\hoare{\errp}{\pre}{\trace}{\post}$ is \emph{valid} iff
for any state $\stt \in \pre$, we have $\mu(\overline{\post})
\leq \stt(\errp)$, where $\mu = \sem{\trace}(s)$ and $\overline{\post} = \stts
\setminus \post$.  In other words, the probability that the trace starts in
$\pre$ and \emph{does not} end up in $\post$ is \emph{upper bounded} by $\errp$.
We extend this notation to programs $\prog$ in the natural way, writing
$\hoare{\errp}{\pre}{\prog}{\post}$ iff for any input state
$\stt \in  \pre$, the output distribution $\mu = \sem{\prog}(s)$ satisfies the
bound $\mu(\overline{\post}) \leq \stt(\errp)$.

%% file: annotations.tex
%!TEX root=paper.tex

\section{Trace Abstraction Modulo Probability} \label{sec:annotations}

With the preliminaries out of the way,
we begin to introduce a version of trace abstraction for probabilistic
programs and show how to use it to prove accuracy specifications. Given a
program $\prog$, suppose we want to establish the following accuracy
specification:
$
  \hoare{\errp}{\pre}{\prog}{\post}.
$
We will overapproximate the traces of $\prog$ with a set of automata $\fsms$ and
analyze each automaton separately; this way, we can focus on smaller groups of
possible traces. If we can show that the probability $\post$ does not hold
across all automata is at most $\errp$, this implies the accuracy specification.
We formalize this argument in the following proof rule. (We defer all
proofs to \smref{smith2019:arxiv}{app:proofs}.)

\begin{theorem}[General proof rule]\label{thm:gproofrule}
  The specification $
    \hoare{\errp}{\pre}{\prog}{\post}
  $ is valid
  if there exists a set of automata $\fsms$ such that
  \begin{align}
    \lang(\prog) \subseteq \lang(\fsms) \tag{Trace inclusion}\\
    \forall s \in \pre \ldotp \sum_{\trace \in \lang(\fsms)}  \sem{\trace}(s)(\overline{\post}) \leq \stt(\errp) \tag{Failure probability upper bound}
  \end{align}
\end{theorem}

This proof rule is concise but difficult to apply in practice, even given the
set of automata $\fsms$---while the trace inclusion property can be checked via
regular language inclusion, the failure probability upper bound is more
complicated. To make this second condition easier to check, we enrich the
automata with additional information on each state; local properties of these
labeled automata will then imply the failure probability upper bound.

\paragraph{Enriching Automata with Labels.}
We work with automata where every node is labeled with a predicate on states
(equivalently, a set of states), and a function representing the failure
probability---we call such automata \emph{failure automata}. The rough intuition
is that at each node $q$, the predicate label represents a program invariant
that holds on all traces reaching $q$ from the beginning of the program, except
with probability given by the failure probability label.

\begin{definition}[Failure automata]\label{def:abstract}
A \emph{failure automaton} $\fsm = \tuple{Q,\delta,\labelp,\labele}$ is an
automaton $\tuple{Q,\delta}$ with two labeling functions,
$\labelp$ and $\labele$, where
\begin{itemize}
  \item $\labelp$ maps every node $q \in Q$ to a set of states, and
  \item $\labele$ maps every node $q \in Q$ to a [0,1]-valued function over $\varsd$.
\end{itemize}
We say that $\fsm$ is \emph{well-labeled} iff the following conditions hold:
\begin{enumerate}
\item $\labele(\qentry) = 0$ and $\labele(\qexit)$ is a $[0, 1]$-valued
  function over the input variables $\varsi \subseteq \varsd$, and
\item for every transition $q_i \lto{\stmt} q_j$, the statement
  $
    \hoare{\wpprob(\labele(q_j), \stmt) - \labele(q_i)}{\labelp(q_i)}{\stmt}{\labelp(q_j)},
  $ is valid
  where $\wpprob$ is a weakest-precondition operation over failure-probabilities:
  $\wpprob(\expr,\stmt) \triangleq \expr$ for assume and sampling statements, and
  $\wpprob(\expr_1, \var \gets \expr_2) \triangleq \expr_1[\var \mapsto \expr_2]$.
\end{enumerate}
\end{definition}

The two conditions ensure that if we take any trace $\trace \in \lang(\fsm)$,
then
$\hoare{\labele(\qexit)}{\labelp(\qentry)}{\trace}{\labelp(\qexit)}$ is valid.
Point (2) ensures that failure probability accumulates
additively as we move along the trace, starting from being 0 at $\qentry$, as
stipulated by point (1). Crucially, both points are \emph{local} conditions:
they can be easily checked given a failure automaton. However, coming
up with well-labeled automata for a given program is not at all trivial---we
return to this question in the next two sections.

\begin{example}
  Recall our example from~\cref{sec:ex1}, illustrated in~\cref{fig:ex1}.
  The lower part of~\cref{fig:ex1} shows a failure automaton named $\trace_1'$
  with $\labelp$ and $\labele$ shown above and below the nodes, respectively.
  Notice that the initial node $\texttt{in}$ is labeled with $\labelp(\texttt{in}) \triangleq \true$
  and $\labele(\texttt{in}) \triangleq 0$.
  Focusing on the edge from node $\texttt{th}$,
  the labeling at $\texttt{ac}$ satisfies condition (2) for well-labeledness in~\cref{def:abstract}.
  The condition says that the following statement must be valid:
  $
  \hoare{p}{x}{y \sim \bern(p)}{\neg y}.
  $
  The failure probability $p$ is the simplification of the expression $\wpprob(0.5 + p, y \sim \bern(p)) - 0.5$.
  The statement is valid since $y$ is true with probability $p$
  after executing $y \sim \bern(p)$.
  %, therefore failing the condition
  %$\neg y$ with probability $p$.
%
  % Recall the example from~\cref{sec:ex2}, illustrated in~\cref{fig:ex2}.
  % The lower part of~\cref{fig:ex2} shows an failure automaton
  % with $\labelp$ and $\labele$ shown above and below the respective nodes.
  % Notice that the initial node, $\texttt{in}$, is labeled with $\labelp(\texttt{in}) \triangleq \true$
  % and $\labele(\texttt{in}) \triangleq 0$.
\end{example}

The following theorem establishes soundness of annotations on well-labeled
automata. Specifically, the failure probability label on $\qexit$---namely,
$\labele(\qexit)$---is an upper bound on the probability that executions through
$\fsm$ do not end up in a state in $\labelp(\qexit)$.

\begin{theorem}[Well-labeled automata soundness] \label{thm:labeled-sound}
Let $\fsm$ be a well-labeled failure automaton.
Then, for every $\stt \in \labelp(\qentry)$ and $\mu = \sem{\trace}(\stt)$, we have
$
\sum_{\trace \in \lang(\fsm)} \mu(\overline{\labelp(\qexit)}) \leq \stt(\labele(\qexit)) .
$
\end{theorem}

\paragraph{Proofs from Well-labeled Automata.}
Now that we have established soundness of well-labeled automata, we refine our
original proof rule (\cref{thm:gproofrule}) using failure automata. The
following theorem demonstrates how to establish correctness using a set of
failure automata. %whose sum of failure probabilities does not exceed $\errp$.

\begin{theorem}[Proof rule with failure automata]\label{thm:proofrule}
  The statement
  $
  \hoare{\errp}{\pre}{\prog}{\post}
  $
  is valid
  if there exist well-labeled automata $\fsms = \{\fsm_1, \ldots, \fsm_n\}$
  such that the following conditions hold:
  \begin{align}
   \lang(\prog) \subseteq \lang(\fsms) \tag{Trace inclusion}\\
    \forall i \in [1,n] \ldotp \pre \subseteq \labelp_i(\qentry_i) \tag{Precondition inclusion}\\
   \forall i \in [1,n] \ldotp \labelp_i(\qexit_i) \subseteq \post \tag{Postcondition inclusion}\\
   \forall \stt \in \pre \ldotp \sum_{i=1}^n \stt(\labele_i(\qexit_i)) \leq \stt(\errp) \tag{Failure probability upper bound}
  \end{align}
\end{theorem}

The trace, precondition, and postcondition inclusion conditions are the same as
in trace abstraction for non-probabilistic programs.
The failure probability upper bound condition
ensures that the overapproximation of failure probability resulting from
abstraction does not exceed $\errp$. Notice that precondition and postcondition
inclusion checks can be performed using an \abr{SMT} solver, assuming labels are
encoded in a supported first-order theory. Similarly, the failure probability
upper bound condition involves summing up the labels on the accepting nodes of
all failure automata, allowing us to perform the check with an \abr{SMT} solver.

\begin{example}
  Recall the example program $\mathsf{ex2}$ from~\cref{sec:ex2}, illustrated in~\cref{fig:ex2}.
  The two automata, denoted $\trace_1$ and $\trace_2$ in~\cref{fig:ex2},
  are well-labeled.
  The automata cover all program traces.
  The initial nodes, denoted $\texttt{in}$, have the labels $\labelp$
  as $\true$, therefore satisfying the precondition inclusion condition.
  The accepting nodes, denoted $\texttt{ac}$, both imply the postcondition, $\post$.
  Finally, the sum of the failure probabilities on accepting nodes is $0 + p
  \cdot n \leq p \cdot n$,  satisfying the failure probability condition.
\end{example}

%% file: algorithm.tex
%!TEX root=paper.tex

\section{Constructing Trace Abstractions}\label{sec:algorithm}

\begin{figure}[t!]
  \centering
\begin{prooftree}
\justifies
\fsms \longrightarrow \emptyset
\using \cinit
\end{prooftree}
\hspace{.4in}
\begin{prooftree}
  \trace \in \lang(\prog) \cap \overline{\lang(\fsms)}
  \quad
  \quad
  \fsm_\trace = \labeltrace(\trace, \pre,\post,\errp)
\justifies
\fsms \longrightarrow \fsms \cup \{\fsm_\trace\}
\using \ctrace
\end{prooftree}

\vspace{.35in}

\begin{prooftree}
  \begin{array}{c}
  \fsm = \tuple{Q,\delta,\labelp,\labele} \in \fsms
  \quad
  \quad
  q_i,q_j \in Q
  \quad
  \quad
  \stmt \in \stmts
  \vspace{.05in}\\
  \fsm' = \tuple{Q,\delta\cup\{q_i\lto{\stmt}q_j\},\labelp,\labele}
  \quad\quad
  \hoare{\wpprob(\labele(q_j),\stmt) - \labele(q_i)}{\labelp(q_i)}{\stmt}{\labelp(q_j)}
\end{array}
\justifies
\fsms \longrightarrow (\fsms \setminus \{\fsm\}) \cup \{\fsm'\}
\using \cgen
\end{prooftree}

\vspace{.35in}

\begin{prooftree}
  \begin{array}{c}
    \fsm_1, \fsm_2 \in \fsms
    \quad\quad
    \fsm = \fsm_1 \mergeop \fsm_2
  \end{array}
\justifies
\fsms \longrightarrow (\fsms \setminus \{\fsm_1,\fsm_2\}) \cup \fsm
\using \cmerge
\end{prooftree}

\vspace{.35in}

\begin{prooftree}
  \begin{array}{c}
    \lang(\prog) \subseteq \lang(\fsms)
    \quad\quad
    \forall \stt \in \pre \ldotp \sum_{i=1}^{|\fsms|} \stt(\labele_i(\qexit_i)) \leq \stt(\errp)
  \end{array}
\justifies
  \hoare{\errp}{\pre}{\prog}{\post}
\using \ccorr
\end{prooftree}

\caption{Overall abstract algorithm}\label{alg:overall}
\end{figure}

\Cref{thm:proofrule} reduces checking accuracy properties to finding a set of
well-labeled automata. Our algorithm for automating this proof rule is
technically complex, and spans the following two sections. Here, we will present
the algorithm and prove soundness, assuming a procedure for well-labeling
single traces; we will detail this last piece in \cref{sec:interpolation}.
Then, we compare our algorithm with two existing techniques: the union bound
logic \abr{aHL}, and standard trace abstraction.

\subsection{Algorithm Overview}
Our algorithm maintains a set $\{\fsm_i\}_i$ of well-labeled failure automata
modeling some of the program traces, and repeatedly finds traces $\trace \in
\lang(\prog)$ that are not in $\{\fsm_i\}_i$. If a trace can be
well-labeled, it is converted into a well-labeled automaton $\fsm_i$ proving
that $\hoare{\errp}{\pre}{\trace}{\post}$ and added to the current automaton
set. Throughout, the algorithm may simplify or transform the automaton set by
merging automata together and generalizing automata by adding new edges. The
% generalizing automata by adding new edges, and merging automata together. The
process terminates successfully if the set of failure automata $\{\fsm_i\}_i$
satisfies the conditions in \cref{thm:proofrule}.

The input to the algorithm is a program $\prog$, a pre- and post-condition
$\pre$ and $\post$, and a target failure probability $\beta$, a function over
the input variables of the program. The entire algorithm is presented in
\cref{alg:overall} as a set of non-deterministic guarded rules.
%our
%implementation detailed in \cref{sec:evaluation} searches for a succeeding
%sequence of rules to apply.
The algorithm preserves the invariant that the set
of automata $\fsms$ are well-labeled. We briefly consider each rule in turn.

\paragraph{Initialization.}
The rule $\cinit$ is the only rule with no premises and serves as the
initialization rule. Not surprisingly, the set of failure automata $\fsms$ is
initially empty.

\paragraph{Trace Sampling.}
The rule $\ctrace$ picks a trace $\trace$ that is in the program $\prog$ but
not covered by the set of automata $\fsms$.  It then uses the function
$\labeltrace$ to construct a well-labeled automaton $\fsm_\trace$ implying that
$\hoare{\errp}{\pre}{\trace}{\post}$.
We will detail the $\labeltrace$
operation in \cref{sec:interpolation};
for now, we just note that $\labeltrace$ may
fail, in which case the rule $\ctrace$ does not fire and the algorithm tries a
different trace.

\paragraph{Generalizing Automata.}
The rule $\cgen$ expands the language $\lang(\fsms)$ by adding new
edges to an automaton $\fsm \in \fsms$. When the new edges form loops, this rule
can be seen as generalizing from automata modeling finite unrollings of looping
statements to automata overapproximating loops. The side-conditions ensure that
this transformation preserves well-labeledness.

\paragraph{Merging Automata.}
The rule $\cmerge$ combines automata whose traces are \emph{mutually exclusive},
allowing us to take the \emph{maximum} failure probability instead of the
\emph{sum}. Intuitively, automata that begin with the
same prefix of statements before making mutually exclusive assumptions---say,
$\assume(b)$ and $\assume(\neg b)$---can have their prefixes merged together if
they have equivalent labels. This operation can  be seen as constructing an
automaton combining two branches of a conditional.

Concretely, the operator $\mergeop$ takes two automata, $\fsm_1$ and $\fsm_2$,
and returns a new automaton that accepts the union of the traces. We formalize
 $\mergeop$  and its preconditions below:

\begin{definition}\label{def:merge}
  We assume the two automata $\fsm_1,\fsm_2$ are of the form $\fsm_i =
  \tuple{Q_i,\delta_i,\labelp_i,\labele_i}$ with the initial and final nodes
  $\qentry_i,\qexit_i$. Suppose there is a prefix of statements
  $\stmt_1,\ldots,\stmt_n$ such that
  \begin{enumerate}
    \item every path from $\qentry_1$ to $\qexit_1$ is of the form
      $
        \qentry_1 \lto{\stmt_1} q_{1,1} \lto{\stmt_2} q_{1,2} \ldots q_{1,n}\lto{\assume(b)}q_{1,n+1} \cdots \qexit_1 .
      $
    \item every path from $\qentry_2$ to $\qexit_2$ is of the form
      $
        \qentry_2 \lto{\stmt_1} q_{2,1} \lto{\stmt_2} q_{2,2} \ldots q_{2,n} \lto{\assume(\neg b)}q_{2,n+1} \cdots \qexit_2 .
      $
    \item each prefix node $q \in \{\qentry_1,q_{1,1},\ldots,q_{1,n}\}$ has
      equivalent labels ($\labelp$ and $\labele$) to its corresponding node in
      $\{\qentry_2,q_{2,1},\ldots,q_{2,n}\}$.
  \end{enumerate}
  Then, $\fsm_1 \mergeop \fsm_2$ yields a failure automaton $\fsm = \tuple{Q,\delta,\labelp,\labele}$ with
  \begin{itemize}
    \item $Q = Q_1 \cup (Q_2 \setminus \{\qentry_2,q_{2,1},\ldots,q_{2,n},\qexit_2\})$;
    \item $\delta = \delta_1 \cup \delta_2 \cup \{q_i \lto{\stmt} \qexit \mid q_i \lto{\stmt} \qexit_2 \in \delta_2\}$, with all edges to/from undefined nodes removed;
    \item $\qentry = \qentry_1$ and $\qexit = \qexit_1$;
    \item $\labelp$ agrees with $\labelp_1$ and $\labelp_2$, except that $\labelp(\qexit) = \labelp(\qexit_1) \cup \labelp(\qexit_2)$; and
    \item $\labele$ agrees with $\labele_1$ and $\labele_2$, except that
      $\labele(\qexit) = \max(\labele(\qexit_1),\labele(\qexit_2))$.
  \end{itemize}
  More advanced extensions of this operation are also possible, for instance,
  also merging common post-fixes along with common prefixes, but we stick with
  this version for concreteness.
\end{definition}

If two automata are well-labeled and the $\cmerge$ rule applies, then the
resulting merged automaton is also well-labeled.  It is, however, important to
note  condition (3) in~\cref{def:merge}, which states that the shared prefix
between the two automata must have the same labels on both automata.  If that
condition is violated, the result may not be well-labeled, as illustrated in the
following example.

\begin{figure}
  \footnotesize

  \begin{subfigure}[t]{0.5\textwidth}
    \centering
  \includegraphics[scale=1]{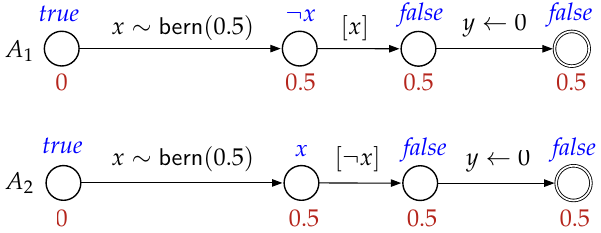}

  (a) Example demonstrating $\mergeop$'s condition (3)
\end{subfigure}
\begin{subfigure}[t]{0.49\textwidth}
  \centering
  \includegraphics[scale=1]{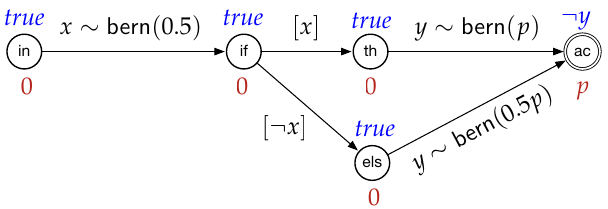}

  (b) Example demonstrating $\mergeop$ on $\mathsf{ex1}$
\end{subfigure}
\caption{Merge examples}\label{fig:merge}
\end{figure}

\begin{example}
  Consider the two well-labeled single-trace automata $\fsm_1$ and $\fsm_2$ in~\cref{fig:merge}(a), which model a conditional statement and share the prefix $x \sim \bern(0.5)$.
  The annotations prove that both traces satisfy $\hoare{0.5}{\true}{A_i}{y>0}$.
  The operation $\mergeop$ does not apply here, since the automata disagree on the label of the second node.
  However, suppose that we apply $\mergeop$ nonetheless.
  This results in a final node with failure probability $\max(0.5,0.5) = 0.5$.
  But this is not sound, since the probability of failing to achieve $y > 0$
  is $1$ when both traces are considered together, since
  both traces set $y$ to $0$.
\end{example}

We also give an example of a sound application of merge.
\begin{example}
  Consider the two well-labeled automata $\trace_1$ and $\trace_2$ from~\cref{fig:ex1} in~\cref{sec:ex1}.
  They satisfy the conditions for $\mergeop$.
  \cref{fig:merge}(b) shows the result of applying $\mergeop$ to these two automata.
  Notice that the accepting node, denoted $\texttt{ac}$,
  has a label $\labele(\texttt{ac}) = \max(p,0.5p)$, which is equal to $p$.
\end{example}
\begin{lemma} \label{lem:sound-merge}
  If $\fsm_1,\fsm_2$ are well-labeled and satisfy the
  $\mergeop$ conditions, then $\fsm = \fsm_1 \mergeop
  \fsm_2$ is well-labeled.
\end{lemma}
% \begin{proof}
%   By definition of well-labeled automaton.
% \end{proof}

\paragraph{Termination.}
Finally, the rule $\ccorr$ gives the termination condition for the algorithm,
corresponding to the conditions from \cref{thm:proofrule}.  Notice that
precondition and postcondition inclusion hold by construction, since they were
ensured by the labeling function $\labeltrace$ when the
first trace in each automaton was added to the automaton set by rule $\ctrace$.

\subsection{Theoretical Properties}\label{sec:theory}

\paragraph{Soundness.}
As expected, the algorithm is sound.

\begin{theorem}[Soundness of algorithm] \label{thm:sound-alg}
  If $\ccorr$ applies, then $\hoare{\errp}{\pre}{\prog}{\post}$
  is valid.
\end{theorem}

\paragraph{(In)completeness.}
Our approach is generally incomplete.
The incompleteness primarily results from the application of the \emph{union bound}, which, in some programs, does not allow us to prove the tightest possible failure probabilities.
E.g., consider
$\hoare{0.75}{\true}{x \sim \bern(0.5); y \sim \bern(0.5)}{x\land y}.$
Any well-labeled automaton will upper bound the failure probability by 1,
since we have no means of assuming independent sampling in both statements.
This example can be handled by coalescing the two sampling statements into a single statement; however, the general issue arises in loops, too.

Nevertheless, we can compare the expressivity of our approach with two existing
techniques: the union bound logic~\citep{DBLP:conf/icalp/BartheGGHS16} and
trace abstraction~\citep{heizmann2009refinement,heizmann2013software}.

\paragraph{Union Bound Logic.}
The \emph{union bound logic}~\citep{DBLP:conf/icalp/BartheGGHS16} is an
extension of Hoare logic with failure probabilities, where Hoare triples are
analogous to our annotations $\hoare{\errp}{\pre}{\prog}{\post}$.
Our notion of well-labeled automata can capture
proofs in the union bound logic with the exception of a few points, and our
algorithm can recover a precise class of well-labeled automata.
 We formalize  this correspondence and prove a completeness result in \smref{smith2019:arxiv}{app:ahl}.

\paragraph{Trace Abstraction.}
Our technique generalizes trace abstraction for non-probabilistic,
single-procedure programs~\citep{heizmann2009refinement,heizmann2013software}.
When given a non-probabilistic program $\prog$ and Hoare triple
$\{\pre\}\ \prog\ \{\post\}$, we can construct trace-abstraction
proofs by simply setting the failure probability upper bound to $0$ in the specification.
Consequently, the failure probability labels of nodes of all automata in $\fsms$ must be $0$ for the proof to hold.
In this setting, the state labels ($\labelp$) are overapproximations of reachable states at a specific node, corresponding to the annotations of \emph{Floyd--Hoare automata} defined by~\citet{heizmann2013software}.

%% file: interpolants.tex
%!TEX root=paper.tex

\section{Labeling Individual Traces: Proofs and Interpolation} \label{sec:interpolation}
In the algorithm we presented in \cref{alg:overall}, the key subroutine is the
$\labeltrace$ operation for rule \ctrace.
Recall that given a single trace $\trace$, pre- and post-conditions $\pre$ and
$\post$, and failure probability $\errp$, $\labeltrace$ attempts to construct a
well-labeled automaton $\fsm_\trace$ for $\trace$ proving
$\hoare{\errp}{\pre}{\trace}{\post}$.
We now show how to reduce this task to a constraint-solving
problem.
Our approach is inspired by \emph{interpolation-based verification}~\cite{mcmillan2006lazy}, where the
semantics of $\trace$ are encoded as a formula in first-order logic to check if
it can falsify a Hoare triple.
If the trace does not falsify the triple, Craig interpolants are computed along
the trace forming a Hoare-style annotation.
However, our setting is richer: we need to \rone handle traces with
probabilistic semantics and \rtwo construct two kinds of annotations---sets of
states and failure probability expressions.
We demonstrate how to reduce this problem to Craig interpolation over a
first-order theory, thus eliminating probabilistic reasoning.
We summarize our approach below:
\begin{enumerate}
% \item \textbf{Probabilistic satisfiability:}
% First, we encode the problem of checking $\hoare{\errp}{\pre}{\trace}{\post}$
% as a \emph{probabilistic satisfiability} problem, where we have a formula with logical and probabilistic quantifiers. That is, we construct a formula of the form $\forall X \ldotp \pr[\varphi] \leq \beta$.

\item \textbf{Axiomatizing distributions:}
We demonstrate how to encode $\hoare{\errp}{\pre}{\trace}{\post}$ as a logical formula.
The key challenge is in encoding semantics of sampling statements.
We address this challenge by observing that we can encode sampling statements by introducing appropriate \emph{logical axioms} about the  distributions.
This results in a \emph{constraint-based synthesis} problem of the form $\exists f \ldotp \forall X \ldotp \varphi$,
where discovering a function $f$ amounts to finding an appropriate axiom for
each sampling statement in order to establish correctness of the trace.

\item \textbf{Craig interpolation:}
Once we have solved the synthesis problem by finding a solution for $f$,
we are left with a valid logical formula of the form
$\forall X \ldotp\varphi$, which we can use to compute interpolants using standard techniques.
We demonstrate that
these interpolants can be converted to a well-labeling of $\fsm_\trace$.
\end{enumerate}

\begin{table}[t]
  \caption{Example families of distribution axioms ($\var$ does not occur in dist. expr.)}
  \label{tbl:axioms}
  \footnotesize
  \begin{tabular}{llll}
    \toprule
    Statement & Assumption $\axasm$ & Upperbound $\axub$ & Parameters \\
    \midrule
    Bernoulli: $\var \sim \bern(\var')$
    &
    $(f(\varsi) = 1 \land \var) \lor (f(\varsi) = 2 \land \neg \var)$
    &
    $\begin{array}{lr}
      \var' & \text{if }f(\varsi) = 1\\
      1 - \var' & \text{if }f(\varsi) = 2\\
      0 & \text{otherwise}
    \end{array}$
    &
    $f(\varsi) \in \{1,2,3\}$
    \\
    Uniform: $\var \sim \unif(\var')$
    &
    $\var \in f(\varsi)$
    &
    $|f(\varsi)|/|\var'|$
    &
    $f(\varsi) \subseteq \var'$
    \\
    Laplace: $\var \sim \lap(\var_1,\var_2)$
    &
    $|\var - \var_1| > \var_2 \log\left(\frac{1}{f(\varsi)}\right)$
    &
    $f(\varsi)$
    &
    $f(\varsi) \in (0,1]$
    \\
    Exponential: $\var \sim \olap(\var_1,\var_2)$
    &
    $\var < \var_1 \lor v - \var_1 > \var_2 \log\left(\frac{2}{f(\varsi)}\right)$
    &
    $f(\varsi)$
    &
    $f(\varsi) \in (0,1]$
    \\
    \bottomrule
  \end{tabular}
\end{table}

\subsection{Proofs via Distribution Axiomatization}
We now describe how we can check validity of the specification $\hoare{\errp}{\pre}{\stmt_1;\cdots;\stmt_n}{\post}$.
Our approach is analogous to logical encodings of program paths in verification of non-probabilistic programs; there, each statement $\stmt_i$ is encoded as a formula $\varphi_{\stmt_i}$ in some appropriate first-order theory, e.g., the theories of linear arithmetic or arrays.
Novel to our setting, we use \emph{distribution axioms} to approximate the
semantics of sampling statements in a first-order theory.

\paragraph{Logical Theory.}
We assume that deterministic program expressions correspond to a first-order theory, like linear arithmetic.
Given a formula $\varphi$, a model $\model$ of $\varphi$, denoted $\model \models \varphi$, is a valuation of its free variables $\fv(\varphi)$ satisfying the formula---e.g.,  $\model \models x + y > 0$ where $\model = \{x \mapsto 0, y \mapsto 1\}$.
We use $\model(\varphi)$ to denote $\varphi$ with all free variables replaced by their interpretation in $\model$.
A formula $\varphi$ is \emph{satisfiable} if there exists $\model$ such that $\model \models \varphi$; a formula is \emph{valid} if $\model \models \varphi$ for all models $\model$.

% To eliminate the probabilistic variables, we \emph{axiomatize} the distributions.
% We begin by describing what we mean by axioms and then demonstrate how to reduce the probabilistic satisfiability problem to solving a \emph{synthesis} problem of the form $\exists f \ldotp \forall X \ldotp \varphi$.

\paragraph{Distribution Axioms.}
Given a sampling statement $\var \sim \dexpr$,
an axiom is of the form
\[
\pr_{\var \sim \dexpr} [\axasm] \leq \axub
\]
where $\axub$ is a [0,1]-valued expression over $\vars$ and $\axasm$ is a formula over $\vars$.
The axiom must be true for all possible valuations of the program variables $\vars \setminus \{\var\}$.
We will use the axioms as follows:
When encoding the effect of a sampling statement $\var \sim \dexpr$,
we can assume that $\neg \axasm$ is true, with a failure probability
of at most $\axub$.
This allows us to sidestep probabilistic reasoning and encode program semantics
in our first-order theory.
% %
% \jh{Does this mean that distribution expressions can only be parameterized by
% input variables, not other program expressions? It's fine, as long as we are
% consistent.}
% %
% \aws{$\dexpr$ can talk about any variable}
% %
% \jh{OK, but it seems there is a mismatch if $\axasm$ must be a formula over
%   $\var$ and $\varsi$. For the Laplace axioms, for instance...}
% %

Since axioms are approximations of primitive distributions, there are
many possible axioms for any given distribution. In some cases, axioms
may be parameterized, e.g., by the failure probability. We call parameterized
axioms \emph{axiom families}; \cref{tbl:axioms} collects example axiom families for the distributions in \cref{sec:model}.

\begin{definition}[Laplace axiom family]
  Recall that the (discrete) Laplace distribution expression $\lap(\var_1,\var_2)$
  is parameterized by two parameters,
  the \emph{mean} $\var_1 \in \mathds{Z}$ and the \emph{scale} $\var_2 \in \mathds{R}$.
  Sampling from   $\lap(\var_1,\var_2)$ returns an integer $\var$
  with probability proportional to $\exp(- |\var - \var_1| / \var_2)$.
  The Laplace distribution supports the following family of axioms, parameterized
  by a $(0,1]$-valued function $f$:
  \[
  \pr_{\var \sim \lap(\var_1,\var_2)}\left[|\var - \var_2| > \frac{1}{\var_1} \log\left(\frac{1}{f(\varsi)}\right)\right] \leq f(\varsi)
  \]
  Different instantiations of $f$ yield different axioms.
\end{definition}

The exponential distribution's axiom family is similar; note that
 $\olap(\var_1, \var_2)$ has zero probability of returning elements
smaller than $\var_1$, and this information is incorporated into the axiom.
The Bernoulli distribution's family is parameterized by a function $f(\varsi)$
which decides whether to assume $\var$ is $\true$, $\false$ or treat it non-deterministically.
The uniform distribution's axiom family is parameterized by a function $f(\varsi)$
returning a subset of the set defined by $\var'$.

\begin{example}
Recall trace $\trace_2$ (from program $\mathsf{ex2}$) in \cref{sec:ex2} and \cref{fig:ex2},
which contains the statement $a[i] \sim \lap(q[i], 1/\eps)$.
To prove correctness of $\trace_2$, we instantiated the Laplace
axiom family with $f(\varsi) = p$ where $p \in \varsi$,
yielding the axiom
$
\pr_{a[i] \sim \lap(q[i],1/\eps)}\left[|a[i] - q[i]| > \frac{1}{\eps} \log\left(\frac{1}{p}\right)\right] \leq p.
$
\end{example}

% \jh{May be simpler to replace the $\expr_1, \expr_2, \log$ stuff with an
% uninterpreted function $T_{\lap}(\expr_1, \expr_2, \beta)$. In this case,
% $T_{\lap}$ would be equal to $(1/\expr_1) \log(1/\beta)$, and we could maybe
% include some specific axioms about $T_{\lap}$? I guess $\beta$ would still be
% $f(\varsi)$, and the goal would still be to find $f$.}

%The distribution axioms are sound.

\begin{theorem} \label{thm:soundness-axioms}
  Each axiom in \cref{tbl:axioms} is sound: given any input state $\stt$ and
  well-typed distribution expression $\dexpr$, the probability that $\axasm$
  holds in $\stt(\dexpr)$ is at most $\stt(\axub)$.
\end{theorem}

\paragraph{Logical Encoding.}
We now present our encoding for checking $\hoare{\errp}{\pre}{\trace}{\post}$.
First, without loss of generality, we assume that $\trace$ is in \emph{static single assignment} (\abr{SSA}) form; this ensures that variables are not assigned to more than once, simplifying our encoding.
We also assume that $\pre$ and $\post$ are  logical formulas over program variables.
%
% Now, we use axiom families to
% logically encode sampling statements.
Our encoding explicitly maintains failure probability
using a special set of real-valued variables $\cost_i$, which encode
failure probability after statement $\stmt_i$ along $\trace$.
In order to encode failure probability on
unsatisfiable subtraces, we also use a special set of
Boolean variables $\halt_i$ to track if an execution
was blocked by an assume statement.

The function $\encodeD$, defined in \cref{fig:enc}, is used to encode assignment, assume, and sampling statements; it maintains the variables $\cost_i,\halt_i$ and axiomatizes sampling statements using the aforementioned distribution axioms.
\begin{figure}[t]
  \begin{align*}
    \encodeD(i, \var \gets \expr) &\triangleq \var = \expr \land \cost_i = \cost_{i-1} \land \halt_i = \halt_{i-1}\\
    \encodeD(i, \assume(\bexpr)) &\triangleq \cost_i = \cost_{i-1} \land \halt_i = (\halt_{i - 1} \lor \neg \bexpr) \\
    \encodeD(i, \var \sim \dexpr) &\triangleq \cost_i = \cost_{i-1} + \axub \land \halt_i = (\halt_{i-1} \lor \axasm)
    \qquad \text{given axiom family: } \pr_{\var \sim \dexpr} [\axasm] \leq \axub
  \end{align*}
  \caption{Logical encoding of statement semantics}
  \label{fig:enc}
\end{figure}
Consider, for instance, the encoding for assignment statements:
it constrains $\var$ to $\expr$, while maintaining the same failure probability and blocked status, $\cost_i$ and $\halt_i$.
Intuitively, the semantics of assignment statements is precisely captured by our
logical encoding, so assignment statements do not increase the probability of failure.
In contrast, the probability of failure increases when an axiom is applied
for a sampling statement.  Concretely, if
the axiom family $\pr_{\var \sim \dexpr} [\axasm] \leq \axub$ is applied, we assume that $\neg \axasm$
is true while accumulating probability of failure $\axub$, as encoded in the constraint $\cost_i = \cost_{i-1} + \axub$.

The following theorem formalizes the encoding
of $\hoare{\errp}{\pre}{\trace}{\post}$
and states its correctness.

\begin{theorem}[Soundness of logical encoding]\label{thm:enc}
  The specification
  $\hoare{\errp}{\pre}{\stmt_1,\ldots,\stmt_n}{\post}$
  is valid if the following formula is satisfiable:
  \begin{align}\label{eq:enc}
  \forall \vars, \cost_i, \halt_i \ldotp \left(\pre \land \cost_0 = 0 \land \halt_0 = \false \land \bigwedge_{i=1}^n \encodeD(i,\stmt_i)\right) \Longrightarrow (\cost_{n} \leq \beta \land (\neg \halt_n \Rightarrow \post))
  \end{align}
\end{theorem}

% \jh{Do we have to pay for failure probability for blocked traces?}
% \aws{nth is free, bro}

Observe that in the above encoding the only free symbols are
the uninterpreted functions $f_1,\ldots,f_m$ introduced by the axiom families used in the encoding of sampling statements.
Thus, checking satisfiability involves synthesizing interpretations for $f_1, \ldots, f_m$.
(Equivalently, we can think of $f_1,\ldots,f_m$ as existentially quantified
so that we check validity of $\exists f_1,\ldots,f_m \forall V\ldots$.)

\begin{example}
  Recall the trace $\trace_1$ from \cref{sec:ex1}
  and \cref{fig:ex1} (program $\mathsf{ex1}$), where we proved:
  \[\hoare{p}{\true}{\underbrace{x \sim \bern(0.5)}_{\stmt_1}; \underbrace{\assume(x)}_{\stmt_2}; \underbrace{y \sim \bern(p)}_{\stmt_3}}{\neg y}\]
  Using the encoding in \cref{thm:enc}, we get the following formula:
  \begin{align*}
  \forall x,y,p, \cost_i, \halt_i \ldotp \left(\cost_0 = 0 \land \halt_0 = \false \land
  \bigwedge_{i=1}^3 \encode(i, \stmt_i)
  % \encodeD(1, x \sim \bern(0.5)) \land
  % \encodeD(2, \assume(x)) \land
  % \encodeD(3, y \sim \bern(p))
  \right) \Longrightarrow (\cost_{3} \leq p \land (\neg \halt_3 \Rightarrow \neg y))
\end{align*}
To illustrate, $\encode(1,x \sim \bern(0.5))$ is the following constraint, using
the axiom family in \cref{tbl:axioms}:
\[
\cost_1 = \cost_0 +\underbrace{
\left\{\begin{array}{lr}
   0.5 & \text{if }f_x(p) = 1\\
   0.5 & \text{if }f_x(p) = 2\\
   0 & \text{otherwise}
\end{array}\right\}}_{\axub}
\mathbin{\land}
\left(h_1 = h_0 \lor \underbrace{(f_x(p) = 1 \land x) \lor (f_x(p) = 2 \land \neg x)}_{\neg \axasm}\right)
\]
The proof in \cref{sec:ex1} used
the interpretation $f_x(p) = 3$, allowing $x$ to take any value.
\end{example}

\subsection{From Synthesis to Craig Interpolation}

Now that we have defined our  logical constraints, we can apply
\emph{Craig interpolation} on the above encoding in \cref{thm:enc} to construct the labeling functions, $\labelp$ and $\labele$, for an automaton accepting $\trace$.

\paragraph{Craig Interpolants.}
The standard notion of \emph{sequence interpolants}~\citep{mcmillan2006lazy}
generalizes Craig interpolants between two formulas to a sequence of
unsatisfiable formulas in first-order logic.

\begin{definition}[Sequence interpolants]
  Let $\bigwedge_{i=1}^n \varphi_i$ be unsatisfiable.
  There exists a sequence of formulas $\psi_1,\ldots,\psi_n$
  such that:
  \begin{enumerate}
    \item $\varphi_1 \Rightarrow \psi_1$ and $\psi_n \Rightarrow \false$ are valid,
    \item for all $i \in (1,n)$, $\psi_i \land \varphi_{i+1} \Rightarrow \psi_{i+1}$ is valid, and
    \item $\fv(\psi_i) \subseteq  \fv(\varphi_1,\ldots,\varphi_i) \cap \fv(\varphi_{i+1},\ldots,\varphi_n)$.
  \end{enumerate}
  Note that sequence interpolation is equivalent to
  solving a form of recursion-free \emph{Horn clauses}~\cite{rummer2013classifying};
  we use an interpolation-based presentation to reduce notational overhead.
\end{definition}

\paragraph{Labeling Automata via Interpolation.}
Suppose that we have discovered interpretations for $f_1,\ldots,f_m$
that satisfy \cref{eq:enc} from \cref{thm:enc}.
This implies that the following formula, which is \cref{eq:enc} after negating it and instantiating $f_1,\ldots,f_m$ with their interpretations, is unsatisfiable:
\[
\left(\pre \land \cost_0 = 0 \land \bigwedge_{i=1}^n \encodeD(i,\stmt_i)\right) \land \neg(\cost_{n} \leq \beta \land (\neg\halt_n \Rightarrow \post))
\]
It follows that we can construct a \emph{sequence of Craig interpolants}
for the following problem:
\begin{align*}
 \underbrace{\pre \land \cost_0 = 0}_{\varphi_0} \land
 \underbrace{\encodeD(1,\stmt_1)}_{\varphi_1} \land
\cdots \land
\underbrace{\encodeD(n,\stmt_n)}_{\varphi_{n}} \land
\underbrace{ \neg(\cost_{n} \leq \beta \land (\neg \halt_{n} \Rightarrow \post))}_{\varphi_{n+1}}
\end{align*}
Every interpolant $\psi_i$ encodes the set of reachable
states and the failure probability after executing the first $i$ program statements
beginning from a state in $\pre$. The free-variable condition for interpolants
implies that the only free variables in $\psi_i$ are $\halt_i, \cost_i$, and
live program variables after the $i$th statement.
The challenge is that interpolants describe both the program state invariants
and the failure probability invariants, corresponding to the $\labelp$ and $\labele$
needed to label the failure automaton.
Fortunately, these labels can be extracted from the interpolants.
The following theorem formalizes the transformation and states its correctness.

\begin{theorem}[Well-labelings from interpolants]\label{thm:interpolant}
  Let $\{\psi_i\}_i$ be the interpolants computed as shown above.
  Let $\fsm_\trace = \tuple{Q,\delta,\labelp,\labele}$ be the failure automaton that accepts only the trace
  $\trace = \stmt_1,\ldots,\stmt_n$,
  i.e., $\delta = \{\qentry\lto{\stmt_1} q_1, q_1 \lto{\stmt_2} q_2, \ldots q_{n-1} \lto{\stmt_n} \qexit\}$.
  Set the labeling functions as follows:
  \begin{enumerate}
    \item $\labelp(\qentry) \triangleq \pre$ and $\labele(\qentry) \triangleq 0$.
    \item $\labelp(q_i) \triangleq \exists \cost_i \ldotp \psi_i[\halt_i \mapsto \false]$ and $\labelp(\qexit) \triangleq \exists \cost_n \ldotp \psi_n[\halt_n \mapsto \false]$.
    \item $\labele(q_i) \triangleq f(\varsd)$, where $f(\varsd)$
    is the function that returns, for any valuation of $\varsd$,
    the largest value of $\omega_i$ that satisfies $\exists \vars \setminus \varsd \ldotp \exists \halt_i \ldotp \psi_i$.
    For $\labele(\qexit)$, we use $\exists \vars \setminus \varsi \ldotp \exists \halt_n \ldotp \psi_n$.
  \end{enumerate}
  Then, $\fsm_\trace$ is well-labeled and implies $\hoare{\errp}{\pre}{\trace}{\post}$.
\end{theorem}
Notice that for $\labelp$ we set $\halt_i$ to be $\false$, since we are
only interested in states that pass assume statements (reachable states). We
existentially quantify $\cost_i$, as it is not a program variable.
Also notice the technicality in constructing $\labele$; this arises because the interpolant is a \emph{relation} over values of $\cost_i$ and $\varsd$, while the label of $\labele(q_i)$ is technically a \emph{function} from $\varsd$ to [0,1].
In practice, we need not construct the function $f$; we can perform all needed checks using relations.

% \begin{example}
%   continue previous example
% \end{example}

%% file: evaluation.tex
%!TEX root=paper.tex

\input{benchmarks}

\section{Implementation and Case Studies}\label{sec:evaluation}

\subsection{Overview of Implementation}

We have implemented our approach atop the Z3 \abr{SMT} solver~\citep{demoura08}.
We encode statements using the following first-order theories: linear arithmetic, uninterpreted functions, and arrays.
Below, we describe our implementation;
we refer to \smref{smith2019:arxiv}{app:implementation} for further details.

\paragraph{Algorithmic Strategy.}
Our implementation is a determinization of the algorithm presented in \cref{sec:algorithm}.
To ensure that we get tight upper bounds on failure probability, our implementation aggressively tries to apply the $\cmerge$ rule---recall that the $\cmerge$ rule allows us take the maximum failure probability across two automata, instead of the sum.
Specifically, we modify the rule $\ctrace$ to return a set of traces $\trace_1, \ldots, \trace_n \in \lang(\prog) \cap \overline{\lang(\fsms)}$.
Then, we attempt to simultaneously label all traces with the same interpolants at nodes pertaining to the same control location.
To ensure that we compute similar interpolants across traces, we use the same distribution axiom for the same sampling instruction in all traces it appears in.
Finally, we attempt to apply the rule $\cgen$ to create cycles in the resulting automaton.

\paragraph{Discovering Axioms.}
Given a formula of the form $\exists f \ldotp \forall X \ldotp \varphi$,
we check its validity using a propose-and-check loop:
\rone we propose an interpretation of $f$ and then \rtwo check if $\forall X \ldotp \varphi$ is valid with that interpretation using the \abr{SMT} solver (more on this below).
The first step proposes interpretations of $f$ of increasing size,
e.g., for a unary function $f(x)$, it would try $0,1,x,x+1$, etc.
As we shall see, even for complex randomized algorithms from the literature, the required axioms are syntactically simple, so this simple strategy works rather well.

\paragraph{Checking Validity.}
The case studies to follow make heavy use of non-linear arithmetic (e.g., $\frac{x \cdot y}{z} + u > 0$) and transcendental functions (namely, $\log$).
Non-linear theories are generally undecidable.
To work around this fact, we implement an incomplete formula validity checker using an eager version of the \emph{theorem enumeration} technique recently proposed by~\citet{Srikanth2017}.
First, we treat non-linear operations as uninterpreted functions, thus overapproximating their semantics.
Second, we strengthen formulas by instantiating \emph{theorems} about those non-linear operations. For instance, the following theorem relates division and multiplication: $\forall x,y \ldotp y > 0 \Rightarrow \frac{x \cdot y}{y} = x$.
We then instantiate $x$ and $y$ with terms over variables in the formula. Since there are infinitely many possible instantiations of $x$ and $y$,
we restrict instantiations to terms of size 1, i.e., variables/constants.

Our implementation uses a fixed set of theorems about multiplication, division, and logarithms. These are instantiated for every given formula, typically resulting in $\sim$1000 additional conjuncts.

\paragraph{Interpolation Technique.}
Given the richness of the theories we use, we found that existing proof-based interpolation techniques either do not support the theories we require (e.g., the MathSAT solver~\citep{cimatti2013mathsat5}) or fail to find generalizable interpolants, e.g., cannot discover quantified interpolants (e.g., Z3).
As such, we implement a \emph{template-guided} interpolation technique~\citep{albarghouthi2013beautiful,rummer2013exploring},
where we force interpolants to follow syntactic forms that appear in the program.
Specifically, for every Boolean predicate $\varphi$ appearing in the program, the specification, or the axioms, we create a template $\varphi^t$ by replacing its variables with placeholders, denoted $\wild_i$. For instance, given $x > y$, we generate the template $\wild_1 > \wild_2$.

Given a set of templates, our interpolation technique searches for an interpolant as a conjunction of instantiations of those templates, where each placeholder can be replaced by a well-typed term over formula variables.
Given the infinite set of possible instantiations, our implementation fixes the size of possible instantiations (e.g., to size 1), and proceeds by finding the smallest possible interpolants in terms of number of conjuncts.
If it cannot, it expands the search to terms of larger sizes.

\subsection{Case Studies in Privacy-Preserving Algorithms}
\emph{Differential privacy}~\citep{DMNS06} is a strong probabilistic property
modeling statistical data privacy. Informally, a randomized database query
satisfies differential privacy if adding/removing a single individual's
private data to/from the database does not change the output distribution very much, so that
differentially private algorithms reveal little about any single individual's record.
To achieve this property, algorithms add random noise at key points in the computation.
Sophisticated differentially private algorithms are known for a wide
variety of common data analyses, and differential privacy is starting to see
deployments in both industry~\citep{erlingsson2014rappor,johnson2018towards} and
government~\citep{abowdschmutte,haney2017utility}.

Intuitively, more random noise yields stronger privacy guarantees at the expense
of accuracy---the noisy answers may be too far from the exact answers to be of any practical use.
Therefore, the designer of a differentially private algorithm aims to maximize
accuracy of the computed results while achieving some target level of privacy.
We now consider a number of algorithms from the
differential privacy literature and demonstrate our technique's ability to
automatically prove their accuracy guarantees.  The algorithms and their
specifications are shown in \cref{fig:bench} and described below;
\cref{tbl:results} provides runtime and other statistics, which we discuss  later in this section.

%\aws{add a one-line motivation for each algorithm}

\paragraph{Randomized Response ($\rresponse$).}
One of the oldest randomized schemes for protecting privacy is \emph{randomized
response}, proposed by~\citet{warner1965randomized} decades before the
formulation of differential privacy. In the typical setting, an individual has a
single bit (0 or 1) as their private data, representing e.g. the
presence of some disease or genetic marker.  Under randomized response, the
individual flips two fair coins: if the first result is heads, the individual
reports their bit honestly, otherwise they ignore their private bit and report
the result of the second flip. In this way, randomized response guarantees a degree
of privacy by introducing plausible deniability---an individual's reported bit
could have been the result of chance. At the same time, randomized response guarantees
a weak notion of accuracy, as the output is biased towards the true private bit with
probability $3/4$.

\begin{table}[t]
  \centering
  \caption{Results on private algs. PA: \# of proposed axioms; TI: \# of theorem instantiations; time is in sec.}
  \label{tbl:results}
  \small
  \begin{tabular}{lllll}
    \toprule
    Algorithm & Axiom(s) synthesized & PA & TI & Time \\
    \midrule
    $\rresponse$ & $\emph{priv} \iff \mathit{snd}(r)$ & 162 & 0 & 2\\
    $\noisysum$ & $\ifrac{\len{\qs}}{p}$ & 5 & 5496 & 98\\
    $\rnm$ & $\ifrac{\len{\qs}}{p}$ & 4 & 1768 & 33 \\
    $\expmech$ & $\ifrac{\len{R}}{p}$ & 3 & 1768 & 27 \\
    $\abovet$ & $\ifrac{2}{p}$ and  $\ifrac{2\len{\qs}}{p}$ & 22 & 752 & 23 \\
    $\sparsevec$ & $\ifrac{3}{p}$, $\ifrac{3\len{\qs}}{p}$, and $\ifrac{3}{p}$ & 941 & 1330 & 97 \\
    \bottomrule
  \end{tabular}
\end{table}

We encode randomized response as the first program in \cref{fig:bench} and prove
the accuracy guarantee. In the code, $\mathit{priv}$ is the individual's private
bit. The program draws two bits uniformly and then decides what
to return; $\mathit{fst}$ and $\mathit{snd}$ extract the results of the first
and second bits, respectively. The accuracy guarantee states that
the returned answer is equal to the true private bit, except with probability at
most $1/4$.
Our implementation synthesizes the axiom $\mathit{priv} \iff \mathit{snd}(r)$; this ensures that the second bit has the same value as $\mathit{priv}$,
so if the first bit is 0 and the else branch is taken, the algorithm is forced to return the right result, with a failure probability of $1/4$.

\paragraph{Noisy Sum ($\noisysum$).}
Our next algorithm computes the sum of a set of numeric queries, adding noise to
the answer of each query in order to ensure differential privacy.
This is a simplified version of the private counters by~\citet{CSS10}
and~\citet{DNPR10}, which are used to publish aggregate statistics privately, e.g., total number of website visitors.

The $\noisysum$ program
takes three inputs: a set $\qs$ of integer-valued queries, a database $\db$ holding
the private data, and a parameter $\epsilon \in \mathds{R}$ representing the desired level of
privacy.\footnote{%
  We encode $\qs$ as an integer array where index $i$ holds the result of $\qs_i(\db)$.}
The program populates an integer array $a$ with
answers to each query, with noise drawn from the Laplace distribution
with scale controlled by $\epsilon$; smaller $\epsilon$ is more private, but
requires more noise. Finally, the output is the sum of all noisy answers.

The accuracy guarantee bounds how far the noisy sum deviates from the true sum
with failure probability $p$, where $p$ is a parameter. Our implementation synthesizes an axiom for each Laplace sampling, setting the failure
probability to be $p/\len{\qs}$ each time.
Therefore, at step $i$, $|a_i - \qs_i(\db)| \leq \frac{1}{\eps}\log(|\qs|/p)$.
Since there are
$\len{\qs}$ iterations,
after the loop exits
we have $|s - \sum_j^{|\qs|} \qs_j(\db)| \leq \frac{|Q|}{\eps}\log(|Q|/p)$
with a failure probability of at most $|Q| \cdot \frac{p}{|Q|} = p$.

\paragraph{Report Noisy Max ($\rnm$) and the Exponential Mechanism ($\expmech$).}
Our next pair of algorithms select an approximate maximum element from a set of private data.

In Report Noisy Max~\citep{dwork2014algorithmic}, $\rnm$ in \cref{fig:bench}, the algorithm is presented with a set
$\qs$ of integer queries, a private database $\db$, and a privacy level
$\epsilon$. The algorithm then evaluates each query on $\db$ and adds
Laplace random noise to protect privacy.
Finally, the index of the query with
the largest noisy value is returned.
For example, if each query counts the
number of patients with a certain disease, then Report Noisy Max will report a
disease that may not be true most prevalent disease, but whose count is not too
far from the true maximum count.

The postcondition states that the
answer of the returned query $\qs_\best$ is not too far below the answer of the actual
maximum query. To achieve failure probability $p$, our implementation synthesizes an axiom for the Laplace sampling statement with failure probability $p/\len{\qs}$.
Since the loop executes $\len{\qs}$ times, we establish
that the postcondition holds with probability $p$.
To do so, the interpolation engine discovers a number of key facts; we outline two of them:
$$\forall j \in [1,i) \ldotp |a_j - \qs_j(d)| \leq  \frac{2}{\epsilon}\log \frac{\len{\qs}}{p} \quad \text{ and } \quad \forall j \in [1,i) \ldotp a_\best \geq a_j$$
The first formula specifies that, for every element of $j$ of $a$,
its distance from the corresponding valuation of $\qs_j(\db)$ is bounded above by $2/\epsilon\log \len{\qs}/p$---this follows directly from the choice of distribution axiom.
The second formula states that the best element is indeed larger than all previously seen ones.
Upon loop exit, these facts, along with others, are sufficient to imply the postcondition.
Notice that the
$2/\epsilon\log \len{\qs}/p$
in the first formula weakens to $4/\epsilon\log \len{\qs}/p$
in the postcondition. This is due to the two-sided error introduced by the absolute value in the Laplace axiom.
The proof computed for $\rnm$ is presented in detail in \smref{smith2019:arxiv}{app:rnm}.

The algorithm $\expmech$ is a discrete version of the seminal Exponential
mechanism~\citep{MT07}, a fundamental algorithm in differential privacy.  This
algorithm is used to achieve differentially privacy in non-numerical queries, as well as a mechanism for achieving certain notions of fairness in decision-making algorithms~\citep{dwork2012fairness}.
$\expmech$ takes a set $R$ of possible output elements, a \emph{utility
function} $\futil$ mapping each element of $R$ and private database to a numeric
score, a private database $d$, and privacy parameter $\epsilon$. The
algorithm aims to return an element of $R$ that has large utility on the given
database. $\expmech$ differs from $\rnm$ through the use of the exponential
distribution; because the exponential distribution never produces results lower
than the shift, the accuracy bound for the $\expmech$ is better. The distance to the
true maximum is at most $\frac{2}{\eps}\log(2|R|/p)$ instead of $\frac{4}{\eps}\log(|Q|/p)$, with
failure probability at most $p$.  To prove this, our implementation
synthesizes an axiom analogous to that used for $\rnm$.
% The main difference compared to $\rnm$ is the use of the exponential distribution.  Because the exponential distribution never
% produces results lower than the shift, the accuracy bound for the exponential
% mechanism is a bit better: the distance to the true maximum is at most
% $\frac{2}{\eps}\log(2|R|/p)$ instead of $\frac{4}{\eps}\log(|Q|/p)$, with
% failure probability at most $p$.  To prove this, our implementation
% synthesizes an axiom analogous to that used for $\rnm$.

\paragraph{Above Threshold ($\abovet$) and the Sparse Vector Mechanism ($\sparsevec$).}
A useful differential privacy primitive is to return the first query in a
list with a numeric answer (approximately) above some given threshold,
ignoring queries with small answers. Our final two privacy examples
do just this. The Above Threshold algorithm~\citep{dwork2014algorithmic}
takes a list $\qs$ of queries, a private database $d$, a numeric threshold $T$,
and the target privacy level $\epsilon$. First, the program computes a noisy
threshold $t$ by adding noise to the true threshold $T$.  The program loops
through the queries, comparing the noisy answer of each query to the
noisy threshold. If the noisy answer is above the noisy threshold, the program
sets the flag $\mathit{done}$ and exits the loop. Finally, the algorithm
returns the index of the approximately above threshold query, or a default value
$\bot$ if no such query was found.

The accuracy guarantee requires some care. There are two
cases: the returned value is either a query index, or $\bot$.  In the first
case, $q_{\ansv}$ should have true value not too far below the exact threshold
$T$, and all prior queries should have true value not too far above $T$. In the
second case, no query was found to be above threshold after adding noise, so no true
answer should be too far above $T$. To prove this property, we synthesize axioms for the Laplace sampling instructions with
different failure probabilities: $\frac{p}{2}$ for the threshold sampling, and
$\frac{p}{2\len{\qs}}$ for each loop sampling. There is one threshold
sampling and at most $\len{\qs}$ loop iterations, so the total failure probability
is at most $\frac{p}{2} + \len{\qs} \cdot \frac{p}{2 \len{\qs}} = p$.

A slightly more involved variant of this algorithm, called Numeric Sparse
Vector~\citep{dwork2014algorithmic}, also returns a noisy answer to the above
threshold query along with the query's index. Again, the accuracy property
describes the two cases---above threshold query found, and no above threshold
queries. In both cases, the noisy query answer should be close to the true
answer. The proof proceeds much like in the simpler variant, adjusting the
failure probabilities when applying axioms in order to take the additional noisy
answer sampling into account.

\paragraph{Discussion of Results.}
\cref{tbl:results} summarizes the results of applying our implementation to the above algorithms.
The table lists the synthesized axiom per sampling statement---recall that our implementation strategy forces different instances of a sampling statement to use the same axiom.
Additionally, we list the number of proposed and checked axioms (PA),\footnote{%
  PA does not include the many possible axiom instantiations that are not well-typed.}
the largest number of theorem instantiations for dealing with non-linear arithmetic (TI), and the total time in seconds.

Consider the $\abovet$ algorithm. The implementation attempts 22 different pairs (because there are two sampling statements) of axioms. \cref{tbl:results} lists the synthesized interpretation of the function $f(\varsi)$ for the first and second sampling statements. The implementation discovers the axiom that assigns a failure probability $p/2$ for the first sampling statement and $p/(2|Q|)$ for the second sampling statement.
Proving accuracy of $\abovet$ takes 23 seconds and 752 theorems are instantiated to interpret non-linear arithmetic.
Notice that $\noisysum$ takes the longest amount of time, even though it only attempts 5 axioms. This is due to the large number ($\sim$5500) of theorem instantiations.
For $\sparsevec$, the implementation proposes 941 axioms before discovering the shown axioms.

To the best of our knowledge, no existing tools can automatically reason about the algorithms and accuracy properties we have discussed here.
The algorithms we considered are small yet sophisticated.
As the number of sampling statements increases, the space of possible axioms grows combinatorially, impacting synthesis performance. As research into constraint-based program synthesis progresses, our approach can directly benefit from these developments.

\newcommand{\rely}{\emph{unrel}}
\subsection{Case Study in Unreliable Hardware}
To demonstrate our approach's versatility, we consider another possible
application: analyzing programs executing on approximate hardware, which is unreliable but efficient.

% \begin{wrapfigure}{r}{0.5\textwidth}
%   \vspace{-.1in}
% \begin{algorithmic}[0]
%     \State $\mathit{nblocks,height,width} \gets 20,16,16$
%     \Function{$\searchref$}{$\mathit{pblocks}$, $\mathit{cblock}$}
%         \State {\textbf{reliable}: $i, j,k$}
%         \State {\textbf{unreliable}: $\mathit{minssd}, \mathit{minblock}, \emph{ssd}, t, t_1, t_2$}

%         \State{$i \gets 0$}
%         \While{$i < \mathit{nblocks}$}
%           \State{$ssd, j \gets 0,0$}
%             \While{$j < \mathit{height}$}
%               \State $k \gets 0$
%               \While {$ k < \mathit{width}$}
%                 \State{$t_1, t_2 \gets \mathit{pblocks}[i, j], \mathit{cblock}[j]$}
%                 \State{$t \gets t_1 \mathbin{\dot{-}} t_2$}
%                 \State{$ssd \gets ssd \mathbin{\dot{+}} (t \mathbin{\dot{\times}} t)$}
%                 \State $k \gets k + 1$
%               \EndWhile
%               \State{$j \gets j + 1$}
%             \EndWhile
%             \If{$\mathit{sdd} \mathbin{\dot{<}} \mathit{minssd}$}
%                 \State{$\mathit{minssd}, \mathit{minblock} \gets \emph{ssd},i$}
%             \EndIf
%             \State{$i \gets i + 1$}
%         \EndWhile
%         \State \Return $\mathit{minblock}$
%     \EndFunction
% \end{algorithmic}
% \caption{Reliability ex.~\citep{carbin2013verifying}}
% \label{fig:rely}
% \vspace{-.1in}
% \end{wrapfigure}
%
We use the program $\searchref$ from the \emph{Rely}
system by \citet{carbin2013verifying}, shown in \cref{fig:rely}, which
implements a pixel-block search algorithm from x264 video encoders. The program
receives a constant number of pixel blocks $(\mathit{nblocks} = 20)$ of
size $16 \times 16$ ($\mathit{height} \times \mathit{width}$).

\begin{wrapfigure}{r}{0.5\textwidth}
  \vspace{-.1in}
\begin{algorithmic}[0]
    \State $\mathit{nblocks,height,width} \gets 20,16,16$
    \Function{$\searchref$}{$\mathit{pblocks}$, $\mathit{cblock}$}
        \State {\textbf{reliable}: $i, j,k$}
        \State {\textbf{unreliable}: $\mathit{minssd}, \mathit{minblock}, \emph{ssd}, t, t_1, t_2$}

        \State{$i \gets 0$}
        \While{$i < \mathit{nblocks}$}
          \State{$ssd, j \gets 0,0$}
            \While{$j < \mathit{height}$}
              \State $k \gets 0$
              \While {$ k < \mathit{width}$}
                \State{$t_1, t_2 \gets \mathit{pblocks}[i, j], \mathit{cblock}[j]$}
                \State{$t \gets t_1 \mathbin{\dot{-}} t_2$}
                \State{$ssd \gets ssd \mathbin{\dot{+}} (t \mathbin{\dot{\times}} t)$}
                \State $k \gets k + 1$
              \EndWhile
              \State{$j \gets j + 1$}
            \EndWhile
            \If{$\mathit{sdd} \mathbin{\dot{<}} \mathit{minssd}$}
                \State{$\mathit{minssd}, \mathit{minblock} \gets \emph{ssd},i$}
            \EndIf
            \State{$i \gets i + 1$}
        \EndWhile
        \State \Return $\mathit{minblock}$
    \EndFunction
\end{algorithmic}
\caption{Reliability ex.~\citep{carbin2013verifying}}
\label{fig:rely}
\vspace{-.1in}
\end{wrapfigure}
This program is expected to provide adequate video encoding despite potential hardware failures.
Rely's programming model exposes unreliable arithmetic operations, denoted with
a dot (e.g. $\dot{+}, \dot{-}$, etc.), which may fail with small
probability (say, $10^{-7}$).
Reading from variables typed as \emph{unreliable} may also fail with a small probability.
Rely assumes loops over unreliable data have a constant bound on the
number of iterations, so these loops can be
unrolled.

Our goal is to prove the probability of a reliable execution is at least 0.99, where reliability implies no failures along the execution.\footnote{Rely multiplies this probability by the reliability of the inputs ($\mathit{pblocks}$, $\mathit{cblock}$)---this does not impact the analysis.}
To do so, we analyze a version of the program instrumented with a Boolean flag $\rely$,
which is initialized to $\false$. We model each unreliable operation by adding a
sampling from the Bernoulli distribution to determine whether the operation
fails.
For instance, a read $y \gets x$ from an unreliable $x$ is transformed into
$y \gets x; \rely \gets \rely \lor \bern(10^{-7})$.
We then use our implementation to prove $\hoare{0.01}{\true}{\searchref}{\rely = \false}$.

Unlike Rely, we do not assume independent failures.
Our analysis thus gives a more conservative estimate of failure probability, but, as a benefit, retains
soundness even if failures are correlated.
% Our method gives a more conservative (larger) estimate of failure probability,
% but as a benefit, our analysis remains sound even if failures can be correlated.
Nevertheless, we are able to prove that the program is reliable with probability $\geq 0.992832$, compared to the $0.994885$ computed by Rely.
Moreover, since our approach is symbolic, we can prove a \emph{symbolic} reliability bound as a function of the number of blocks and their size.
This allows us to ask: \emph{how many blocks can we use, and how large, and still be reliable?}
We automatically establish the parameterized failure probability
$1.4 \cdot 10^{-6} \cdot \mathit{nblocks} \cdot \mathit{height} \cdot \mathit{width}$,  describing how program parameters affect reliability.
For instance, we can increase the number of blocks to 25 and still maintain $\geq 0.99$
reliability, or quadruple the size of each block to $32^2$ pixels and get
$\geq .97$ reliability.
In both settings, our approach completes the proof in less than 2 seconds.

%% file: benchmarks.tex
\definecolor{brickred}{rgb}{0.8, 0.25, 0.33}
\newcommand{\spec}[1]{{\color{blue} \left\{#1\right\}}}
\newcommand{\fprob}[1]{{\color{brickred} #1}}
\newcommand{\qs}{Q}
\newcommand{\db}{d}
\newcommand{\len}[1]{|#1|}
\newcommand{\eval}{\mathsf{eval}}

\newcommand{\sumv}{s}
\newcommand{\maxv}{\mathit{max}}
\newcommand{\best}{\mathit{b}}
\newcommand{\donev}{\mathit{done}}
\newcommand{\ansv}{\mathit{ans}}
\newcommand{\noisy}{\mathit{noisy}}
\newcommand{\futil}{\mathit{u}}

% case study algorithms
\newcommand{\rresponse}{\mathsf{randResp}}
\newcommand{\rnm}{\mathsf{noisyMax}}
\newcommand{\noisysum}{\mathsf{noisySum}}
\newcommand{\abovet}{\mathsf{aboveT}}
\newcommand{\sparsevec}{\mathsf{sparseVec}}
\newcommand{\expmech}{\mathsf{expMech}}
\newcommand{\searchref}{\mathsf{searchRef}}

\begin{figure}[t]
  \centering
  \smaller
  \begin{minipage}[c]{.25\textwidth}
\begin{algorithmic}
    \State $\vdash_\fprob{0.25} \spec{\true}$
    \Function{$\rresponse$}{$\emph{priv}$}
        \State $r \sim \textsf{unif}(\{ 0, 1 \}^2)$
        \If{$\mathit{fst}(r) = 1$}
            \State $\ansv \gets \emph{priv}$
        \Else
            \State $\ansv \gets \mathit{snd}(r)$
        \EndIf
        \State{\Return $\ansv$}
    \EndFunction
    \State $\spec{\ansv \iff \emph{priv}}$
\end{algorithmic}
\end{minipage}
\begin{minipage}[c]{.4\textwidth}
  \begin{algorithmic}
      \State $\vdash_\fprob{p} \spec{\eps > 0}$
      \Function{$\rnm$}{$\qs, d, \epsilon$}
          \State $\best, \emph{max}, i \gets \bot, \bot, 1$
          \State $a \gets \mathds{Z}[\len{\qs}]$
          \While{$i \leq \len{\qs}$}
              \State $q \gets \qs_i(d)$
              \State $a_i \sim \lap(q, \frac{2}{\eps})$
              \If{$a_i > \emph{max} \lor \best = \bot$}
                  \State $\best \gets i$
                  \State $\emph{max} \gets a_i$
              \EndIf
              \State $i \gets i + 1$
          \EndWhile
          \State \Return $\best$
      \EndFunction
      \State $\spec{\forall j \in [1,\len{\qs}] \ldotp \qs_\best(\db) \geq \qs_j(\db) - \frac{4}{\epsilon}\log \frac{\len{\qs}}{p}}$
  \end{algorithmic}
\end{minipage}
\begin{minipage}[c]{.3\textwidth}
  \begin{algorithmic}
      \State $\vdash_\fprob{p} \spec{\eps > 0}$
      \Function{$\abovet$}{$\qs, \db, T, \epsilon$}
      \State $i, \donev \gets 1, \mathit{false}$
          \State $t \sim \lap(T,\frac{1}{\eps})$
          \While{$i \leq \len{\qs} \land \neg \donev$}
              \State $q \gets \qs_i(\db)$
              \State $a \sim \lap(q,\frac{2}{\eps})$
              \If{$a > t$}
                  \State{$\donev \gets \mathit{true}$}
              \EndIf
              \State $i \gets i + 1$
          \EndWhile
          \State $\ansv \gets \donev \mathbin{?} i - 1 : \bot$
          \State \Return $\ansv$
      \EndFunction
      \State $\spec{%
          (\ansv \neq \bot \Rightarrow \varphi_{\top}) \land
          (\ansv = \bot \Rightarrow \varphi_{\bot})
      }$

  \end{algorithmic}
\end{minipage}
%
%%%%%%%%%%% second row
\vspace{.1in}

\begin{minipage}[c]{.3\textwidth}
  \begin{algorithmic}
      \State $\vdash_\fprob{p} \spec{\eps > 0}$
      \Function{$\noisysum$}{$\qs, \db, \epsilon$}
          \State $\sumv \gets 0$
          \State $i \gets 1$
          \State $a \gets \mathds{Z}[\len{\qs}]$
          \While{$i \leq \len{\qs}$}
              \State $q \gets \qs_i(\db)$
              \State $a_i \sim \lap(q, \frac{1}{\eps})$
              \State $\sumv \gets \sumv + a_i$
              \State $i \gets i + 1$
          \EndWhile
          \State \Return $\sumv$
      \EndFunction
      \State $\spec{|\sumv - \sum_{j=1}^{\len{\qs}} \qs_j(\db)|
        \leq \frac{\len{\qs}}{\eps} \log\frac{\len{\qs}}{p}}$
  \end{algorithmic}
\end{minipage}
\begin{minipage}[c]{.36\textwidth}
  \begin{algorithmic}
    \State $\vdash_\fprob{p} \spec{\eps > 0}$
    \Function{$\expmech$}{$R,\futil,\db,\eps$}
    \State $r_\best \gets \bot$
    \State $\maxv \gets 0$
    \For {$r \in R$}
    \State $\mathit{util} \gets \futil(r, d)$
    \State $\mathit{nu} \gets \olap(\mathit{util},\frac{2}{\eps})$
    \If {$\mathit{nu} > \maxv \lor r_\best = \bot$}
    \State $r_\best \gets r$
    \State $\maxv \gets \mathit{nu}$
    \EndIf
    \EndFor
    \State \Return $r_\best$
    \EndFunction
    \State $\spec{\forall j \in R \ldotp \futil(r_\best,\db) \geq \futil(j,\db) - \frac{2}{\epsilon}\log \frac{2\len{R}}{p}}$
  \end{algorithmic}
\end{minipage}
\begin{minipage}[c]{.33\textwidth}
  \begin{algorithmic}
      \State $\vdash_\fprob{p} \spec{\eps > 0}$
      \Function{$\sparsevec$}{$\qs, \db, T, \epsilon$}
          \State $i, \donev \gets 1, \mathit{false}$
          \State $t \sim \lap(T,\frac{1}{\eps})$
          \While{$i \leq \len{\qs} \land \neg \donev$}
              \State $q \gets \qs_i(\db)$
              \State $a \sim \lap(q,\frac{2}{\eps})$
              \If{$a > t$}
                  \State $\noisy \sim \lap(q,\frac{1}{\eps})$
                  \State{$\donev \gets \mathit{true}$}
              \EndIf
              \State $i \gets i + 1$
          \EndWhile
          \State $\ansv \gets \donev \mathbin{?} (i - 1, \noisy) : \bot$
          \State \Return $\ansv$
      \EndFunction
      \State $\spec{%
          (\ansv \neq \bot \Rightarrow \varphi_{\top}') \land
          (\ansv = \bot \Rightarrow \varphi_{\bot}')
      }$
  \end{algorithmic}
\end{minipage}

$
  \varphi_{\top} \triangleq
  \begin{cases}
    \forall j \in [1,\ansv).\, \qs_j(\db) \leq T +
    \frac{2}{\eps}\log\frac{2\len{\qs}}{p} + \frac{1}{\eps}\log\frac{2}{p}
    \\
    \qs_{\ansv}(\db) \geq T - \frac{2}{\eps}\log\frac{2\len{\qs}}{p} - \frac{1}{\eps}\log\frac{2}{p}
  \end{cases}
$
~~
$
  \varphi_{\bot} \triangleq \forall j \in  [1,\len{\qs}].\, \qs_j(\db) \leq T +
  \frac{2}{\eps}\log\frac{2\len{\qs}}{p} + \frac{1}{\eps}\log\frac{2}{p}
$
\[
\text{We elide } \varphi'_\top \text{ and } \varphi'_\bot, \text{which are defined similarly.}
\]

\caption{Privacy-preserving algorithms and their accuracy specifications}
\label{fig:bench}
\end{figure}

%% file: relatedwork.tex
%!TEX root=paper.tex

\section{Related Work}\label{sec:relatedwork}

\paragraph{Interpolation \& Trace Abstraction.}
% Craig interpolants made their way into verification through
% McMillan's seminal work on interpolation-based model checking.
% In \abr{SAT}-based model checking~\citep{biere2018sat}, a finite unrolling of the transition system is used to prove absence of bugs of bounded length.
% Craig interpolants act as an
% over-approximation of the set of reachable states after $n$ iterations,
% with the hope that $n$ generalizes.
%
In software verification, interpolants were first used for
constructing predicate abstract domains in \emph{counterexample-guided
abstraction refinement} (\abr{CEGAR}).~\citep{henzinger2004abstractions}.
McMillan's work on lazy abstraction with interpolants~\citep{mcmillan2006lazy}
used proofs of correctness of program traces to directly construct Hoare-style annotations.
By unrolling the program's \abr{CFG} into a tree and adding annotations, he showed
how to generalize a tree of paths into an automaton by adding back edges, proving the correctness
of infinitely many traces.

Our approach is inspired by work by
\citet{Heizmann10,heizmann2013software,heizmann2009refinement}, which provided
an insightful and general view of interpolation-based verification through the
lens of automata. Elegance aside, the automata view better suits our probabilistic
setting: in McMillan's original formulation, a program can be unrolled into a tree, as paths
with common prefixes can be combined. In our setting this combining is more subtle---our $\mergeop$ is
a restricted version. Further, the automata view allows us to maintain sets of traces separately and
sum up their probabilities of failure.

\paragraph{Deductive Probabilistic Verification.}
Deductive verification techniques for probabilistic
programs include probabilistic Hoare
logics~\citep{RandZ15,BEGGHS16,Hartog:thesis,Chadha07}.
and the lightweight
probabilistic logic of~\citet{DBLP:conf/icalp/BartheGGHS16} our technique is closely related to, just as the
classical interpolation-based techniques mirror Hoare-style proofs.
Deductive techniques are highly expressive, but the complex proofs typically must be
constructed manually or in an interactive setting. In contrast, our approach has
the advantage of automation.

\paragraph{Pre-Expectation Calculus.} The pre-expectation calculus and associated predicate transformers \citep{Kozen:1985, Morgan:1996} can prove properties of probabilistic programs, but have practical obstacles to full automation. Computing pre-expectations across sampling instructions yields an integral over the distribution being sampled from. Complex distributions, like the infinite-support Laplace distribution, yield correspondingly complex integrals that are difficult to reason about. Any automation of the pre-expectation calculus will need to establish algebraic properties about these mathematical expressions. Our use of distribution axioms (\cref{sec:interpolation}) obviates the need to reason directly about integrals via a reduction to synthesis.

\paragraph{Martingales.} Martingales---probabilistic analogues of loop invariants---are used in automated tools to prove termination conditions \citep{Chakarov-martingale,Chatterjee2016,Chatterjee:2016:AAQ:2837614.2837639,Chatterjee:2017:SIP:3009837.3009873,mciver2016new} and properties of expected values \citep{BEFFH16}. Automated martingale synthesis techniques are restricted to linear or polynomial invariants, which alone are unable to prove the accuracy properties we are interested in.

\paragraph{Probabilistic Model Checking.}
Probabilistic model checking is perhaps the
most well-developed technique for automated reasoning about probabilistic
systems.  Traditionally, it focused on
temporal properties of \emph{Markov
Decision Processes} (\abr{MDP})---surveys by \citet{kwiatkowska2010advances}
and \citet{Katoen:2016:PMC:2933575.2934574} overview the current
state of the art.

Our program model can be cast as an infinite-state \abr{MDP},
with non-determinism at program entry to pick an initial state.
There have been a number of abstraction-based techniques for reducing the
size of large (or infinite) \abrs{MDP}  ~\citep{kattenbelt,Kattenbelt2010,hermanns08}.  To our
knowledge, existing works cannot handle the programs and properties we
consider here. The general limitation is the inability of existing
 model checking techniques to handle distribution
expressions---e.g., a Laplace whose scale is a parameter---and failure
probabilities that are expressions.  Probabilistic
\abr{CEGAR}~\citep{hermanns08} uses a guarded-command language where probabilistic choice is a
real-value determining the probability of executing each command.  Other
techniques limit distribution expressions to
finite distributions with constant parameters~\citep{kattenbelt}.

%We are aware of one work that generalizes interpolation to a probabilistic
%setting:
\citet{teige2011generalized} consider interpolation in \emph{stochastic
Boolean satisfiability}~\citep{littman2001stochastic},
where formulas contain existential and probabilistic quantifiers. The approach
has been used for generalizing bounded encodings of finite-state \abrs{MDP}, in
an analogous fashion to the original work on interpolation-based model
checking~\citep{mcmillan2003interpolation}.

\paragraph{Other Probabilistic Analyses.}
Probabilistic abstract interpretation~\citep{cousot2012probabilistic} generalizes the
abstract interpretation framework to a probabilistic setting; other
techniques can be cast in this
framework~\citep{monniaux2000abstract,monniaux2001backwards,monniaux2005abstract,claret2013bayesian}.
Recently, \citet{wangpmaf} presented \abr{PMAF}, an elegant algebraic framework
for constructing analyses of probabilistic programs.  The approach is
rather general, accepting recursive programs and supporting interprocedural
analyses. Unlike \abr{PMAF}, whose results depend highly on the expressiveness of the chosen abstract domain, our technique constructs abstractions on demand, \emph{\`a la} interpolation-based verification, at the risk of never generalizing.  \abr{PMAF} instantiations considered by \citet{wangpmaf} cannot
prove our target accuracy properties, but alternative instantiations might achieve something similar.

Another line of work reduces probabilistic verification to a form of
counting~\citep{Albarghouthi17,chistikov15,belle15,mardziel2011dynamic}.  To
compute the probability that a formula is \abr{SAT}, these techniques count the number
of satisfying assignments---or perform numerical volume
estimation in the infinite-state case.  While these techniques can compute very
precise---often exact---probabilities, they target simpler program models.
Specifically, programs have no inputs, probability distribution are not
parameterized, and loops are handled via unrolling.

Our technique is related to works verifying
relational probabilistic properties, including
differential privacy and
uniformity~\citep{AlbarghouthiH18,AlbarghouthiH18b}. These systems encode the space of
\emph{coupling proofs} as a constraint-based synthesis problem. Our technique handles
different properties, but shares the high-level design principle of reducing probabilistic reasoning to logical reasoning.

Computer algebra  and symbolic
inference methods (e.g.,
\citep{gehr2016psi,cusumano2018incremental,narayanan2016probabilistic}) have
 been applied to probabilistic programs in different domains (e.g.,
\citep{gehr2018bayonet}). While these tools can automatically generate symbolic
representations of output distributions, proving
properties about these distributions remains challenging. Modern implementations
use a variety of custom heuristics and reduction strategies to try to simplify
complex algebraic terms, a computationally-expensive task.

%% file: conclusion.tex
\section{Conclusions and Future Directions} \label{sec:conclusions}
We have presented a generalization of trace abstraction for proving
accuracy properties of probabilistic programs.
This required four key ideas: \rone representing
probabilistic traces with failure automata labeled with formulas and probabilities, \rtwo merging and generalizing
these failure automata, \rthree axiomatizing distributions by
solving a  synthesis problem, and \rfour applying Craig interpolation
to construct labels of failure automata. These
ideas enable automated verification of accuracy properties using logic-based techniques, while handling rich programs and properties.

For future work, we see trace abstraction modulo probability being extended with other kinds of
probabilistic reasoning, perhaps based on independence or expectations; the challenge is keeping
the complexity of the logical encoding under control.
Another natural path is to connect to recent work on probabilistic abstract interpretation, e.g. by~\citet{wangpmaf}. One could imagine enhancing a standard abstract domain with failure probabilities. However,
the treatment of sampling instructions is unclear---there are typically multiple axioms for a given
distribution, and a proof may need several different axiomatizations to prove a target property.

%% file: appendix.tex
%!TEX root=paper.tex

\section{Omitted Proofs} \label{app:proofs}

\subsection{Proof of \cref{thm:gproofrule}}
\begin{proof}
  By definition of the semantics, the output distribution of a program $\prog$
  on input state $\stt$ is
  \[
    \sem{\prog}(\stt) = \sum_{\trace \in \lang(\prog)} \sem{\trace}(\stt) .
  \]
  Hence for any input state $\stt \in \pre$, we have
  \[
    \sem{\prog}(\stt)(\overline{\post}) = \sum_{\trace \in \lang(\prog)} \sem{\trace}(\stt)(\overline{\post})
    \leq \sum_{\trace \in \lang(\fsms)} \sem{\trace}(\stt)(\overline{\post})
    \leq \stt(\errp)
  \]
  by the trace inclusion and failure probability upperbound conditions.
\end{proof}

\subsection{Proof of \cref{thm:proofrule}}
\begin{proof}
  Let $\stt \in \pre$ be any input state satisfying the pre-condition. For
  each automaton $\fsm_i$, the pre-condition inclusion condition implies that
  $\stt \in \pre \subseteq \labelp_i(\qentry_i)$ and so
  \cref{thm:labeled-sound} gives
  \[
    \sum_{\trace \in \lang(\fsm_i)} \sem{\trace}(\stt)(\overline{\labelp_i(\qexit_i)})
    \leq \stt(\labele_i(\qexit_i)) .
  \]
  By the post-condition inclusion property, we also have
  $\overline{\post} \subseteq \overline{\labelp_i(\qexit_i)}$ and so
  \[
    \sum_{\trace \in \lang(\fsm_i)} \sem{\trace}(\stt)(\overline{\post})
    \leq \stt(\labele_i(\qexit_i)) .
  \]
  Finally we can conclude by the trace inclusion and failure probability
  upperbound conditions:
  \[
    \sem{\prog}(\stt)(\overline{\post})
    = \sum_{\trace \in \lang(\prog)} \sem{\trace}(\stt)(\overline{\post})
    \leq \sum_{i = 1}^n \sum_{\trace \in \lang(\fsm_i)} \sem{\trace}(\stt)(\overline{\post})
    \leq \sum_{i = 1}^n \stt(\labele_i(\qexit_i))
    \leq \stt(\errp) .
  \]
\end{proof}

\subsection{Proof of \cref{thm:labeled-sound}}
\begin{proof}
  We first consider the simpler case when $\fsm$ has no directed loops. In such
  an automaton, the valuation of the deterministic variables $\varsd$ at any
  node $q_i$ is the same for all execution traces starting at $\qentry$ with
  initial state $\stt_0$ and reaching $q_i$; we write $v_i$ for these
  valuations, and we write $v_{in}$ and $v_{ac}$ for these valuations at
  $\qentry$ and $\qexit$, respectively.

  We need to work with a slightly more general version of well-labeled automata,
  where the initial and final nodes are labeled by a function of the
  deterministic variables $\varsd$. We show that for any initial state $\stt_0
  \in \labelp(\qentry)$, we have
  \[
    \sum_{\trace \in \lang(\fsm)} \mu(\overline{\labelp(\qexit)})
    \leq v_{ac}(\labele(\qexit)) - v_{in}(\labele(\qentry)) .
  \]
  where $\mu = \sem{\trace}(\stt_0)$ is the output distribution.  Note that when
  $\labele(\qentry) = 0$ and $\labele(\qexit)$ is labeled by input variables
  $\varsi$ only, we recover:
  \[
    \sum_{\trace \in \lang(\fsm)} \mu(\overline{\labelp(\qexit)})
    \leq \stt_0(\labele(\qexit)) .
  \]
  The proof is by induction on the number $k$ of branches (i.e., nodes with two
  outgoing e

  In the base case $k = 0$, the automaton represents a sequential composition
  $\stmt_1; \cdots; \stmt_n$. Let the corresponding nodes be $q_0, \dots, q_n$,
  with $q_0 = \qentry$ and $q_n = \qexit$. Since the probability labels
  $\labelp(q_i)$ depend on deterministic variables only, given any initial state
  $s_0 \in \labele(q_0)$ there is a sequence of valuations $v_0, \dots, v_n$ for
  the deterministic variables such that the deterministic variables $\varsd$ of
  any state with non-zero probability in $\sem{\stmt_1; \cdots ;
  \stmt_{i}}(s_0)$ are set to $v_i$, with $v_0 = s_0(\varsd)$. By the
  well-labeled condition, we have:
  \[
    \hoare{v_i(\labele(q_i)) - v_{i - 1}(\labele(q_{i - 1})) }
    {\labelp(q_{i - 1}) \land \varsd = v_{i - 1}}
    {\stmt_i}
    {\labelp(q_i) \land \varsd = v_i}
  \]
  By the sequential composition rule of the union bound logic, we have
  \[
    \hoare{v_n(\labele(\qexit)) - v_0(\labele(\qentry))}
    {\labelp(\qentry) \land \varsd = v_0}
    {\stmt_1 ; \cdots ; \stmt_n}
    {\labelp(\qexit)}
  \]
  By definition, $v_n = v_{ac}$ and $v_0 = v_{in}$ so we have
  \[
    \hoare{v_{ac}(\labele(\qexit)) - v_{in}(\labele(\qentry))}
    {\labelp(\qentry) \land \varsd = v_{in}}
    {\stmt_1 ; \cdots ; \stmt_n}
    {\labelp(\qexit)}
  \]
  and we conclude by soundness of the union bound logic.

  Now, suppose there are $k > 0$ branches in $\fsm$. Starting from the initial
  node $\qentry$, let the first branching node be $q_r$ with outgoing
  edges to $q_t$ and $q_f$, labeled by $\assume(\bexpr)$ and $\assume(\neg
  \bexpr)$ respectively. We let $\fsm_0$ be the sub-automaton with
  initial node $\qentry$ and final node $q_r$; note that this automaton is a
  single path along nodes $q_0 = \qentry, q_1, \dots, q_r$ with edge labels
  $\stmt_1, \dots, \stmt_r$. Letting $\mu_r = \sem{A_0}(\stt_0)$ be the output
  distribution of this automaton, the base case yields
  \[
    \sum_{\trace \in \lang(\fsm_0)} \mu_r(\overline{\labelp(q_r)})
    \leq v_r(\labele(q_r)) - v_{in}(\labele(\qentry)) .
  \]
  Now, we consider the rest of the automaton.  Let $\fsm_t$ be the sub-automaton
  of all reachable nodes starting from $q_t$, and let $\fsm_f$ be the
  sub-automaton starting from $q_f$. Note that $\fsm_t$ and $\fsm_f$ are both
  well-labeled automata with entry nodes $q_t$ and $q_f$ respectively, and have
  at most $k - 1$ branching nodes each.  Since the $\assume$ statements do not
  modify this variables, $v$ is also the deterministic valuation of $\varsd$ at
  the entry nodes of $q_t$ and $q_f$. By induction, for any state $\stt \in
  \labelp(q_b)$ such that $\stt(\varsd) = v_r$ and $b \in \{ t, f \}$ we have
  \[
    \sum_{\trace \in \lang(\fsm_b)} \sem{\trace}(\stt)(\overline{\labelp(\qexit)})
    \leq v_{ac}(\labele(\qexit)) - v_r(\labele(q_b))
  \]
  To combine our bounds for $\fsm_0, \fsm_t, \fsm_f$ back together, we assume
  that the state labels at the branching node $q_r$ satisfy
  \[
    \labelp(q_r) \subseteq \labelp(q_t) \cap \{ \stt \mid \stt(\bexpr) \}
    \qquad\text{and}\qquad
    \labelp(q_r) \subseteq \labelp(q_f) \cap \{ \stt \mid \stt(\neg \bexpr) \} .
  \]
  If either fails, then the edge condition for well-labeled automata ensures
  that $v_{ac}(\qexit) - v_r(q_r) \geq v_r(q_b) - v_r(q_r) \geq 1$ and so
  $v_{ac}(\qexit) - v_{in}(\qentry) \geq 1$, and our target bound is trivial.
  Now, every trace in $\lang(\fsm)$ is of the form $q_0, \dots, q_r, q_b, \dots,
  \qexit$ for $b = t$ or $b = f$; since $\fsm$ has no loops, the trace after
  $q_r$ is entirely contained in $\fsm_b$.

  Now, we decompose $\mu_r = \mu_t + \mu_f + \mu_{\mathit{err}}$ into three pieces:
  \begin{itemize}
    \item $\mu_{\mathit{err}}$ is the restriction to states not in $\labelp(q_r)$;
    \item $\mu_{t}$ is the restriction to states in $\labelp(q_r)$ with
      $\bexpr$ is true;
    \item $\mu_{f}$ is the restriction to states not in $\labelp(q_r)$ with
      $\bexpr$ false.
  \end{itemize}
  Note that all states in the support of $\mu_t$ and $\mu_f$ lie in
  $\labelp(q_t)$ and $\labelp(q_f)$, respectively. Since $\fsm_0, \fsm_t,
  \fsm_f$ are all loop free with at most $k - 1$ branches, applying the
  induction hypothesis gives
  \begin{align}
    \sum_{\trace \in \lang(\fsm)} \sem{\trace}(\stt_0) (\overline{\labelp(\qexit)})
    &\leq \dbind(\mu_t, \sem{\fsm_t})(\overline{\labelp(\qexit)})
    + \dbind(\mu_f, \sem{\fsm_f})(\overline{\labelp(\qexit)})
    + |\mu_{\mathit{err}}|
    \tag{semantics}
    \\
    &\leq |\mu_t| \cdot (v_{ac}(\labele(\qexit)) - v_r(\labele(q_t)))
    + |\mu_f| \cdot (v_{ac}(\labele(\qexit)) - v_r(\labele(q_f)))
    \notag
    \\
    &+ (v_r(\labele(q_r)) - v_{in}(\labele(\qentry)))
    \tag{IH}
    \\
    &\leq |\mu_t| \cdot (v_{ac}(\labele(\qexit)) - v_r(\labele(q_r)))
    + |\mu_f| \cdot (v_{ac}(\labele(\qexit)) - v_r(\labele(q_r)))
    \notag
    \\
    &+ (v_r(\labele(q_r)) - v_{in}(\labele(\qentry)))
    \tag{$\labele(q_r) \leq \labele(q_b)$}
    \\
    &= (|\mu_t| + |\mu_f|) \cdot v_{ac}(\labele(\qexit))
    + (1 - |\mu_t| - |\mu_f|) \cdot v_r(\labele(q_r))
    - v_{in}(\labele(\qentry))
    \notag
    \\
    &\leq v_{ac}(\labele(\qexit)) - v_{in}(\labele(\qentry))
    \tag{$\labele(q_r) \leq \labele(\qexit)$}
  \end{align}
  Finally, we consider the general case where $\fsm$ may have directed loops.
  The basic idea is to reduce to the acyclic case we have just considered by
  performing finite unrollings of $\fsm$. The argument uses standard
  constructions on automata and regular expressions (see, e.g., prior work
  giving an algebraic view of program schemes~\citep{AngusKozen}); we just
  sketch the proof
  here. Let $C$ be the set of all statements appearing in $\fsm$. We can view
  $\fsm$ as a deterministic automaton $D$ over the alphabet $\Sigma = Q \times Q
  \times C$ by viewing each transition $q_i \lto{\stmt} q_j$ as a transition on
  letter $(q_i, q_j, \stmt)$. To make this a deterministic automaton, we can add
  a new dead node $q_{\mathit{dead}}$ with a self loop on all letters, and add a
  transition from every existing node $q \in Q$ to $q_{\mathit{dead}}$ on all letters
  that don't appear as outgoing transitions from $q$ in $\fsm$. Then, we mark
  $\qexit$ as the sole accepting node in $D$. Now, the language $\lang_D$
  accepted by $D$ is evidently equal to the language $\lang(\fsm)$ of all traces
  in $\fsm$.

  By Kleene's theorem, this language can also be represented as a regular
  expression $R$ over $\Sigma$. Now, we can define finite unrollings in terms of
  $R$. For $n \in \mathbb{N}$, let $R_n$ be the regular expression obtained by
  repeatedly replacing each subterm $r^*$ where $r$ is star-free by the finite
  approximation $1 + r + \cdots + r^n$; the order of replacement will not matter
  for our purposes. Now $\lang(R) = \cup_n \lang(R_n)$, and $\lang(R_i)
  \subseteq \lang(R_j)$ for all $i \leq j$. Again by Kleene's theorem, the
  language of each $R_n$ is recognized by a deterministic finite automaton; let
  $D_n$ be a minimal automaton for each $R_n$.

  Now since the language of $R_n$ is finite and $D_n$ is minimal, the only
  cycles in $D_n$ must occur as self-loops on a single (non-accepting) dead
  node $p_n$. All transitions from the initial node to non-dead nodes must be
  labeled by $(\qentry, -, -)$. There are at most two such transitions since
  there are at most two transitions out of $\qentry$ in the original automaton
  $\fsm$, and if there are two transitions they must be of the form $(\qentry,
  q_t, \assume(\bexpr))$ and $(\qentry, q_f, \assume(\neg \bexpr))$. By a
  similar inductive argument, each non-dead node has at most two outgoing
  transitions to non-dead nodes and if there are two transitions, they are of
  the form $(q, q_t, \assume(\bexpr))$ and $(q, q_f, \assume(\neg \bexpr))$.
  Thus, we can associate each node $p$ in $D_n$ with a node $a(p)$ in $\fsm$
  and convert $D_n$ to a well-labeled acyclic automaton $\fsm_n$ by labeling
  $\labelp(p) \triangleq \labelp(a(p))$ and $\labele(p) \triangleq
  \labele(a(p))$ and removing the dead node; note that $\lang(\fsm_n) =
  \lang(D_n) = \lang(R_n)$.

  Finally, let $\stt$ be any initial state in $\labelp(\qentry)$. By reduction
  to the acyclic case, we have
  \[
    \sum_{\trace \in \lang(\fsm_n)}
    \sem{\trace}(\stt)(\overline{\labelp(\qexit)}) \leq \stt(\labele(\qexit))
  \]
  for every $n \in \mathbb{N}$. Since the left-hand side is increasing in $n$
  and bounded above by $\stt(\labele(\qexit))$, the limit exists and we have
  \[
    \lim_{n \to \infty} \sum_{\trace \in \lang(\fsm_n)}
    \sem{\trace}(\stt)(\overline{\labelp(\qexit)}) \leq \stt(\labele(\qexit)) .
  \]
  But since $\lang(\fsm_n)$ is increasing and $\cup_n \lang(\fsm_n) =
  \lang(\fsm)$, we conclude
  \[
    \sum_{\trace \in \lang(\fsm)}
    \sem{\trace}(\stt)(\overline{\labelp(\qexit)}) \leq \stt(\labele(\qexit)) .
  \]
\end{proof}
%%%%%%%%%%%%%%%%%%%%%%%%%%%%%%%%
\subsection{Proof of \cref{thm:sound-alg}}
\begin{proof}
  We show by induction on the derivation of rules used by the algorithm that the
  automaton set $\fsms$ is always well-labeled, and each automaton in $\fsms$
  satisfies the pre- and post-condition inclusion properties in
  \cref{thm:proofrule}.  The base case, rule $\cinit$, is trivial. Each trace
  added to the automaton set by $\ctrace$ is well-labeled by construction, and
  the simplification rules $\cgen$ and $\cmerge$ keep the automaton set
  well-labeled by definition (\cref{lem:sound-merge}). Finally, if the
  termination rule $\ccorr$ fires, then the automata are well-labeled and
  satisfy pre- and post-condition inclusion properties by induction, and the
  side-conditions guarantee the trace inclusion and failure probability
  upperbound conditions. Therefore by \cref{thm:proofrule}, the accuracy
  judgment is valid.
\end{proof}

%%%%%%%%%%%%%%%%%%%%%%%%%%%%%
\subsection{Proof of \cref{thm:enc}}
We first begin by proving the following lemma, which captures correctness of the encoding of $\trace$.
Specifically, the following lemma formalizes the correspondence between models of the encoding and the support of the output distribution of $\trace$:
we show that for any initial state $\stt$, the models of the logical encoding correspond to a set of states $R_\stt$ and a failure probability $c$ such that
$\sem{\trace}(\stt)(\overline{R_\stt}) \leq c$.

\begin{lemma}[Soundness of $\encodeD$]\label{lem:enc}
Fix trace $\trace = \stmt_1;\cdots;\stmt_n$.
Let \[\varphi \triangleq \cost_0 = 0 \land \halt_0 = \false \land \bigwedge_{i=1}^n \encodeD(i,\stmt_i)\]
where all uninterpreted functions resulting from distribution axiom families have been given a fixed interpretation.
Fix a state $\stt \in \stts$.
Let $\model_1,\ldots,\model_m$ be the set of models of $\varphi$
such that $\stt(\varsi) = \model_i(\varsi)$ and $\model_i(\halt_n) = \false$, for all $i\in [1,m]$.
Let \[R_\stt = \{\stt' \mid  \stt'(\vars) = \model_i (\vars) \}\]
Then, for any $\model_s \models \varphi$ such that $\model_s(\varsi) = \stt(\varsi)$, we have $\sem{\trace}(\stt)(\overline{R_\stt}) \leq \model_\stt(\cost_n)$.
\end{lemma}

\begin{proof}
  First, we note that  all models $\model_s \models \varphi$
  such that $\model_s(\varsi) = s(\varsi)$ agree on the value of $\omega_i$.
  This is because the constraints $\omega_i$ are functions of $\varsi$.
  Second, note that by construction, there is always $\model_s \models \varphi$---i.e., it is never unsatisfiable.

  We proceed by induction on the length of $\trace$.
  For $n = 1$, we have three cases. Fix an $\stt$ as in lemma statement.
  \begin{itemize}
  \item \emph{Case 1:} $\trace = \var \gets \expr$.

  We have $\varphi \triangleq \var = \expr \land \cost_1 = 0 \land \halt_1 = \false$ (after simplification).
  %Fix an $\model$, $R_\stt$, and $\stt$ (as in lemma statement).
  From $\varphi$, $\model_s(\cost_1) = 0$.
  Therefore, lemma states:
  $\sem{\var \gets \expr}(\stt)(\overline{R_\stt}) \leq 0$.
  Suppose this does not hold, then, by definition of $\sem{\var \gets \expr}$,
  there is a state $\stt' \in \stts \setminus R_\stt$ such that
  $\stt' = \stt[v \mapsto \stt(\expr)]$.
  However, by definition of $\varphi$, $\stt' \in R_\stt$,
  since $\model_i(\vars) = \stt'(\vars)$, for some $i$.

  \item \emph{Case 2:} $\trace = \assume(\bexpr)$.

  We have $\varphi \triangleq \cost_1 = 0 \land (\halt_1 =  \neg \bexpr)$ (after simplification).
  From $\varphi$, $\model_s(\cost_1) = 0$.
  Therefore, lemma states:
    $\sem{\assume(\bexpr)}(\stt)(\overline{R_\stt}) \leq 0$.
    Suppose this does not hold, then, by definition of $\sem{\assume(\bexpr)}$,
    we have $\stt \in \stts \setminus R_\stt$ and $\stt(\bexpr) = \true$.
    However, by definition of $\varphi$, $\stt \in R_\stt$,
    iff $s(\bexpr) = \true$.

  \item \emph{Case 3:} $\trace = \var \sim \dexpr$.

  We have $\varphi \triangleq  \cost_1 = \axub \land \halt_1 = \axasm$ (after simplification).
  Lemma states that
  $\sem{\var \sim \dexpr}(\stt)(\overline{R_\stt}) \leq \model_s(\axub)$.
  This follows from the definition of a distribution axiom: that
  $\pr_{\var \sim \dexpr}[\axasm] \leq \axub$ is true for any valuation of $\vars \setminus \{\var\}$.
  \end{itemize}

Assume that lemma holds for traces of length $n$.
We show that it also holds for $n+1$,
where $\trace'$ is a trace of length $n$.

\begin{itemize}
\item \emph{Case 1:} $\trace = \trace'; \var \gets \expr$.

The encoding is $\varphi \triangleq \varphi' \land \var = \expr \land \cost_{n+1} = \cost_{n} \land \halt_{n+1} = \halt_{n}$.

Let $\model'_1, \ldots, \model_m'$ and $R_\stt'$ be defined for $\varphi'$ and $\trace'$, as per lemma statement.
By hypothesis, $\sem{\trace'}(s)(\overline{R_\stt'}) \leq M_s(\omega_n)$.
By semantics of assignment, we have
$\sem{\trace'; \var \gets \expr}(s)(\overline{X}) \leq \model_s(\omega_{n+1})$,
where $X = \{\stt \mid \stt' \in R_\stt', \stt = \stt'[\var \gets \stt'(\expr)]\}$.
Observe that $X = R_\stt$:
by definition of $\varphi$, its models are a subset of $\{\model_1',\ldots,\model_m'\}$ such that $\var = \expr$.
It then follows that $\sem{\trace}(\stt)(\overline{R_\stt}) \leq \model_s(\cost_{n+1})$.

% By hypothesis:
% $\sem{\trace}(\stt)(\overline{R_\stt}) \leq \model_i(\cost_n)$,
% where $\model$, $R_\stt$, and $\stt$ are from lemma statement for $\trace'$.

\item \emph{Case 2:} $\trace = \trace'; \assume(\bexpr)$.

The encoding is $\varphi \triangleq \varphi' \land \cost_n = \cost_{n-1} \land (\halt_n = (\halt_{n - 1} \lor \neg \bexpr))$.
Let $\model'_1, \ldots, \model_m'$ and $R_\stt'$ be defined for $\varphi'$ and $\trace'$, as per lemma statement.
By hypothesis, $\sem{\trace'}(s)(\overline{R_\stt'}) \leq M_s(\omega_n)$.
We know that models $\{\model_i\}$ of $\varphi$ are a subset of $\{\model'_i\}$
where $\bexpr$ is $\true$.
Therefore $\overline{R_\stt} \supseteq \overline{R_\stt'}$.
But we have that all states in $\overline{R_\stt} \setminus \overline{R_\stt'}$
are those were $\bexpr$ is false.
By definition of $\sem{\assume(\bexpr)}$, all those states are assigned probability 0.
Therefore, $\sem{\trace}(\stt)(\overline{R_\stt}) \leq \model_s(\cost_{n+1})$.

\item \emph{Case 3:} $\trace = \trace'; \var \sim \dexpr$.

The encoding is $\varphi \triangleq \varphi' \land \cost_{n+1} = \cost_n + \axub \land \halt_{n+1} = (\halt_n \lor \axasm)$.
Let $\model'_1, \ldots, \model_m'$ and $R_\stt'$ be defined for $\varphi'$ and $\trace'$, as per lemma statement.
By hypothesis, $\sem{\trace'}(s)(\overline{R_\stt'}) \leq M_s(\omega_n)$.
Let $X$ be the set of all states that satisfy $\neg \axasm$.
From $\varphi$, we know that $R_\stt = R_\stt' \cap X$.
By the union bound and the distribution axiom,
$\sem{\trace'; \var \sim \dexpr}(\stt)(\overline{R_\stt'} \cup \overline{X}) \leq M_s(\omega_{n+1})$

% We know that models $\{\model_i\}_i$ of $\varphi$ are a subset of $\{\model'_i\}_i$ such that $\neg \axasm$ is true.

\end{itemize}

\end{proof}
Now, correctness of \cref{thm:enc} follows from \cref{lem:enc}.

\subsection{Proof of \cref{thm:soundness-axioms}}
\begin{proof}
  Soundness of the Bernoulli and Uniform axioms is straightforward. The Laplace
  axiom is \citet[Lemma 5]{DBLP:conf/icalp/BartheGGHS16}. The exponential axiom
  follows from the Laplace axiom, noting that
  \[
    \stt(\olap(\var_1, \var_2))(z) \leq 2 \cdot \stt(\lap(\var_1, \var_2))(z)
  \]
  for all $z > \stt(v_1)$, so the failure probability for the exponential axiom
  is at most twice the failure probability for the Laplace axiom.
\end{proof}

\subsection{Proof of \cref{thm:interpolant}}

\begin{proof}
  Notice that by construction
  we have $\labelp(\qentry) \triangleq \true$ and $\labele(\qentry) = 0$.

  We first show that $\labelp(\qexit) \Rightarrow \post$.
  By construction of encoding:
  $\bigwedge_{i=1}^n \varphi_i \Rightarrow (\neg \halt_n \Rightarrow \post)$ is valid.
  By definition of sequence interpolants:
  $\psi_n \Rightarrow (\neg \halt_n \Rightarrow \post)$ is valid.
  Therefore $\psi_n[\halt_n \mapsto \false] \Rightarrow \post$ is valid.
  Since $\cost_n$ does not appear in $\post$,
  $\exists \cost_n \ldotp \psi_n[\halt_n \mapsto \false] \Rightarrow \post$ is valid.

  Second, we show that $\pre \Rightarrow \labele(\qexit) \leq \errp$.
  By construction,
  $\labele(\qexit) \triangleq f(\varsi)$, where $f(\varsi)$
  is the function that returns, for any valuation of $\varsi$,
  the largest value of $\omega_n$ that satisfies $\exists \vars \setminus \varsi \ldotp \exists \halt_n \ldotp \psi_n$.
  By definition of sequence interpolants:
  $\psi_n \Rightarrow \cost_n \leq \beta$.
  Since $\beta$ is over $\varsi$,
  we have
$\exists \vars \setminus \varsi \ldotp \exists \halt_n \ldotp \psi_n \Rightarrow \omega_n \leq \beta$ is valid.
Pick a model $\model$ of $\pre$ with the largest
possible $\omega_n$ interpretation that satisfies $\exists \vars \setminus \varsi \ldotp \exists \halt_n \ldotp \psi_n$.
By construction of the encoding this model exists,
since any model satisfying $\pre$ can be extended to a model of $\bigwedge_i \varphi_i$.
It follows that this model satisfies $\omega_n \leq \beta$.

Finally, we need to show that for every edge $q_i \lto{\stmt} q_j$,
where $j = i+1$,
we have
\[
  \hoare{\wpprob(\labele(q_j), \stmt) - \labele(q_i)}{\labelp(q_i)}{\stmt}{\labelp(q_j)}
\]
We break the proof by statement type:
\begin{itemize}
  \item \emph{Assignment:}
    From definition of seq. interpolants, we know the following is valid
    \[
    \psi_i \land \var = \expr \land \cost_j = \cost_i \land \halt_j = \halt_i
    \Rightarrow \psi_j
    \]
    Set $\halt_i$ to $\false$ on left-hand side of implication.
    The following is valid:
    \[
    \psi_i[\halt_i \mapsto \false] \land \var = \expr \land \cost_j = \cost_i \land \halt_j = \false
    \Rightarrow \psi_j
    \]
    It follows that we can set $h_j$ to $\false$ on
    both sides, resulting in the following valid statement:
    \[
    \psi_i[\halt_i \mapsto \false] \land \var = \expr \land \cost_j = \cost_i
    \Rightarrow \psi_j[\halt_j \mapsto \false]
    \]

    Weaken rhs by existentially quantifying $\cost_j$.
    The following is valid:
    \[
    \psi_i[\halt_i \mapsto \false] \land \var = \expr \land \cost_j = \cost_i
    \Rightarrow \exists \cost_j \ldotp \psi_j[\halt_j \mapsto \false]
    \]
    Since $\cost_i$ is, by encoding, a function of $\varsi$, we can project it out on the lhs. The following is valid:
    \[
    (\exists \cost_i \ldotp \psi_i[\halt_i \mapsto \false]) \land \var = \expr \land \cost_j = \cost_i
    \Rightarrow \exists \cost_j \ldotp \psi_j[\halt_j \mapsto \false]
    \]
    As a result, we can drop the $\cost_j = \cost_i$
    constraint, resulting in the following valid statement:
    \[
    (\exists \cost_i \ldotp \psi_i[\halt_i \mapsto \false]) \land \var = \expr
    \Rightarrow \exists \cost_j \ldotp \psi_j[\halt_j \mapsto \false]
    \]
    This implies that the following Hoare triple, since
    $\labelp(q_i) \equiv \exists \cost_i \ldotp \psi_i[\halt_i \mapsto \false]$
    and $\labelp(q_j) \equiv \exists \cost_j \ldotp \psi_j[\halt_j \mapsto \false]$:
    \[\hoare{c}{\labelp(q_i)}{\stmt}{\labelp(q_j)}\]
    for any $c \in [0,1]$.

    It now remains to show that
    $\wpprob(\labele(q_j), \stmt) - \labele(q_i) \geq 0$,
    for any state $\stt$ in $\labelp(q_i)$.
    % Suppose, towards a contradiction, that
    % for some $\stt' \in \labelp(q_i)$, the above is not true.
    % $\cost_i$ are, by construction, functions of input variables.
    From our constraint, for any values of $\cost_i$ and $\varsd$
    that satisfy
    $\exists \vars \setminus \varsd \ldotp \exists \halt_i \ldotp \psi_i$,
    the same values where $\cost_j = \cost_i$ also satisfy
    $\exists \vars \setminus \varsd \ldotp \exists \halt_j \ldotp \psi_j$.
    Therefore, it is always the case that
    $\wpprob(\labele(q_j), \stmt) - \labele(q_i) \geq 0$
    % By construction,
    % we have
    % \[
    % (\exists \vars \setminus \varsd \ldotp \exists \halt_i \ldotp \psi_i)
    % \land \var = \expr \land \cost_j = \cost_i
    % \Rightarrow
    % (\exists \vars \setminus \varsd \ldotp \exists \halt_j \ldotp \psi_j)
    % \]

    \item \emph{Sample:}
    Following a similar simplification path to the one we used
    for assignment statements,
    we arrive at the following valid statement:
    \[
    (\exists \cost_i \ldotp \psi_i[\halt_i \mapsto \false]) \land \neg \axasm
    \Rightarrow \exists \cost_j \ldotp \psi_j[\halt_j \mapsto \false]
    \]
    Since we know that $\pr[\axasm] \leq \axub$, from the applied axiom,
    this implies that the following Hoare triple, since
    $\labelp(q_i) \equiv \exists \cost_i \ldotp \psi_i[\halt_i \mapsto \false]$
    and $\labelp(q_j) \equiv \exists \cost_j \ldotp \psi_j[\halt_j \mapsto \false]$:
    \[\hoare{\axub}{\labelp(q_i)}{\stmt}{\labelp(q_j)}\]

    It now remains to show that
    $\wpprob(\labele(q_j), \stmt) - \labele(q_i) \geq \axub$,
    for any state $\stt$ in $\labelp(q_i)$.
    % By construction,
    % we have
    % \[
    % (\exists \vars \setminus \varsd \ldotp \exists \halt_i \ldotp \psi_i)
    % \land \cost_j = \cost_i + \axub
    % \Rightarrow
    % (\exists \vars \setminus \varsd \ldotp \exists \halt_j \ldotp \psi_j)
    % \]
    % Therefore,
    % $\labelp(q_i) \land \cost_j = \labele(q_j) - \labele(q_i) \land \halt_j = \axasm \Rightarrow (\neg  \halt_j \Rightarrow \labelp(q_j))\land  \cost_j \leq \labele(q_j) $ is valid.

    Following argument from base case of \cref{lem:enc},
    we establish the specification.
    From our constraint, for any values of $\cost_i$ and $\varsd$
    that satisfy
    $\exists \vars \setminus \varsd \ldotp \exists \halt_i \ldotp \psi_i$,
    the same values where $\cost_j = \cost_i + \axub$ also satisfy
    $\exists \vars \setminus \varsd \ldotp \exists \halt_j \ldotp \psi_j$.
    Therefore, it is always the case that
    $\wpprob(\labele(q_j), \stmt) - \labele(q_i) \geq \axub$

    \item \emph{Assume:} Similar to sampling statements.

\end{itemize}

\end{proof}

\input{simpencoding}

\section{Capturing the Union Bound Logic} \label{app:ahl}

\begin{figure*}
  \begin{mathpar}
    \inferrule*[Right=Assn]
    { ~ }
    { \hoare{0}{ \Phi[\var \mapsto \expr] }{ \Assg{\var}{\expr} }{ \Phi } }
    \and
    \inferrule*[Right=Rand]
    { \forall \store.\, \store(\Phi) \implies \Pr_{\sem{\Rand{\var}{\dexpr}}(\store)}(\neg \Psi) \leq \store(\errp) }
    { \hoare{\errp}{ \Phi }{ \Rand{\var}{\dexpr} }{ \Psi } }
    \and
    \inferrule*[Right=Seq]
    { \hoare{\errp}{ \Phi }{ \prog }{ \Psi } \\ \hoare{\errp'}{ \Psi }{ \prog' }{ \Theta } }
    { \hoare{\errp + \errp'}{ \Phi }{ \Seqn{\prog}{\prog'} }{ \Theta } }
    \and
    \inferrule*[Right=If]
    { \hoare{\errp}{ \Phi \land \bexpr }{ \prog }{ \Psi } \\
    \hoare{\errp}{ \Phi \land \neg \bexpr }{ \prog' }{ \Psi } }
    { \hoare{\errp}{ \Phi }{ \Cond{\bexpr}{\prog}{\prog'} }{ \Psi } }
    \and
    \inferrule*[Right=While]
    { \forall \store, k.\, \store(\Phi \land \bexpr \land \expr_v = k)
      \implies \Pr_{\sem{\prog}(\store)}( \expr_v \geq k ) = 0 \\\\
      \expr_v : \mathbb{N} \\
      \models \Phi \land \expr_v \leq 0 \implies \neg \bexpr \\
      \hoare{\errp}{ \Phi \land \bexpr }{ \prog }{ \Phi } }
    { \hoare{\rho \cdot \errp}{ \Phi \land e_v \leq \rho }{ \Whil{\bexpr}{\prog} }{ \Phi \land \neg \bexpr } }
    \and
    \inferrule*[Right=Weak]
    { \models (\Phi' \implies \Phi) \land (\Psi \implies \Psi') \land (\errp \leq \errp')  \\
      \hoare{\errp}{ \Phi }{ \prog }{ \Psi } }
    { \hoare{\errp'}{ \Phi' }{ \prog }{ \Psi' } }
    \and
    \inferrule*[Right=Frame]
    { MV(\prog) \cap FV(\Phi) = \emptyset }
    { \hoare{0}{ \Phi }{ \prog }{ \Phi } }
    \and
    \inferrule*[Right=And]
    { \hoare{\errp}{ \Phi }{ \prog }{ \Psi } \\
    \hoare{\errp'}{ \Phi }{ \prog }{ \Psi' } }
    { \hoare{\errp + \errp'}{ \Phi }{ \prog }{ \Psi \land \Psi' } }
    \and
    \inferrule*[Right=Or]
    { \hoare{\errp}{ \Phi }{ \prog }{ \Psi } \\
    \hoare{\errp}{ \Phi' }{ \prog }{ \Psi } }
    { \hoare{\errp}{ \Phi \lor \Phi' }{ \prog }{ \Psi } }
    \and
    \inferrule*[Right=False]
    { ~ }
    { \hoare{1}{ \Phi }{ \prog }{ \bot } }
  \end{mathpar}
  \caption{The union bound logic, core rules \citep{DBLP:conf/icalp/BartheGGHS16} \label{fig:ahl}}
\end{figure*}

Our trace abstraction technique is inspired by the union bound
logic (\abr{aHL}), proposed by \citet{DBLP:conf/icalp/BartheGGHS16}. The core rules of
this program logic are presented in \cref{fig:ahl}; the only omitted rules are
the ones for the skip command (trivial to add to our language) and the rules for
external and internal procedure calls (we do not consider interprocedural
analysis). We comment briefly on a few rules; the others are largely standard.
The sampling rule \textsc{[Rand]} encodes distribution axioms. The most
complicated rule is \textsc{[While]}---intuitively, the side-conditions ensure
that there is a non-increasing integer variant $\expr_v$ whose initial value
bounds the maximum number of loop iterations. The program logic also features an
interesting complement of structural rules. Along with the usual rule of
consequence \textsc{[Weak]} and rule of constancy \textsc{[Frame]}, the
disjunction rule \textsc{[Or]} combines two pre-conditions (keeping the failure
probability unchanged) and the conjunction rule \textsc{[And]} combines two
post-conditions, while summing failure probabilities. Finally, the rule
\textsc{[False]} states that a judgment with failure probability at most $1$ can
prove any post-condition.

A minor but important difference between the setting in \abr{aHL} and our setting is
in the treatment of the failure probability expression $\errp$. In \abr{aHL}, these
expressions range over some fixed set of \emph{logical variables}, which appear
only in assertions and not in programs.  In our setup, we would model these
variables as \emph{input variables} $\varsi$, which may appear in programs by
cannot be modified. We will assume that input variables $\varsi$ correspond
precisely to the logical variables in \abr{aHL}.

We will show that our proof technique is complete with respect to the logic \abr{aHL},
subject to two restrictions on \abr{aHL} proofs:
\begin{enumerate}
  \item The rule \textsc{[While]} is applied to ``for''-loops.\footnote{%
      This can be slightly generalized to loops with a deterministic
    variant $e_v$, but we make this restriction to simplify proofs.}
  \item The rule \textsc{[Or]} is not used.
\end{enumerate}
Both of these restrictions stem from how our approach keeps track of the failure
probability. Roughly speaking, the original \abr{aHL} can analyze loops where the
guard is probabilistic but there is a deterministic bound on the number of
iterations. Since our failure probabilities must be deterministic along the
trace, we cannot directly handle such loops. However, these programs still have
a deterministic bound on the number of iterations and so they can be directly
transformed to be of the following form:
\[
  \Whil{\bexpr}{\prog} \equiv \Assg{i}{0} ; \Whil{i < \rho}{\Assg{i}{i + 1} ; \Cond{\bexpr}{\prog}{\Skip}}
\]
The situation with the rule \textsc{[Or]} is similar. If we have two
well-labeled automata modeling the two proofs in the premise, we would like to
combine them into a single automata but this is not possible---the labels on the
edges would need to be of the form $\Assm{\Phi}$ or $\Assm{\Phi'}$, but these
guards to not appear in the program $\prog$.  While it does not appear possible
to eliminate the \textsc{[Or]} rule, in our experience this rule is quite rarely
used. The rule can also be avoided entirely by applying a program transformation
to mark the logical cases:
\[
  \prog \equiv \Cond{\Phi}{\prog}{\prog}
\]
and then applying the standard conditional rule \textsc{[If]}.

We will prove completeness in two steps. First, we will show that for any
derivable judgment in \abr{aHL}, there exists a well-labeled automata modeling the
judgment (i.e., satisfying the conditions of \cref{thm:proofrule} for the given
pre-condition, post-condition, failure probability, and program). Then, we show
that well-labeled automata derived from programs can be found by a run of our
algorithm, given some labeling oracle $\labeltrace$.

Before we begin, we fix an automata representation of imperative programs once
and for all. Each automaton will have one entry node and one exit node. The rest
of the nodes, edge labels, and transition structure will be constructed
inductively given a program $\prog$.
\begin{itemize}
  \item Basic statements $\stmt \in \stmts$. Automaton with single edge from
    entry to exit node labeled by $\stmt$.
  \item Sequential composition $\Seqn{\prog}{\prog'}$. Identify the exit node
    for the automaton from $\prog$ with the entry node for the automaton from
    $\prog'$.
  \item Conditionals $\Cond{\bexpr}{\prog}{\prog'}$. Make new entry node, add
    directed edges labeled by $\Assm{\bexpr}$ and $\Assm{\neg \bexpr}$ to the
    entry nodes of automata from $\prog$ and $\prog'$ respectively, and then
    identify the exit nodes of the two automata.
  \item Loops $\Whil{\bexpr}{\prog}$. Make new entry node with an
    $\Assm{\bexpr}$ edge to the entry node of the automaton for $\prog$, and an
    $\Assm{\neg \bexpr}$ edge to a new exit node. From the exit node of $\prog$,
    add an edge back to the new entry node labeled $\Assm{\bexpr}$ and an edge
    to the new exit node labeled $\Assm{\neg \bexpr}$.
\end{itemize}
We call such automata derived from programs \emph{well-structured}.

\begin{theorem}[Completeness of well-labeled automata] \label{thm:complete-label}
  Let $\hoare{\beta}{\Phi}{\prog}{\Psi}$ be derivable in the fragment of \abr{aHL}
  indicated above. Then, there exists a well-structured and well-labeled
  automaton $\fsm$ satisfying the conditions of \cref{thm:proofrule} for this
  accuracy specification.
\end{theorem}
\begin{proof}
  Let $\fsm$ be the well-structured automaton corresponding to $\prog$. We will
  show that the nodes of $\fsm$ can each be labeled by a predicate and a failure
  probability expression, such that the entire automaton is well-labeled and
  satisfies the conditions of \cref{thm:proofrule}.  By induction on the proof
  derivation.
  \begin{description}
    \item[\textsc{[Assn]}]
      Label the entry and exit nodes by the pre- and post-condition
      respectively, with failure probability $0$.
    \item[\textsc{[Rand]}]
      Label the entry and exit nodes by the pre- and post-condition
      respectively, with failure probability $0$ and $\beta$.
    \item[\textsc{[Seq]}]
      Take the well-labelings for $\prog$ and $\prog'$ by induction. Label the
      node at the join point with invariant $\Psi$. For each node in the
      $\prog'$ automaton, add $\beta$ to the failure probability label.
    \item[\textsc{[If]}]
      Take the well-labelings for $\prog$ and $\prog'$ by induction. We may
      label the entry nodes $\Phi \land \bexpr$ and $\Phi \land \neg \bexpr$
      while preserving the well-labeling. Label the new entry node by $\Phi$
      with failure probability $0$, and the new exit node by $\Psi$ with failure
      probability $\beta$.
    \item[\textsc{[While]}]
      Let $\rho$ be the loop upper bound and let $i$ be the loop counter.  Take
      the well-labeling of the body $\prog$ by induction. By assumption on the
      structure of the while loop, there is a single transition from the body
      entry node $q_0$ to another node $q_1$, and it is labeled by $\Assg{i}{i +
      1}$. Furthermore, $q_0$ and $q_1$ are both labeled with failure
      probability $0$. Add the deterministic expression $(i - 1) \cdot \beta$ to
      all failure probability labels except at node $q_0$, and set the failure
      probability of $q_0$ to be $i \cdot \beta$. Label the new initial node by
      $\Phi$ and failure probability $0$, and the new exit node by $\Phi \land
      \neg \bexpr$ and failure probability $\rho \cdot \beta$.
    \item[\textsc{[Weak]}]
      Take the well-labeled automaton by induction.
    \item[\textsc{[Frame]}]
      Label all nodes by $\Phi$ and failure probability $0$.
    \item[\textsc{[And]}]
      Take the two well-labelings $(\labele_1, \labelp_1)$ and $(\labele_2,
      \labelp_2)$ by induction. By assumption, both  of these well-labeled
      automata have the same structure (given by the well-structured automaton
      corresponding to $\prog$). Set the new labeling functions to be $\labele =
      \labele_1 \land \labele_2$, and $\labelp = \labelp_1 + \labelp_2$.
    \item[\textsc{[Or]}]
      Not allowed.
    \item[\textsc{[False]}]
      Label the entry node by $\Phi$ and failure probability $0$. Label all
      other nodes by $\bot$ and failure probability $1$.
  \end{description}
\end{proof}

\begin{theorem}[Completeness of algorithm] \label{thm:complete-algo}
  Let $\fsm$ be a well-structured and well-labeled automaton. Then, there exists
  a run of our algorithm in \cref{alg:overall} given some labeling oracle
  $\labeltrace_\fsm$ that produces $\fsm$ along its execution.
\end{theorem}
\begin{proof}
  We provide a sketch of the proof. First, our algorithm can recover any
  loop-free well-labeled automaton (possibly not well-structured). In a bit more
  detail, let $\lang(\fsm)$ be the set of all paths from entry to exit node;
  note that this set is finite for loop-free automata. By repeatedly applying
  $\ctrace$, our algorithm can label each of these traces using the
  well-labeling in $\fsm$, yielding a set of well-labeled traces. Then by
  repeatedly applying $\cmerge$, our algorithm can merge all traces and recover
  the automaton $\fsm$.

  Now, suppose that $\fsm$ is well-structured but not loop-free. We can convert
  $\fsm$ to a loop-free automaton $\fsm_{lf}$ by simply deleting each back edge
  from the exit node of each while loop back to its corresponding entry node;
  dropping edges evidently keeps the automaton well-labeled. By the previous
  argument, our algorithm can generate $\fsm_{lf}$ by repeatedly applying
  $\ctrace$ and $\cmerge$. Then, we can apply $\cgen$ repeatedly to add the
  deleted edges, noting that there are at most finitely many such edges since
  the originally program has finitely many loops. These new edges preserve
  well-labeling and recover $\fsm$.
\end{proof}

As an immediate corollary, we have the following completeness result.

\begin{corollary}
  Let $\hoare{\beta}{\Phi}{\prog}{\Psi}$ be derivable in the fragment of \abr{aHL}
  indicated above.  Then, there exists a run of our algorithm in
  \cref{alg:overall} given some labeling oracle $\labeltrace_\fsm$ terminating
  successfully with rule $\ccorr$.
\end{corollary}
\begin{proof}
  By \cref{thm:complete-label}, there exists a well-labeled automaton $A$
  proving the specification. By \cref{thm:complete-algo}, there is a run of the
  algorithm that constructs this automaton. At that point in the execution, rule
  $\ccorr$ applies and the algorithm succeeds.
\end{proof}

%% file: simpencoding.tex
%!TEX root=paper.tex
\section{A Simplified Encoding} \label{app:enc}
The encoding in \cref{sec:interpolation} is designed for full generality:
it assumes that a trace may be infeasible, which is why it introduces the auxiliary variables $\halt_i$ to track states that cannot make it through the trace.
In the case where the trace is feasible for some input states,
the encoding and interpolation problems become much simpler by doing away with the auxiliary $\halt_i$ variables.
The simplified version of $\encodeD$ is shown in \cref{fig:encs}.

Henceforth we assume that for a trace $\trace$,
all Boolean expressions appearing in assume statements are over $\varsd$.
Second, we assume that there is a state $\stt$
such that $\trace(\stt)$ is a distribution.

\begin{figure}[t]
  \begin{align*}
    \encodeD(i, \var \gets \expr) &\triangleq \var = \expr \land \cost_i = \cost_{i-1} \\
    \encodeD(i, \assume(\bexpr)) &\triangleq \bexpr \land \cost_i = \cost_{i-1}  \\
    \encodeD(i, \var \sim \dexpr) &\triangleq \neg \axasm \land \cost_i = \cost_{i-1} + \axub
    \qquad \text{given axiom family: } \pr_{\var \sim \dexpr} [\axasm] \leq \axub
  \end{align*}
  \caption{Simplified logical encoding of statement semantics for feasible traces}
  \label{fig:encs}
\end{figure}

\begin{theorem}[Soundness of simplified encoding]\label{thm:simpenc}
  The specification
  $\hoare{\errp}{\pre}{\stmt_1,\ldots,\stmt_n}{\post}$
  is valid if the following formula is satisfiable:
  \begin{align}\label{eq:enc}
  \forall \vars, \cost_i \ldotp \left(\pre \land \cost_0 = 0 \land \bigwedge_{i=1}^n \encodeD(i,\stmt_i)\right) \Longrightarrow (\cost_{n} \leq \beta \land \post)
  \end{align}
\end{theorem}

\cref{thm:simpenc} follows from the next lemma:

\begin{lemma}[Soundness of simplified $\encodeD$]\label{lem:simpenc}
Fix trace $\trace = \stmt_1;\cdots;\stmt_n$.
Let $\varphi \triangleq \cost_0 = 0  \land \bigwedge_{i=1}^n \encodeD(i,\stmt_i)$,
where all uninterpreted functions resulting from distribution axiom families have been given a fixed interpretation.
Fix a state $\stt \in \stts$.
Let $\model_1,\ldots,\model_m$ be the set of models of $\varphi$
such that $\stt(\varsi) = \model_i(\varsi)$, for all $i\in [1,m]$.
Let \[R_\stt = \{\stt' \mid  \stt'(\vars) = \model_i (\vars) \}\]
Then, for any $\model_i$, we have $\sem{\trace}(\stt)(\overline{R_\stt}) \leq \model_i(\cost_n)$.
\end{lemma}

\begin{proof}
  First, we note that  all models $\model_i \models \varphi$
  agree on the value of $\omega_i$.
  This is because the constraints $\omega_i$ are functions of $\varsi$.
  Second, note that by our assumption, there is always $\model_i \models \varphi$---i.e., it is never unsatisfiable.

  We proceed by induction on the length of $\trace$.
  For $n = 1$, we have three cases. Fix an $\stt$ as in lemma statement.
  \begin{itemize}
  \item \emph{Case 1:} $\trace = \var \gets \expr$.

  We have $\varphi \triangleq \var = \expr \land \cost_1 = 0$.
  %Fix an $\model$, $R_\stt$, and $\stt$ (as in lemma statement).
  From $\varphi$, $\model_i(\cost_1) = 0$, for all $i$.
  Therefore, lemma states:
  $\sem{\var \gets \expr}(\stt)(\overline{R_\stt}) \leq 0$.
  Suppose this does not hold, then, by definition of $\sem{\var \gets \expr}$,
  there is a state $\stt' \in \stts \setminus R_\stt$ such that
  $\stt' = \stt[v \mapsto \stt(\expr)]$.
  However, by definition of $\varphi$, $\stt' \in R_\stt$,
  since $\model_i(\vars) = \stt'(\vars)$, for some $i$.

  \item \emph{Case 2:} $\trace = \assume(\bexpr)$.

  We have $\varphi \triangleq \bexpr \land \cost_1 = 0 $.
  From $\varphi$, $\model_i(\cost_1) = 0$, for all $i$.
  Therefore, lemma states:
    $\sem{\assume(\bexpr)}(\stt)(\overline{R_\stt}) \leq 0$.
    Suppose this does not hold, then, by definition of $\sem{\assume(\bexpr)}$,
    we have $\stt \in \stts \setminus R_\stt$ and $\stt(\bexpr) = \true$.
    However, by definition of $\varphi$, $\stt \in R_\stt$,
    iff $s(\bexpr) = \true$.

  \item \emph{Case 3:} $\trace = \var \sim \dexpr$.

  We have $\varphi \triangleq \neg \axasm \land \cost_1 = \axub$.
  Lemma states that
  $\sem{\var \sim \dexpr}(\stt)(\overline{R_\stt}) \leq \model_i(\axub)$,
  for all $i$.
  This follows from the definition of a distribution axiom: that
  $\pr_{\var \sim \dexpr}[\axasm] \leq \axub$ is true for any valuation of $\vars \setminus \{\var\}$.
  \end{itemize}

Assume that lemma holds for traces of length $n$.
We show that it also holds for $n+1$,
where $\trace'$ is a trace of length $n$.

\begin{itemize}
\item \emph{Case 1:} $\trace = \trace'; \var \gets \expr$.

The encoding is $\varphi \triangleq \varphi' \land \var = \expr \land \cost_{n+1} = \cost_{n} $.

Let $\model'_1, \ldots, \model_m'$ and $R_\stt'$ be defined for $\varphi'$ and $\trace'$, as per lemma statement.
By hypothesis, $\sem{\trace'}(s)(\overline{R_\stt'}) \leq M_i'(\omega_n)$.
By semantics of assignment, we have
$\sem{\trace'; \var \gets \expr}(s)(\overline{X}) \leq \model_i'(\omega_{n+1})$,
where $X = \{\stt \mid \stt' \in R_\stt', \stt = \stt'[\var \gets \stt'(\expr)]\}$.
Observe that $X = R_\stt$:
by definition of $\varphi$, its models are a subset of $\{\model_1',\ldots,\model_m'\}$ such that $\var = \expr$.
It then follows that $\sem{\trace}(\stt)(\overline{R_\stt}) \leq \model_i(\cost_{n+1})$, for all $i$.

% By hypothesis:
% $\sem{\trace}(\stt)(\overline{R_\stt}) \leq \model_i(\cost_n)$,
% where $\model$, $R_\stt$, and $\stt$ are from lemma statement for $\trace'$.

\item \emph{Case 2:} $\trace = \trace'; \assume(\bexpr)$.

The encoding is $\varphi \triangleq \varphi' \land \bexpr \land \cost_n = \cost_{n-1}$.
Let $\model'_1, \ldots, \model_m'$ and $R_\stt'$ be defined for $\varphi'$ and $\trace'$, as per lemma statement.
By hypothesis, $\sem{\trace'}(s)(\overline{R_\stt'}) \leq M_i'(\omega_n)$,
for all $i$.
We know that models $\{\model_i\}$ of $\varphi$ are a subset of $\{\model'_i\}$
where $\bexpr$ is $\true$.
Therefore $\overline{R_\stt} \supseteq \overline{R_\stt'}$.
But we have that all states in $\overline{R_\stt} \setminus \overline{R_\stt'}$
are those were $\bexpr$ is false.
By definition of $\sem{\assume(\bexpr)}$, all those states are assigned probability 0.
Therefore, $\sem{\trace}(\stt)(\overline{R_\stt}) \leq \model_i(\cost_{n+1})$, for all $i$.

\item \emph{Case 3:} $\trace = \trace'; \var \sim \dexpr$.

The encoding is $\varphi \triangleq \varphi' \land \neg \axasm \land \cost_{n+1} = \cost_n + \axub$.
Let $\model'_1, \ldots, \model_m'$ and $R_\stt'$ be defined for $\varphi'$ and $\trace'$, as per lemma statement.
By hypothesis, $\sem{\trace'}(s)(\overline{R_\stt'}) \leq M_i'(\omega_n)$,
for all $i$.
Let $X$ be the set of all states that satisfy $\neg \axasm$.
From $\varphi$, we know that $R_\stt = R_\stt' \cap X$.
By the union bound and the distribution axiom,
$\sem{\trace'; \var \sim \dexpr}(\stt)(\overline{R_\stt'} \cup \overline{X}) \leq M_i(\omega_{n+1})$, for all $i$.

% We know that models $\{\model_i\}_i$ of $\varphi$ are a subset of $\{\model'_i\}_i$ such that $\neg \axasm$ is true.

\end{itemize}
\end{proof}

Assume we construct a sequence of interpolants for the above
encoding as described in \cref{sec:interpolation}.
Then, the following theorem holds, which is the same as \cref{thm:interpolant}, but without handling $\halt_i$ variables.

\begin{theorem}[Well-labelings from interpolants]\label{thm:simpinterpolant}
  Let $\{\psi_i\}_i$ be the interpolants computed as shown above.
  Let $\fsm_\trace = \tuple{Q,\delta,\labelp,\labele}$ be the failure automaton that accepts only the trace
  $\trace = \stmt_1,\ldots,\stmt_n$,
  i.e., $\delta = \{\qentry\lto{\stmt_1} q_1, q_1 \lto{\stmt_2} q_2, \ldots q_{n-1} \lto{\stmt_n} \qexit\}$.
  Set the labeling functions as follows:
  \begin{enumerate}
    \item $\labelp(\qentry) \triangleq \pre$ and $\labele(\qentry) \triangleq 0$.
    \item $\labelp(q_i) \triangleq \exists \cost_i \ldotp \psi_i$ and $\labelp(\qexit) \triangleq \exists \cost_n \ldotp \psi_n$.
    \item $\labele(q_i) \triangleq f(\varsd)$, where $f(\varsd)$
    is the function that returns, for any valuation of $\varsd$,
    the largest value of $\omega_i$ that satisfies $\exists \vars \setminus \varsd  \ldotp \psi_i$.
    For $\labele(\qexit)$, we use $\exists \vars \setminus \varsi \ldotp \psi_n$.
  \end{enumerate}
  Then, $\fsm_\trace$ is well-labeled and implies $\hoare{\errp}{\pre}{\trace}{\post}$.
\end{theorem}

\begin{proof}
  Similar to \cref{thm:interpolant}.
\end{proof}

%% file: app_implementation.tex
\section{Implementation Details}\label{app:implementation}
This section expands on \cref{sec:evaluation} by providing additional implementation details and examples.

\paragraph{Algorithmic Strategy.}
Our implementation is a determinization of the algorithm presented in \cref{sec:algorithm}.
To ensure that we prove the given specifications by computing tight upper bounds on failure probability, our implementation aggressively tries to apply the $\cmerge$ rule---recall that the $\cmerge$ rule allows us take the maximum failure probability across two automata, instead of the sum.
Specifically, we modify the rule $\ctrace$ to return a set of traces $\trace_1, \ldots, \trace_n \in \lang(\prog) \cap \overline{\lang(\fsms)}$.
Then, we attempt to simultaneously label all traces with the same interpolants at nodes pertaining to the same control location.
To ensure that we compute similar interpolants across traces, we attempt to use the same distribution axiom for the same sampling instruction in all traces it appears in.
Finally, we apply the rule $\cgen$ to attempt to create cycles into the resulting automaton.

The  pseudocode in \cref{fig:implementation}
shows our determinization of the algorithm from \cref{sec:algorithm}.
The loop at \cref{line:ax} goes through axioms as described below, proposing one axiom in every iteration and checking it.
Notice that for every occurrence of a sampling statement, across all traces $\trace_j$, it attempts the same axiom---this is used to force a successful $\cmerge$.
\cref{line:interp} computes interpolants for every trace $\trace_j$'s encoding $\Psi_j$.
This procedure also tries to find the same interpolants for the same control locations---this ensures success of $\cmerge$ and $\cgen$.
In all case studies in \cref{sec:evaluation},
the algorithm succeeds by considering all traces that execute 0 or 1 iterations of every loop.

\begin{figure}
\begin{algorithmic}[1]
    \State $\fsms \gets \emptyset$
    \State $i \gets 1$ \Comment counter
    \While {$\ccorr$ does not apply}
      \State {\color{blue}The following lines implement $\ctrace$
      for a set of traces}
      \State Get all paths $\trace_1, \ldots, \trace_n \in \lang(\prog) \setminus \lang(\fsms)$ that go through each loop at most $i$ times.
      \State For every $\trace_j$, let $\Psi_j$ be the encoding in \cref{thm:enc},
      where different occurrences of the same sampling statement use the same parameter $f(\varsi)$ in their distribution axiom.
      \State $\donev \gets \false$
      \While {not $\donev$}\label{line:ax}
        \State pick an interpretation $M$ for every $f(\varsi)$ in $\{\Psi_j\}_j$
        \If {$M \models \bigwedge_j \Psi_j$}
          \State $\donev \gets \true$
          \State $axioms \gets M$
        \EndIf
      \EndWhile
      \State Compute interpolants for every $\Psi_j$ where $f(\varsi)$ are instantiated by $axioms$ and create well-labeled automata $\{\fsm_j\}_j$
      \label{line:interp}
      \State Add $\{\fsm_j\}_j$ to $\fsms$
      \State  {\color{blue}The following repeatedly applies \cmerge}
      \State Apply $\cmerge$ to every pair of automata in $\fsms$
      until it does not apply any more
      \State {\color{blue} The following loop repeatedly applies $\cgen$}
      \For {every $A_i \in \fsms$}
        \For {all $q,q'$ in  $\fsm_i$
        s.t. $q,q'$ denote the same loop head in $\prog$}
        \State Apply $\cgen$ to $q,q'$ with $\stmt \in \stmts$ being the loop exit condition
        \EndFor
      \EndFor
      \State $i \gets i + 1$
    \EndWhile
\end{algorithmic}
\caption{Implementation of nondeterministic algorithm in \cref{alg:overall}}
\label{fig:implementation}
\end{figure}

\paragraph{Discovering Axioms.}
Given a formula of the form $\exists f \ldotp \forall X \ldotp \varphi$,
we check its validity using a propose-and-check loop:
\rone we propose an interpretation of $f$ and then \rtwo check if $\forall X \ldotp \varphi$ is valid with that interpretation using the \abr{SMT} solver (more on this below).
The first step proposes interpretations of $f$ of increasing size,
e.g., for a unary function $f(x)$, it would try $0,1,x,x+1$, etc.
% As we shall see, even for complex randomized algorithms from the literature, the required axioms are syntactically simple, so this simple strategy works rather well.

Note that this enumerative approach will encounter many axiom parameters
that are not well-typed or do not satisfy the conditions
required for the parameters.
For example, for the Laplace axiom family, we have $f(\varsi) \in (0,1]$.
Therefore, any instantiation that may be real-valued and $\leq 0$ or $>1$ is rejected.

\paragraph{Checking Validity.}
The case studies we consider make heavy use of non-linear arithmetic (e.g., $\frac{x \cdot y}{z} + c > 0$) and transcendental functions (namely, $\log$).
Non-linear theories are generally undecidable.
To work around this fact, we implement an incomplete formula validity checker using an eager version of the \emph{theorem enumeration} technique recently proposed by~\citet{Srikanth2017}.
First, we treat non-linear operations as uninterpreted functions, thus overapproximating their semantics.
Second, we strengthen formulas by instantiating \emph{theorems} about those non-linear operations. For instance, the following theorem relates division and multiplication: $\forall x,y \ldotp y > 0 \Rightarrow \frac{x \cdot y}{y} = x$.
We then instantiate $x$ and $y$ with terms over variables in the formula. Of course, there are infinitely many possible instantiations of $x$ and $y$;
we thus restrict instantiations to terms of size 1, i.e., variables or constants.

Our implementation uses a fixed set of theorems about multiplication, division, and log. These are instantiated for every given formula, typically resulting in $\sim$1000 additional conjuncts.
To give an intuition, we list some of those theorems below:
\begin{itemize}
  \item $\forall x,y \ldotp y > 0 \Rightarrow \frac{x \cdot y}{y} = x$
  \item $\forall x,y,z \ldotp z > 0 \Rightarrow \frac{x \cdot y}{z} + \frac{x}{z} = \frac{x \cdot (y+1)}{z}$
  \item $\forall x,y \ldotp x \geq 0 \land y \geq 0 \Rightarrow x \cdot y \geq 0$
  \item $\forall x,y \ldotp x \geq 0 \land y > 0 \Rightarrow \frac{x}{y} \geq 0$
\end{itemize}
In all of the differentially private algorithms, we can prove correctness by treating $\log$  completely as uninterpreted, requiring no $\log$ specific theorems, just the fact that, e.g., $\log(x) + \log(x) = 2\log(x)$.

\paragraph{Interpolation Technique.}
Given the richness of the theories we use, we found that existing proof-based interpolation techniques either do not support the theories (e.g., the MathSAT solver) or fail to find generalizable interpolants, e.g., cannot discover quantified interpolants (e.g., Z3).
As such, we implemented a \emph{template-guided} interpolation technique~\citep{albarghouthi2013beautiful,rummer2013exploring},
where we force interpolants to follow syntactic forms that appear in the program.
Specifically, for every Boolean predicate $\varphi$ appearing in the program, the specification, or the axioms, we create a template $\varphi^t$, which is $\varphi$ but with variables replaced by placeholders, denoted $\wild_i$. For instance, given $x > y$, we generate the template $\wild_1 > \wild_2$.

Since the failure probabilities, encoded in variables $\cost_i$
increase additively by accumulating $\axub$ expressions from the distribution axioms, we use the template $$\cost_i \leq \sum_{j=1}^{n}\wild_j*\axub_j,$$
where $\axub_j$ is the failure probability of the axiom used in the $j$th sampling statement, assuming there are $n$ such statements along the path,
and $\wild_j$ can take terms of $\varsd$---following the restriction on labels.

Given a set of templates, our interpolation technique searches for an interpolant as a conjunction of instantiations of those templates, where each $\wild_i$ can be replaced by a well-typed term over formula variables.
Given the infinite set of possible instantiations, our implementation fixes the size of possible instantiations (e.g., to size 1), and proceeds by finding the smallest possible interpolants in terms of number of conjuncts.
If it cannot, it expands the search to terms of larger sizes.
We ensure that the special variables $\cost_i$ only appear in their set of inequality predicates defined above. Therefore, given an interpolant $I$, we can syntactically divide it into $I = I_\vars \land I_\cost$,
where $I_\vars$ is over program variables $\vars$ and $I_\cost \triangleq \cost_i \leq \ldots$ provides the upper bound on failure probabilities at that point along the trace.

\subsection{Proof of Report Noisy Max ($\rnm$)}\label{app:rnm}
We give an abridged form of the proof computed for Report Noisy Max
in \cref{fig:rnm}.
The set of queries $\qs$ is assumed to be non-empty,
and for simplicity, we let $\best$ be initialized to 1 instead of $\bot$
and modify the conditional to check if $\best = 1$---resulting in an equivalent program.
The bottom automaton shows a merge of the two paths through the conditional in the loop.
Notice that the propagated error probability is
\[
\frac{p \cdot (i-1)}{|Q|}
\]
This is because in each iteration, we apply the Laplace axiom
with
\[
f(\varsi) = \frac{p}{|Q|}
\]
After $k$ loop iterations, $i = k+1$, and therefore we have accumulated a failure probability of $\frac{p \cdot (i-1)}{|Q|}$.
(If the program were rewritten so as $i$ starts at 0 and the loop condition is $i < |Q|$, we would have the simpler failure expression $\frac{p \cdot i}{|Q|}$.)
Finally, we can infer that the total failure probability is $p$.
This is due to the state label (blue) $\varphi$.

The label $\varphi$ is defined as the conjunction of the following formulas, which we simplify for presentation:

  $$|Q| + 1\geq i \geq 1$$
  $$i \geq \best \geq 1$$
  $$i \neq 1 \Rightarrow \best < i$$
  $$\forall j \in [1,i) \ldotp |a_j - \qs_j(d)| \leq  \frac{2}{\epsilon}\log \frac{\len{\qs}}{p}$$
  $$\forall j \in [1,i) \ldotp a_\best \geq a_j$$

The first two conjuncts specify the range of values $i$ takes throughout the loop iterations.
The third conjunct specifies that $i$ leaps ahead of $\best$ after the first loop iteration, since $i$ is always incremented at the end of the loop, and $\best$ can at most be $i-1$ at that point. (The syntactic form of an implication is derived from the conditional's predicate.)
The fourth conjunct specifies that, for every element of $j$ of $a$,
its distance from the corresponding valuation of $\qs_j(\db)$ is bounded above by $\frac{2}{\epsilon}\log \frac{\len{\qs}}{p}$, which follows from the choice of the axiom.
Finally, the last conjunct states that the best element is indeed larger than all previously seen ones.

The last two conjuncts are primarily responsible for implying the postcondition (via the triangle inequality):
$$\forall j \in [1,\len{\qs}] \ldotp \qs_\best(\db) \geq \qs_j(\db) - \frac{4}{\epsilon}\log \frac{\len{\qs}}{p}$$
Notice that the
$\frac{2}{\epsilon}\log \frac{\len{\qs}}{p}$
in the fourth conjunct translates to $\frac{4}{\epsilon}\log \frac{\len{\qs}}{p}$
in the postcondition. This is due to the absolute value.

\begin{figure}[t]
\includegraphics[scale=1.15]{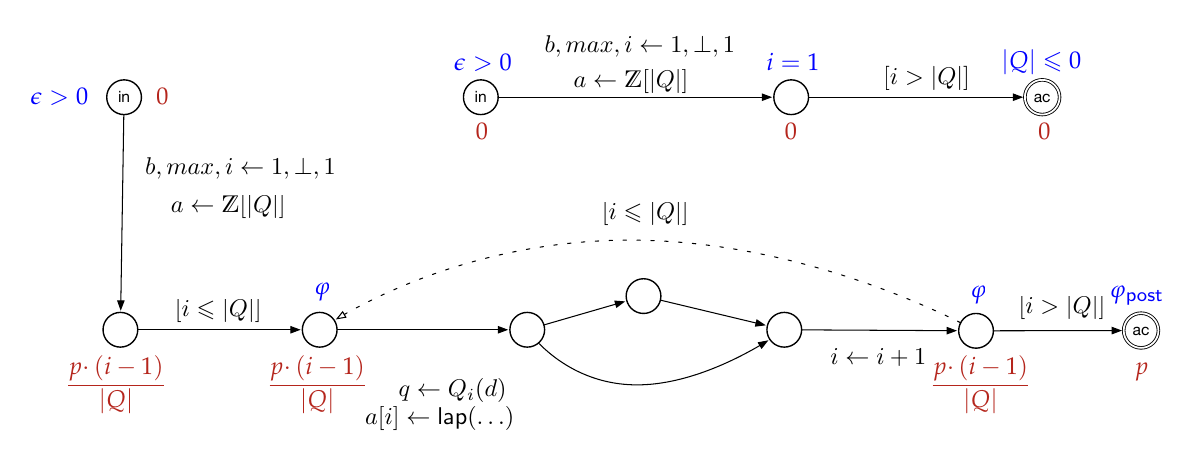}
% The loop terminates with a proof when \rone there are no more traces to pick and \rtwo the combined failure probabilities of $\fsms$ is at most $\beta$.
\caption{Main loop of verification algorithm}
\label{fig:rnm}
\end{figure}